\documentclass[10pt]{amsart}

\usepackage{mathrsfs}\usepackage{graphicx}
 \usepackage[arrow,matrix]{xy}
\usepackage{amsmath,amssymb,amscd,bbm,amsthm,mathrsfs}
\usepackage{color,xcolor}
\usepackage{graphicx}
\usepackage{manfnt}
\newtheorem{thm}{Theorem}[section]
\newtheorem{dfn}{Definition}[section]
\newtheorem{lem}{Lemma}[section]

\newtheorem{prop}{Proposition}[section]
\newtheorem{rem}{Remark}[section]
\newtheorem{ex}{Example}[section]

\theoremstyle{definition}

\newcommand\bZ{\mathbb{Z}}
\newcommand\bR{\mathbb{R}}\newcommand\bS{\mathbb{S}}
\newcommand\bC{\mathbb{C}}

\newcommand\bN{\mathbb{N}}
\newcommand\bH{\mathbb{H}}

\newcommand\cV{\mathcal{V}}
\newcommand\cD{\mathcal{D}}
\newcommand\cE{\mathcal{E}}
\newcommand\cF{\mathcal{F}}
\newcommand\cH{\mathcal{H}}
\newcommand\cI{\mathcal{I}}
\newcommand\cJ{\mathcal{J}}

\newcommand\cM{\mathcal{M}}
 \newcommand\cA{\mathcal{A}}

\newcommand\Is{\operatorname{Is}}

\newcommand\la{\lambda}

\newcommand\0{\emptyset}

\newcommand{\Pol}{\mathit{Pol}}

\newcommand\GG{\mathcal{G}}
\newcommand\PP{\mathcal{P}}

\begin{document}
\numberwithin{equation}{section}

 \title{Massless field equations for  spin $ 3/2$ in dimension $6$ }

\author{R.\ L\'{a}vi\v{c}ka,  V.\ Sou\v{c}ek} 

\address{Charles University, Faculty of Mathematics and Physics, 
	Sokolovsk\'a 83, 186 75 Praha, Czech Republic}

\email[R. L\'avi\v cka]{lavicka@karlin.mff.cuni.cz}

\email[V. Sou\v cek]{soucek@karlin.mff.cuni.cz}

\author{W.\ Wang}

\address{Department of Mathematics, Zhejiang University, Zhejiang 310027, PR China}

\email[W. Wang]{wwang@zju.edu.cn}

\thanks{Support of the grant GACR 24-10887S is gratefully acknowledged by R.\ L\'{a}vi\v{c}ka and  V.\ Sou\v{c}ek.
W.\ Wang is supported by National Nature Science Foundation in China (No. 11971425).}

\begin{abstract}
	Main topic of the paper is a study of properties of massless fields of spin $3/2$.
A lot of information is available already for massless fields in dimension $4.$ Here, we concentrate
on dimension $6$ and we are using the fact that the group $SL(4,\bC)$ is isomorphic with
the group $Spin(6,\bC).$ It makes it possible to use tensor formalism for massless fields. Main problems treated in the paper are  a description of fields which need
to be considered in the spin $3/2$ case, a suitable choice of equations they  should satisfy, 
irreducibility of homogeneous solutions of massless field equations, the Fischer decomposition
and the Howe duality for such fields.
\end{abstract}


\maketitle

\section{Introduction}
\label{introduction}
The classical series of massless field equations for spinor fields on  four dimensional Lorentzian manifold
include several basic systems of PDE's used and studied extensively in theoretical physics.
In physical dimension $4$ and for the Lorentzian signature, the irreducible spinor module for the Clifford
algebra decomposes into two irreducible rank-two self-dual spinor modules $S_+$ and $S_- $,
complex conjugate to each other. There are two basic versions of  massless field equations.  

(i) The helicity $-k/2$  spinor field is a field   $\varphi_{AB..E}\in \odot^k(S_-)$ totally symmetric   in its $k$ indices
satisfying the first order system of equations
\begin{equation}\label{negative}
\nabla^{AA'}\varphi_{AB..E}=0
\end{equation}
where $\nabla_{AA'}$ is the Levi-Civita covariant derivative, 
unprimed indices refer to elements of $S_-,$ and
we use the standard spinor notation.

(ii) The helicity $k/2$  spinor field is a field  $\varphi_{A'B'..E' }\in \odot^k(S_+)$ 
satisfying the first order system of equations
\begin{equation}\label{positive}
\nabla^{AA'}\varphi_{A'B'..E'}=0.
\end{equation}

Note that  the series also includes morally (the second order) wave equation for scalar (spin $0$) field.
R.~Penrose initiated a systematic study of solutions of  massless field equations (\ref{negative})
and (\ref{positive}) in dimension $4$ in his twistor 
program, see \cite{EPW,PM,PR} and also \cite{BBS,SS,V}. 
Twistor theory in dimension $6$ is developed in \cite{KW,KWZ,MRT,SaWo}.
 
The most interesting cases of massless field equations are the case $k=1$ (neutrino equation),
$k=2$ (vacuum Maxwell equations), $k=3$ (the Rarita-Schwinger equation) or $k=4$ (equations for linearized gravitation).
It is well-known that the classical massless field equations are conformally invariant.

There are several generalizations of  massless fields and their equations. 

(a) It is possible to consider
equations for spinor fields of mixed types, i.e., for fields $\varphi$ with values in the tensor product $\odot^k(S_-)\otimes \odot^\ell(S_+)$.

(b) Higher spin fields can be considered on (pseudo-)Riemannian manifolds of different signatures.
In particular, for positive (resp. negative) signature, the massless field equations are elliptic,
which has a dramatic effect on behaviour of their solutions. In theoretical physics, it is often formulated
as the transition to imaginary time. On side of mathematics, this is the topic first studied in a special branch
of mathematics usually called quaternionic analysis, see \cite{Sud}.  It was created as a close generalization of complex function theory
to  functions on $\bR^4$ with values in the field of quaternions but it was soon
recognized as a study of solutions of (the Euclidean version of) the Dirac equation for spin $1/2$
 fields. 

(c) It is possible to consider analogues of massless field equations on higher dimensional (pseudo)- Riemannian manifolds.  S\"amann and Wolf \cite{SaWo}  and Mason et al. \cite{MRT} used  the twistor method to study a kind of higher spin massless field operators on Lorentzian space $\mathbb R^{5,1}$ by embedding it to $ \wedge^{2}\mathbb C^{4}$, while   the function theory of their Euclidean version was investigated in \cite{KW,KWZ} by embedding $\mathbb R^{6}$ to $ \wedge^{2}\mathbb C^{4}$.
On Euclidean spaces of any dimension, Clifford analysis studies massless fields of spin $1/2$ (see \cite{DSS}) but also of higher spins, see e.g.\ \cite{BSSL,BSSL2,ERJ,ES,Soucek1,Soucek2}. 

(d) Fields of different signatures can be embedded into holomorphic fields on complex space times.
In such a way, it is posible to formulate some results in a uniform way for all possible signatures.

The presented paper provides  a special case of  the generalizations mentioned above. We consider
a generalization of massless field equations to the case of
holomorphic spin $3/2$ fields
on $6$-dimensional complex space $\bC^6.$ We show that we are forced to consider for presented
 results also equations of spinor fields of mixed types and fields of spin $1/2.$

There are several types of problems studied in the paper. We discuss mainly the following questions.
Let $G_0=Spin(6)$ and let $V$ be a complex representation of $G_0$, not necessarily irreducible.
We consider a set of $G_0$-invariant differential equations on the space
$\mathcal{V}=\operatorname{Pol}(\bC^6,V)$ of all polynomial fields on $\bC^6$ with values in $V$ and we study properties of its solutions.

\vskip 1mm\noindent
(A) {\bf A choice for values $V$ of the field.} 
Let us denote by $S_{\pm}$ two basic irreducible $G_0$-modules of spin  $1/2$.
Then the space $V$ must contain basic spin  $3/2$ fields of both helicities similarly as
in (\ref{negative}) and (\ref{positive}).
For negative helicities, the operator defining the massless field equations
maps fields with values in $\odot^3(S_-)$ to fields with values in $\odot^2(S_-)\boxtimes S_+,$
see Lemma 3.1 below. So to have a possibility to formulate field equations for such fields, it is necessary
to add the space $\odot^2(S_-)\boxtimes S_+$ to the space $V$ of values of the fields. Using further
results of Section 3, we find finally definition of the space $V$ for our problem.

\vskip 1mm\noindent
(B) {\bf The formulation of field equations for spinor fields of mixed types.}
There is just a unique choice for field equations  invariant under the action of $G_0$ for spin $3/2$  fields with one type of indices  (see Lemma 3.1 below).  But there are different possible choices
for invariant field equations for fields of mixed type.   Various possible choices  are well visible from
results contained in Section 3. Any $G_0$-invariant first order equation is given by the projection
of the gradient of the field to an invariant  subspace of the $G_0$-module $\bC^6 \otimes V$
(see \cite{SS}).
Our choice is a complexified version of the scheme suggested
in Euclidean case by E.M.\ Stein and G.\ Weiss in \cite{SW}. The choice of equations is different from
those treated in \cite{ERJ,ES}. 

\vskip 1mm\noindent
(C) {\bf The decomposition of the space of solutions of field equations.}
\vskip 1mm\noindent
The model case are spin $1/2$ fields. Let $S_- $ be an irreducible  spinor module
in even dimension $2n$. 
Consider the space $\mathcal{V}=\operatorname{Pol}(\bR^{2n},S_-)$ of spinor valued polynomials on the Euclidean space $\bR^{2n}$  and its subspace $\mathcal{M}$ of monogenic polynomials, i.e., spinor valued polynomials satisfying the Dirac equation. 
The compact group $G_0$ acts on $\mathcal{M},$ which hence can be decomposed
into a direct sum of irreducible pieces.
It is well known that the subspaces $\mathcal{M}_k $ of homogeneous monogenic polynomials of order $k$ 
 form irreducible $G_0$-modules which are mutually different, see e.g.\ \cite{DSS}.
Hence the decomposition we are looking for is  a multiplicity one decomposition
\begin{equation}\label{Taylor}
 \mathcal{M}=\bigoplus_{k=0}^\infty \mathcal{M}_k
\end{equation}
Note that for $n=1,$ we are back in complex function theory where the Dirac equation coincides
with the classical Cauchy-Riemann equation. In this case, the group $G_0$ is commutative,
the spaces $\mathcal{M}_k$ are one-dimensional, and the decomposition (\ref{Taylor}) is nothing else than the Taylor
series decomposition. In higher dimension, it is impossible to have the decomposition of $\mathcal{M}$
into $G_0$-invariant one-dimensional summands, the decomposition (\ref{Taylor}) is the best we can have.

\vskip 1mm\noindent
(D) {\bf The Fischer decomposition.}
We can ask a similar question as in (B) for the whole space $\mathcal{V} $ of spinor valued polynomials.
The space $\mathcal{V}$ decomposes again under the action of $G_0$ into a direct sum of irrreducible parts. However, now individual pieces appear  with higher (in fact, infinite) multiplicities.
We can formulate the problem as follows.

Under the action of $G_0, $ the space $\mathcal{V}$ can be  uniquely decomposed into its isotypic components
(i.e.,  the sum of all mutually isomorphic irreducible parts of a given type).
There is a way how to fix a preferred decomposition of individual isotypic components into
irreducible parts.
The model for such decompositions  is the classical theory of spherical harmonics (spin $0$ case), sometimes called  'separation of variables', see \cite[p.\ 118]{HT}.
 Consider the space $\mathcal{V}=\operatorname{Pol}(\bR^n)$ of complex valued polynomials on the Euclidean space $\bR^n$ equipped with the scalar product $(x,y)$ 
and its subspace $\mathcal{I}=\mathcal{V}^{G_0}$ of invariant polynomials. The space $\mathcal{I}$
 is generated by $r^{2k}$ with $r^2=(x,x)$ and $k\in\bZ.$
Let $\mathcal{H}$ be its subspace of solutions of the Laplace equation and let $\mathcal{H}_k$ 
denote the space of harmonic polynomials of degree $k.$
Then we have
\begin{equation}\label{Fischerx}
\mathcal{V}\simeq \mathcal{H}\otimes \mathcal{I}.
\end{equation}
It means  that the space $\mathcal{V}$ decomposes as
\begin{equation}\label{Fischer}
\mathcal{V}=\bigoplus_{k=0}^\infty\Is(k)\text{\ \ with\ \ } \Is(k)=\bigoplus_{j=0}^\infty r^{2j}\mathcal{H}_k.
\end{equation}
Here $\Is(k)$ are different isotypic components of $\cV$.
 The space $\mathcal{H}_k$ is characterized inside the isotypic component $\Is(k)$ 
as the space of polynomials having the minimal possible homogeneity.  

The problem to find a Fischer decomposition of the space $\mathcal{V}$ of all considered fields can be formulated as the problem to find
a $G_0$-invariant decomposition of individual isotypic components of $\mathcal{V}$
 into a direct sum of
irreducible $G_0$-modules given by shifts of the lowest homogeneous piece of the given isotypic
summand by $G_0$-invariant fields. In addition, the lowest homogeneous piece of the given isotypic component is formed by the solutions of corresponding $G_0$-invariant differential equations.

 For spin $1/2$ fields,  \eqref{Fischerx} is the classical
 decomposition of spinor valued polynomials as the tensor product of monogenic polynomials and $G_0$-invariant polynomials, see \cite{DSS}.
 A complex version of \eqref{Fischerx} for classical group $SO(n,\bC)$ can be found 
in \cite[Theorem 5.6.11]{GW}.
 
If we consider only polynomial fields, the Fischer decomposition does not depend on the signature
and can be formulated uniformly for polynomials in the common complexification. This is the reason why we consider in the paper polynomials on the complex vector space $\bC^6.$
The decomposition \eqref{Fischerx} can be extended (in any signature) to real analytic fields.
The case of definite signature is very special because, in this case, all solutions of field equations are real analytic.

\vskip 2mm\noindent
(E) {\bf The Howe duality.}

Let us consider the space $D(\mathcal{V})$ of polynomial coefficient differential operators  of all orders on 
$\mathcal{V}=\operatorname{Pol}(\bR^n)$ and the algebra $\mathcal{A}=D(\mathcal{V})^{G_0}$ of all invariant differential operators.
The algebra 
$\mathcal{A}$ preserves isotypic components of  (\ref{Fischer}) and contains a~copy of the Lie algebra $\mathfrak{g}'=sl(2,\bR)$ generated by the laplacian
$\Delta$ and the operator $r^2.$
 The key fact is that every isotypic component of $\mathcal{V}$ forms an irreducible representation of $G_0\times \mathfrak{g}'.$
  Hence the decomposition (\ref{Fischer}) into isotypic components is a multiplicity one decomposition
  of $\mathcal{V}$ under the action of the product $G_0\times \mathfrak{g}'.$
This is a~basic example of the Howe duality for the dual pair $(\mathfrak{so}(n),\mathfrak{sl}(2,\bR))$.  
For a systematic account of the Howe duality developed by R.\ Howe, we refer e.g.\ to \cite{H}.

There is also a general version of the duality theorem 
described by R.\ Goodman and N.\ Wallach in \cite{GW}. In our real version of an action of the compact group $G_0$ on the space  $\mathcal{V}$ of all fields, it says that first we decompose
 $\mathcal{V}$ into the direct sum of its isotypic components and then we consider
 isotypic components with the action of the algebra $\mathcal{A}$ of all $G_0$-invariant polynomial coefficient differential operators of all orders with values in $\operatorname{End}(V).$
 Individual isotypic components are typically irreducible under the joint action of $G_0\times \mathcal{A}$ and they are mutually different.
 This is an approach considered in the paper. 
 
 Often it is possible to find a Lie (super-)algebra $\mathfrak{g}'$ embedded in $ \mathcal{A}$
 in such a way that isotypic components remain irreducible after restriction to the action of
 $\mathfrak{g}'.$  
 But this was not possible in the case considered in the paper.
 
 We are restricting our attention in the paper just to polynomial fields. There are versions
 of the Howe duality for noncompact groups and their dual partners, which are formulated
 for more general function spaces. An important example for spin $1$ fields on Minkowski space is studied in \cite{HLu,L}.

\vskip 1mm\noindent
(F) {\bf Irreducibility of the space of homogeneous solutions.}                                                                                                                                                                                                                                                                                                                                                                                                                                                                                                                                                                                                                                                                                                                                                                                                                                                                                                                                                                                                                                                                                                                                                                                                                                                                                                                                                                                                          
As we explained above, massless field equations for higher spins are not uniquely determined. 
In this paper, as massless field equations we consider the so-called generalized Cauchy-Riemann equations proposed by E.M.\ Stein and G.\ Weiss in \cite{SW}.
An important question is whether the space of solutions of field equations with a fixed homogeneity forms
an irreducible $G_0$-module.

 \vskip 1mm\noindent
 (G) {\bf Conformal invariance.}
 All classical massless field equations are conformally invariant. Question is whether this is true
 also for the generalized Cauchy-Riemann equations considered in the paper.

\vskip2mm

 In the previous papers \cite{BSKS2,BSLSW,LSW}, we  considered  various cases of massless field equations in dimension $4.$
In the present paper, we are going to study massless fields of spin $3/2$
in dimension $6.$ An advantage of this case is the fact that Lie groups $Spin(6,\bC)$ and
$SL(4,\bC)$ are isomorphic, hence it is possible to present  all calculations in the easier
language of tensor fields. Of course, $G_0=Spin(6)$ is a~compact form of $Spin(6,\bC)$.
 
Now we can describe the content of the paper and its structure. 
Chapter 2 introduces notation, and contains basic facts on finite dimensional  representations of the group $G:=SL(4,\bC)$
and their tensor realizations. The form of the generalized Cauchy-Riemann equations and their
explicit description for spin $3/2$ fields is described in Chapter 3. It contains also a definition
of the algebra $\mathcal{A}^G$  of invariant operators acting on isotypic components. 
The parametrization of the spectrum of the space of spinor fields of spin $3/2$ is given in Chapter 4 based on the description of decomposition of certain tensor products of $SL(4,\bC)$-modules.
Proofs for these tensor product decompositions are technical and they are given in Appendix B.
The detailed structure of individual isotypic components is studied carefully in Chapter 5 using 
explicit description
and calculations with highest weight vectors of individual irreducible pieces. It includes
at the same time the description of the Fischer decomposition for spin $3/2$ fields.
The proof of irreducibility of homogeneous solutions of massless field equations is given in Chapter 6.  Finally, Chapter 7 contains the formulation and proof of the Howe duality
theorem for spin $3/2$ fields.  Real forms of the group $SL(4,\bC)$ are described in Appendix A.
Finally, conformal invariance of the generalized Cauchy-Riemann equations is discussed in Appendix \ref{app_ConfInv}.

\section{Notation}


\subsection{The main problem.}  
In what follows, let $\bC^6$ be a~complex Euclidean space and let $\Pol$ denote the space of polynomials $F : \bC^6 \to \bC$. 
For a given (complex and finite dimensional) irreducible representation $V$ of $Spin(6,\bC)$, the main aim is to find an irreducible decomposition of $Spin(6,\bC)$-module $\Pol(V):=\Pol \otimes V$ of $V$-valued polynomials on $\bC^6$. It is well-known that $Spin(6,\bC) \simeq SL(4,\bC)$ and therefore an equivalent problem is to find an irreducible decomposition of $\Pol(V)$ under the action of the group $G  :=  SL(4,\bC) $.

Moreover, we can consider also real slices $M_{p,q}$ of $\bC^6$ on which the metric has signature $(p,q)$. In Appendix \ref{app_Real}, we give more details for Euclidean, Lorentzian and split signature cases, i.e., the cases when $(p,q)=(6,0),$ $(1,5)$, $(3,3)$ and the corresponding real forms of $SL(4,\bC)$ are $G_{p,q}:=SU(4)$, $SL(2,\bH)$, $SL(4,\bR)$, respectively. Of course, $\Pol(V)$ is a~complexification of the space $\Pol_{p,q}(V)$ of $V$-valued polynomials on the real slice $M_{p,q}$. Thus
it is clear that, as complex representations, an irreducible decomposition of $\Pol_{p,q}(V)$ for $G_{p,q}$ is the same as that of 
$\Pol(V)$ under the action of the group $G$.

\subsection{$SL(4,\bC)$-modules.} All (complex and finite dimensional) irreducible representations of $G$ 
are parametrized by their highest weights, which have a form  $$\la=(\la_1,\la_2,\la_3,\la_4),\ \la_i\in\bZ,\ \la_1\geq\la_2\geq\la_3\geq\la_4$$ together with the
equivalence relation
$$\ \la=(\la_1,\la_2,\la_3,\la_4)\simeq(\la_1+j,\la_2+j,\la_3+j,\la_4+j),\ j\in \bZ.$$
Highest weights of $SL(4,\bC)$-modules will be denoted by triples $(\la_1,\la_2,\la_3)\equiv(\la_1,\la_2,\la_3,0)$ as representatives for classes of equivalence and we shall often use the relation
$$(\la_1+1,\la_2+1,\la_3+1,1)\simeq (\la_1,\la_2,\la_3,0).$$ See Appendix \ref{app_LRR} for more details.
For an irreducible module $F_\la$ of the highest weight $\la$ we often write shortly just $\la$. 

Recall that the Cartan product $F_\la\boxtimes F_\mu$ is a~unique irreducible submodule of the tensor product $F_\la\otimes F_\mu$ with the highest weight $\la+\mu$, i.e., $\la\boxtimes\mu=\la+\mu$.
As is well-known, all irreducible $G$-modules can be produced of the fundamental representations we now describe via Cartan multiplication. 
The first fundamental representation of $G$  is the defining representation $S:=\bC^4$ with the highest weight $(1,0,0)$. The representation $S^*:=(\bC^4)^*$ dual to $S$ has the highest weight $(1,1,1)$. 
The last fundamental representation is $\Lambda^2 S$, the second antisymmetric power of $S$. The $G$-module $\Lambda^2 S$ has the highest weight $(1,1,0)$.
It turns out that the $G$-modules $S$ and $S^*$ correspond to the basic spinor representations $S_{\pm}$ of $Spin(6,\bC)$ and $\Lambda^2 S$ is isomorphic to its defining representation $\bC^6$, see \cite[Section 23.1]{FH}. We give more details below.

\subsection{Spinor notation.}
In what follows, we use the abstract index notation of R. Penrose. 
In this notation, for $s\in S$, we write simply $s_a$, where the index $a$ on $s_a$ means that the object belongs to $S$ and does not refer to a~particular coordinate of $s$, see \cite{PR} for more details.  
In a similar way, we write that the action of $g=g^a_b\in{\rm SL}(4,\bC)$ on $s=s_a\in S$ and on $t=t^b\in S^*$ is given by 
	\begin{equation}\label{eq:action}
	(g\cdot s)_b=s_a g^a_b\text{\ \ and\ \ } (g\cdot t)^a=(g^{-1})^a_b t^b. 
	\end{equation}
	For $S$ and $S^*$ we often write simply $S_a$ and $S^a$, respectively.

	We need also tensors with more lower and upper spinor indices, e.g.,
	$$f^{a_1 \cdots\; a_r}_{b_1 \cdots\; b_s}\in S^r_s:=S^{\otimes s}\otimes (S^*)^{\otimes r},$$
	and we define their symmetrization and antisymmetrization as usual, e.g., 
	$$f^{(a_1 \cdots\; a_r)}_{b_1 \cdots\; b_s}=\frac{1}{r!}\sum_{\sigma} f^{a_{\sigma(1)} \cdots\; a_{\sigma(r)}}_{b_1 \cdots\; b_s}\text{\ and\ }
	f^{[a_1 \cdots\; a_r]}_{b_1 \cdots\; b_s}=\frac{1}{r!}\sum_{\sigma} sgn(\sigma) f^{a_{\sigma(1)} \cdots\; a_{\sigma(r)}}_{b_1 \cdots\; b_s}$$
	Here the sums are over all permutations $\sigma$ of the indices $\{1,\ldots,r\}$ and $sgn(\sigma)$ is the sign of the permutation $\sigma$. 
	In particular, we have $f^c_{(ab)}=\frac 12(f^c_{ab}+f^c_{ba})$ and $f^c_{[ab]}=\frac 12(f^c_{ab}-f^c_{ba})$.
	We denote, e.g., $S^c_{(ab)}:=\{f^c_{ab}\in S^1_2|\ f^c_{ab}=f^c_{(ab)}\}$ and $S^c_{[ab]}:=\{f^c_{ab}\in S^1_2|\ f^c_{ab}=f^c_{[ab]}\}$ and use similar notations.
	For example, elements of $S^c_{(ab)}$ are the tensors of $S^1_2$ symmetric in lower indices.

\subsection{The representation $\bC^6:= S_{[ab]}\simeq S^{[ab]}$.}
It is easy to see that the $G$-module $S_{[ab]}$ is isomorphic to $S^{[ab]}$. An isomorphism between these two modules is
 		 given by contraction with a chosen invariant antisymetric tensor $\varepsilon_{abcd}\in S_{[abcd]}$
 		 defined as the
 		 sign of the permutation $abcd$. 
 		 This is an analogue of the similar tensor $\varepsilon_{AB}$ used for raising and lowering spinor indices in dimension $4$, see \cite{BSLSW}. 
		 Indeed, we use $\varepsilon_{abcd}$ to lower indices of $z^{cd}\in S^{[ab]}$ as follows
		$$z_{ab}:=\frac 12 \varepsilon_{abcd}z^{cd}.$$
		Then we have 
		$$z^{ab}=\frac 12 \varepsilon^{abcd}z_{cd}$$
		where $\varepsilon^{abcd}\in S^{[abcd]}$ is defined again as the sign of the permutation $abcd$. 
		The latter equality follows easily from the identity $$\varepsilon^{abcd}\varepsilon_{cdef}=4\delta_e^{[a}\delta_f^{b]}.$$ 
		It is easy to see that
$z_{12}=z^{34}, z_{13}=z^{42}, z_{14}=z^{23}, z_{23}=z^{14}, z_{24}=z^{31}, z_{34}=z^{12}$.

On $S^{[ab]}\simeq S_{[ab]}$	we have a~natural non-degenerate symmetric bilinear form
$$(z,w):=\frac 14 \varepsilon_{abcd}z^{ab}w^{cd}=\frac 12 z^{ab}w_{ab}=\frac 14 \varepsilon^{abcd}w_{ab}z_{cd}.$$ In particular, we have
$(z,z)=2(z^{12}z^{34}+z^{13}z^{42}+z^{14}z^{23})$.
This is a complex Euclidean space of dimension 6 we denote by $\bC^6$.

The space $\Pol$ of complex valued polynomials on $\bC^6$ is isomorphic to the symmetric algebra $Sym(\bC^6)$. In what follows, we need the standard harmonic decomposition for such polynomials. Indeed, denote $$\nabla_{ab}=\frac{\partial}{\partial z^{ab}}\text{\ \ and\ \ }\nabla^{ab}=\frac 12 \varepsilon^{abcd}\nabla_{cd}.$$ Then it is easy to see that 
$\nabla_{ab} z^{cd}=2\delta_a^{[c}\delta_b^{d]}$ and $\nabla_{ab} z_{cd}=\varepsilon_{abcd}$. We have the following well-known result for some important $G$-invariant operators on $\Pol$.

\begin{prop}\label{l81}
	(i) Define the Laplace operator $\Delta:=\frac 12 \nabla_{ab}\nabla^{ab}$, $r^2:=\frac 12 z_{ab}z^{ab}$ 
	and the Euler operator $E:=\frac 12 z_{ab}\nabla^{ab}$. Then $\Delta$, $r^2$, $E$ form an $\mathfrak{sl}_2(\bC)$-triple, i.e., the following commutation relations hold
  $$[\Delta, r^2]= 4(E+3),\ [E+3,\Delta]=-2\Delta\text{\ and\ }[E+3,r^2]=2r^2.$$
	Moreover, we have $z_{bc}z^{ac}=\frac 12 r^2\delta^a_b$ and $\nabla_{bc}\nabla^{ac}=\frac 12 \Delta\delta^a_b$.
	
	\noindent\smallskip
(ii) For $m\in\bN_0$, denote by 
$\Pol_m=\{P\in\Pol|\ E P=mP\}$
the subspace of $m$-homogeneous polynomials $P\in\Pol$ and
by 
$\cH_m=\{P\in\Pol_m|\ \Delta P=0\}$
the subspace of harmonic polynomials $P\in\Pol_m$.
Then $$\Pol_m=\bigoplus_{j=0}^{\lfloor m/2 \rfloor} r^{2j}\cH_{m-2j}$$ is the well-known harmonic decomposition. 
Here $\lfloor m/2 \rfloor$ is the integer part of $m/2$. In particular, we have that $\cH_m$ forms an irreducible $G$-module with the highest weight $(m,m,0)$.
\end{prop}

\begin{proof} (i) Simple computation. 
(ii) Harmonic decomposition for scalar valued polynomials is well-known in any dimension. This result formulated in antisymmetric variables $z^{ab}$ of $\bC^6$ follows directly from \cite[Theorem 3.8.1, p.\ 44]{How}.
\end{proof}

\begin{ex} It is clear that a~given polynomial $P\in\Pol_2$ can be written uniquely in the form 
$$P(z)=p_{abcd}z^{ab}z^{cd}$$ with $p_{abcd}=p_{[ab][cd]}=p_{cdab}$. By a~straightforward computation, $P\in\cH_2$ if and only if $p_{[abcd]}=0,$ which is equivalent to
$p_{a[bcd]}=0,$ or even to $p_{1234}+p_{1342}+p_{1423}=0$. Indeed, we have $\Delta P=4\varepsilon^{abcd}p_{abcd}$.
By the commutation relations in Proposition \ref{l81} above, we get that 
$P(z)=H(z)+\frac13 \varepsilon^{abcd}p_{abcd}\, r^2$ for some uniquely determined $H\in\cH_2$.
\end{ex}

\subsection{Tensor realization of $SL(4,\bC)$-modules}

The module $S_a$ has the highest weight $(100)$ and has
the weight vector basis $\delta^i_a,$ $i=1,2,3,4$. 
It is well-known that the Cartan power $\boxtimes^k(S_a)\simeq(k00)$ is isomorphic  to the symmetric power $\odot^k(S_a)$.
Similarly, $S^a\simeq(111)$ has
the weight vector basis $\delta_i^a,$ $i=1,2,3,4$ and $\boxtimes^k(S^a)\simeq\odot^k(S^a)\simeq(kkk)$.
For example, we have $S_{(abc)}\simeq (300)$ and $S^{(abc)}\simeq (333)$.

Now we realize other modules needed later on. First we consider tensor product decompositions 
$$ (200)\otimes (111) \simeq (311)\oplus (100),\ \
 (222)\otimes (100) \simeq (322)\oplus (111).$$  
For Littlewood-Richardson rules for tensor product decompositions, see Appendix \ref{app_LRR}.
From these decompositions, it is easy to see that $(311)$ can be realized as the module 
$$\left (S_{(ab)}^c\right)_0:=\{f_{ab}^c\in S_{(ab)}^c|\ f_{ac}^c=0\}.$$ 
Here the subscript $0$ refers to the fact that $\left (S_{(ab)}^c\right)_0$ is the submodule of traceless tensors $f_{ab}^c\in S_{(ab)}^c$.
Similarly, we have $\left (S_{c}^{(ab)}\right)_0\simeq (322)$ where
$$\left (S^{(ab)}_c\right)_0:=\{f^{ab}_c\in S^{(ab)}_c|\ f^{ac}_c=0\}.$$ 
Using the tensor product decompositions
 		  $$ (110)\otimes (111)\simeq (221)\oplus (100),\ \ (110)\otimes (100)\simeq (210)\oplus (111),$$
 			we get the realisations 
$\left (S^{[ab]}_c\right)_0\simeq (210)$ and $\left (S_{[ab]}^c\right)_0\simeq (221)$.			
It turns out that we shall need tensor spaces 
 					\begin{equation}\label{eq_spaces}					
					S^{(abc)},\ S_{(abc)},\ \left(S^{ab}_{c}\right)_0=\left(S^{(ab)}_{c}\right)_0\oplus \left(S^{[ab]}_{c}\right)_0,\ \left(S_{ab}^{ c}\right)_0=\left(S_{(ab)}^c\right)_0\oplus\left (S_{[ab]}^c\right)_0,\ S^a,\ S_a.
           \end{equation}
					
In the paper, we study polynomial fields on $\bC^6$ we realize as polynomials  in the variables $z^{ab}$ with values in the tensor algebra of the vector spaces $S_a$ and $S^a.$	We denote the modules of polynomial fields with values in the spaces \eqref{eq_spaces}, respectively, by
$$\cE^{(abc)},\ \cE_{(abc)},\  {\left(\cE^{ab}_{c}\right)_0=\left(\cE^{(ab)}_{c}\right)_0\oplus \left(\cE^{[ab]}_{c}\right)_0,\
\left(\cE_{ab}^{ c}\right)_0=\left(\cE_{(ab)}^c\right)_0\oplus\left(\cE_{[ab]}^c\right)_0,\ \cE^a,\cE_a}.
$$
and, in addition, put
$$
\cE^{\frac{3}{2}}:=\cE^{(abc)}\oplus(\cE^{ab}_{c})_0\oplus(\cE_{ab}^{c})_0\oplus\cE_{(abc)},\ \
\cE^{\frac{1}{2}}:=\cE^{a}\oplus\cE_{a}.
$$

\section{Invariant differential operators}

\subsection{Generalized Cauchy-Riemann operators}\label{ss_GCR}

Generalized Cauchy-Riemann (GCR) operators  were introduced in any dimension by Stein and Weiss
in \cite{SW}. In dimension~$4$, the GCR operators were used in \cite{LSW} as general massless field equations. Let us briefly recall their definition in dimension 6.

Let us consider a~general (complex and finite dimensional) irreducible $G$-representation $F$. We want to construct $G$-invariant differential operators of first order on $\Pol(F)$, the polynomial fields on $\bC^6$ with values in $F$. 
The same construction gives differential operators even for smooth fields.
To do this, we start with a~multiplicity free decomposition of the tensor product 
$$F\otimes \bC^6=F_0\oplus F_1\oplus\cdots\oplus F_k$$  
into different irreducible representations where $F_0=F\boxtimes \bC^6$ is the Cartan product. Denote by $\pi_j$ the invariant projection of $F\otimes \bC^6$ onto $F_j$.  For $f\in\Pol(F)$ and $z\in\bC^6$,
it is clear that the gradient $\nabla f(z)\in F\otimes \bC^6$. By construction, the Stein-Weiss operator 
$D_j:\Pol(F)\to\Pol(F_j)$
defined by
$$D_j(f)(z):=\pi_j(\nabla f(z)),\ f\in\Pol(F),\  z\in\bC^6,$$
is an invariant differential operator on $\Pol(F)$. We define the GCR operator $\cD^F$ for the given irreducible representation $F$ as the projection of the gradient onto the complement of the Cartan product in $F\otimes \bC^6$, i.e.,
$$\cD^F=D_1+\cdots+D_k.$$
Obviously,  for $f\in\Pol(F)$, we have that $\cD^Ff=0$ if and only if $D_j f=0$ for each $j=1,\ldots,k$. 

If we replace, in the Stein-Weiss construction, the gradient $\nabla$ with the tensor multiplication by the vector variable $z$ we get the invariant multiplication operator $z_j:\Pol(F)\to\Pol(F_j)$ defined by
$$z_j(f)(z):=\pi_j(z\otimes f(z)),\ f\in\Pol(F),\  z\in\bC^6.$$

In this part, we write down the Stein-Weiss operators and the corresponding multiplication operators for fields of spin $3/2$. 
Let us start with the canonical spin 3/2 representations $S_{(abc)}$ and $S^{(abc)}$.

\begin{lem}\label{5.1}
(i)	The first factor in the decomposition
\begin{equation}\label{dec1}
(300)	\otimes (110) = {(410)\oplus(311)}
\end{equation}
	is the Cartan product. Hence the GCR operator for  $S_{(abc)}$ is the  invariant operator
	\begin{equation}\label{SW1}
\mathcal{D}^{300}=	D^{300}_{311}: P(300)\to P(311),
	\end{equation}
	given by $D^{300}_{311}f= \pi(\nabla f),$ where $P(\la)=\Pol(F_{\la})$ and $\pi $ is the invariant projection to the second summand in
\eqref{dec1}. For $f\in\cE_{(abc)}$, we have
 $
D^{300}_{311}f=\nabla^{ab}f_{acd}
	$.

\smallskip\noindent
(ii) Similarly,
the first factor in the decomposition
\begin{equation}\label{dec2}
	(333)	\otimes (110) =(443)\oplus(322 {) }
\end{equation}
is the Cartan product, hence the GCR operator for  $S^{(abc)}$ is the  invariant operator
	\begin{equation}\label{SW2}
\mathcal{D}^{333}= D_{322}^{333}:P(333)\to P(322)
	\end{equation}
given by $D^{333}_{322}f= \pi(\nabla f),$ where     $\pi $ is its invariant projection to the second summand in
\eqref{dec2}. For $f\in\cE^{(abc)}$, we have
 $
D^{333}_{322}f=\nabla_{ab}f^{acd}
	$.

\smallskip\noindent
(iii) In an analogous way, we introduce the corresponding multiplication operators $z^{300}_{311}$ and $z^{333}_{322}$.
For $f\in\cE_{(abc)}$ and $g\in\cE^{(abc)}$, we have
 $
z^{300}_{311}f=z^{ab}f_{acd}$ and $z^{333}_{322}g=z_{ab}g^{acd}$.
\end{lem}

\begin{proof} (i) It is easy to get the decomposition \eqref{dec1} by Littlewood-Richardson rules, see Appendix \ref{app_LRR}.
For an explicit expression of $
D^{300}_{311}$, it is sufficient to notice that the tensor $\nabla^{ab}f_{acd}$ is symmetric in $c,d$ and traceless due to
  the contraction of  an antisymmetric pair 
 with  a symmetric pair of indices. We can prove (ii) similarly. As for (iii), it is clear that we get the corresponding multiplication operators $z^{\la}_{\mu}$ from $D^{\la}_{\mu}$ by interchanging $\nabla^{ab}$ or $\nabla_{ab}$ with $z^{ab}$ or $z_{ab}$, respectively.
\end{proof}

In the image of the two GCR operators of Lemma 3.1 are fields with values in $(311)$ and $(322).$
Hence  we have to consider the GCR operators for these two new types of fields.

\begin{lem}\label{5.2}
(i)	The first factor in the decomposition
	\begin{equation}\label{dec3}
		(311)	\otimes (110) =(421)\oplus[(322)\oplus(210)\oplus (300)]
	\end{equation}
	is the Cartan  product, hence the GCR operator for $(311)$ is the  invariant operator $$\mathcal{D}^{311}=D^{311}_{322}\oplus D^{311}_{210}\oplus D^{311}_{300}$$ where the operators on the right-hand side are the invariant projections of the gradient  onto   the   corresponding summands in the
	 decomposition \eqref{dec3}. For $f\in\left (\cE^a_{(bc)}\right )_0$, we have 	
	\begin{equation}\label{SW2.1}
	D^{311}_{322}f=\nabla^{b(d}f^{a)}_{bc},\ \ 
	D^{311}_{210}f=\nabla^{b[d}f^{a]}_{bc}\text{\ \ and\ \ }
	D^{311}_{300}f=\nabla_{a(d}f^{a}_{bc)}.
	\end{equation}

	\smallskip\noindent
(ii) Similarly, 	the first factor in the decomposition
	\begin{equation}\label{dec4}
	(322)	\otimes (110) =(432)\oplus[(311)\oplus(221)\oplus (333)]
\end{equation}
is the Cartan  product, hence the GCR operator for  $(322)$ is the  invariant operator $$\mathcal{D}^{322}=D^{322}_{311}\oplus D^{322}_{221}\oplus D^{322}_{333}$$ where the operators on the right-hand side are the invariant projections of the gradient  onto   the   corresponding summands in the
	 decomposition \eqref{dec4}.		
For $f\in\left (\cE_a^{(bc)}\right )_0$, we have 	
	\begin{equation}\label{SW2.2}
	D^{322}_{311}f=\nabla_{b(d}f_{a)}^{bc},\ \ 
	D^{322}_{221}f=\nabla_{b[d}f_{a]}^{bc}\text{\ \ and\ \ }
	D^{322}_{333}f=\nabla^{a(d}f_{a}^{bc)}.
	\end{equation}

	\smallskip\noindent
(iii) We denote the corresponding multiplication operators $z^{311}_{322}$ etc.
\end{lem}	

A~proof of Lemma \ref{5.2} is quite analogous as that of Lemma \ref{5.1}. Let us just mention that we can get explicit expressions \eqref{SW2.2} easily from \eqref{SW2.1} by interchanging lower and upper indices.

There are new types of fields in Lemma 3.2, hence we need to write down the GCR equations
for these fields.

\begin{lem}\label{5.3}
(i)	The first factor in the decomposition
	\begin{equation}\label{dec5}
		(210) \otimes (110) =(320)\oplus[(311)\oplus(221)\oplus (100)]
	\end{equation}
	is the Cartan  product, hence the GCR  operator for  $(210)$ is the  invariant operator $$\mathcal{D}^{210}=D^{210}_{311}\oplus D^{210}_{221}\oplus D^{21 0}_{100}$$ where the operators on the right-hand side are the invariant projections of the gradient  onto   the   corresponding summands in the
	 decomposition \eqref{dec5}.	
	For $f\in\left (\cE_a^{[bc]}\right )_0$, we have 	
	\begin{equation}\label{SW3.1}
	D^{210}_{311}f=\left (\nabla_{b(d}f_{a)}^{bc}\right )_0,\ \ 
	D^{210}_{221}f=\left (\nabla_{b[d}f_{a]}^{bc}\right )_0\text{\ \ and\ \ }
	D^{210}_{100}f=\nabla_{bc}f_{a}^{bc}.
	\end{equation}
	Here $f_0$ is the traceless part of a~tensor $f\in\cE^a_{(bc)}$, resp.\ $\cE^a_{[bc]}$.
	
	\smallskip\noindent
(ii) Similarly, 	the first factor in the decomposition
	\begin{equation}\label{dec6}
		(221)	\otimes (110) =(331)\oplus[(322)\oplus(210)\oplus (111)]
	\end{equation}
	is the Cartan product, hence the GCR operator for  $(221)$ is the  invariant operator $$\mathcal{D}^{221}=D^{221}_{322}\oplus D^{221}_{210}\oplus D^{221}_{111}$$ where the operators on the right-hand side are the invariant projections of the gradient  onto   the   corresponding summands in the
	 decomposition \eqref{dec6}.	
		For $f\in\left (\cE^a_{[bc]}\right )_0$, we have 	
	\begin{equation}\label{SW3.2}
	D^{221}_{322}f=\left (\nabla^{b(d}f^{a)}_{bc}\right )_0,\ \ 
	D^{221}_{210}f=\left (\nabla^{b[d}f^{a]}_{bc}\right )_0\text{\ \ and\ \ }
	D^{221}_{111}f=\nabla^{bc}f^{a}_{bc}.
	\end{equation}
		
		\smallskip\noindent
(iii) Again we denote the corresponding multiplication operators $z^{210}_{311}$ etc.
\end{lem}

Finally, we have new fields $(100)\simeq S_a$ and $(111) \simeq S^a$ appearing  in Lemma \ref{5.3}, which are the basic spinor modules and the GCR operators for them are just Dirac operators
from classical Clifford analysis, see e.g.\ \cite{DSS}.
The corresponding tensor product decompositions are
\begin{equation}\label{diracs}
(100)\otimes (110)= (210)\oplus (111),\ \ (111)\otimes (110)= (221)\oplus (100).	
\end{equation}
 The GCR operators $\cD^{100}=D^{100}_{111}$ and $\cD^{111}=D^{111}_{100}$ are given as
 the invariant projections of the gradient  onto   the   corresponding summands in the
	 decompositions  \eqref{diracs}. We also need the twistor operators $D^{100}_{210}$ and $D^{111}_{221}$ because, in this case, the Cartan products $(210)$ and $(221)$ of \eqref{diracs} belong to $V$.
	For $f\in\cE_a$ and $g\in\cE^a$, we have 	
	\begin{equation}\label{SW4}
	D^{100}_{210}f=\left(\nabla^{ab}f_c\right)_0,\ \ D^{100}_{111}f=\nabla^{ab}f_b,\ \ D_{221}^{111}g=\left(\nabla_{ab}g^c\right)_0\ \ \text{\ \ and\ \ }
	D_{100}^{111}g=\nabla_{ab}g^b.
	\end{equation}
	We denote the corresponding multiplication operators $z^{100}_{210}$, $z^{100}_{111}$, $z_{221}^{111}$ and $z^{111}_{100}$.
	
So we learned that for investigation of massless field equations of spin 3/2, we have to consider
8 different types of spinor fields and their GCR equations. Indeed, in this paper we study the following polynomial fields on $\bC^6$.

\begin{dfn}\label{d_values}
	Consider the module
$$
V=(300)\oplus(311)\oplus (322)\oplus(333)\oplus
(210)\oplus(221)\oplus(100)\oplus(111)
$$
and the space $\cV=\Pol(V)$ of $V$-valued polynomial fields on $\bC^6$.
 \end{dfn}
For each irreducible submodule $\la$ of $V$, let
$P(\la)$ be the space of $\la$-valued polynomials on $\bC^6$, and let $P_m(\la)$ be the space of such  polynomials which are homogeneous of degree $m$.
Then it turns out that the $G$-module
$$\cM_m(\la)=\{f\in P_m(\la)|\ \cD^{\la}f=0\}$$ 
of $m$-homogeneous solutions of the GCR equation for $\la$ is irreducible and $\cM_m(\la)\simeq\cH_m\boxtimes\la$.
This is well-known for spin 1/2 modules $(100)$ and $(111)$, see e.g.\ \cite{DSS}. For spin 3/2 submodules $\la$ of $V$, we show this  in 
Section \ref{s_homog_sols}.
In dimension 4, the analogous result was proved in \cite{LSW}.

\begin{ex} 
It is clear that a~given polynomial field $f\in P_2(300)$ is of the form 
$$f_{abc}(z)=f_{abcd_1e_1d_2e_2}z^{d_1e_1}z^{d_2e_2}$$
with $f_{abcd_1e_1d_2e_2}=f_{(abc)[d_1e_1][d_2e_2]}=f_{abcd_2e_2d_1e_1}$.
By Lemma \ref{5.1}, we have $$\cD^{300}f=\nabla^{af}f_{abc}=2f_{abcd_1e_1d_2e_2}\varepsilon^{afd_1e_1}z^{d_2e_2}.$$
Then, obviously, $f\in \cM_2(300)$ if and only if $f_{bc[ad_1e_1]d_2e_2}=0$.
It seems that a~similar characterization of $\cM_m(\la)$ in the general case is not so simple. On the other hand, Theorem \ref{t_basis} below describes an explicit construction of a~basis for any $\cM_m(\la)$.
\end{ex}

Now we verify that multiplication operators $z_{\la}^{\mu}$ are adjoint to differential operators $D^{\la}_{\mu}$.
Indeed, for the compact form $SU(4)$ of $SL(4,\bC)$, there is an invariant Hermitian inner product $(\cdot,\cdot)_1$ on $V$. 
On the space $\Pol$ of complex valued polynomials on $\bC^6$, we define the Fischer inner  product as
		    \begin{equation}\label{e_Fischer_prod}
		(f,g)_0=\overline{g}(\nabla_{ab})f(z^{ab})|_{z^{ab}=0}
	\end{equation} 
For $g\in\Pol$, here $\overline{g}$ is the polynomial whose coefficients are equal to complex conjugate coefficients of $g$	and ${g}(\nabla_{ab})$ is the differential operator we get from the polynomial $g$ by replacing $z_{ab}$ with $\nabla_{ab}$.
Then we introduce an $SU(4)$-invariant Hermitian inner product on $\cV=\Pol\otimes V$ as
\begin{equation}\label{e_inner_prod}
		(f\otimes u,g\otimes v)=(f,g)_0(u,v)_1,\ f,g\in\Pol,\ u,v\in V.
	\end{equation}

\begin{prop}\label{p_duality} 
Let $D^{\la}_{\mu}$ be one of the differential operators given in Lemmas \ref{5.1}, \ref{5.2} and \ref{5.3}. 

\noindent\smallskip
(i) Then the multiplicative  operator $z_{\la}^{\mu}$ is the adjoint to $D^{\la}_{\mu}$ with respect to the Hermitian inner product \eqref{e_inner_prod}, i.e.,
$$(D^{\la}_{\mu}f,g)=(f,z_{\la}^{\mu}g),\ f\in P(\la),\; g\in P(\mu) $$
where $P(\la)$ denotes the space of $\la$-valued polynomials on $\bC^6$. 

\noindent\smallskip
(ii) Let $\cI$ be an irreducible finite dimensional $SL(4,\bC)$-submodule of $P(\mu) $. If the operator $z_{\la}^{\mu}$ is non-trivial on $\cI$ then $z_{\la}^{\mu}$ is an invariant isomorphism of $\cI$ onto its image $\cJ:=z_{\la}^{\mu}(\cI)$ in $P(\la)$. Moreover, 
if $D^{\la}_{\mu}(\cJ)\subset\cI$ then
the inverse operator to $z_{\la}^{\mu}|_{\cI}$ is a~non-zero multiple of $D^{\la}_{\mu}|_{\cJ}$.
\end{prop}

\begin{proof} (i) It is clear that we can extend the $SU(4)$-invariant Hermitian inner product $(\cdot,\cdot)_1$ to the module
$$\tilde V=V\oplus (410)\oplus(443)\oplus(421)\oplus(432)\oplus(320)\oplus(331)$$ 
containing also the Cartan pieces of the tensor products 
\eqref{dec1}, 
\eqref{dec2},
\eqref{dec3},
\eqref{dec4},
\eqref{dec5} and
\eqref{dec6}.
Then we extend the $SU(4)$-invariant Hermitian inner product $(\cdot,\cdot)$ to $\tilde\cV:=\Pol\otimes\tilde V$.

Let $f\in P(\la)$ and $g\in P(\mu)$. Then it is easy to see that 
$$(\nabla f(z), g(z))=(f(z),z\otimes g(z)).$$
Indeed, for $P,Q\in\Pol$, we have that $(\nabla_{ab} P,Q)_0=(P,z^{ab} Q)_0$ by \eqref{e_Fischer_prod}.

Since the spaces $P(\la')$ and $P(\mu)$ are orthogonal for $\la'\not=\mu$ we have that
$$(\nabla f(z), g(z))=(\pi_{\mu}(\nabla f(z)), g(z))=(D^{\la}_{\mu}f,g)$$ where $\pi_{\mu}$ denotes the orthogonal projection of $\tilde\cV$ onto $P(\mu)$. 
Similarly, we have
$$(f(z),z\otimes g(z))=(f(z),\pi_{\la}(z\otimes g(z)))=(f,z_{\la}^{\mu}g),$$
which completes the proof of (i).

\noindent\smallskip
(ii) By (i), we have that
$$(D^{\la}_{\mu}f,g)=(f,z_{\la}^{\mu}g),\ f\in \cJ,\; g\in \cI.$$ 
Hence $D^{\la}_{\mu}$ is non-trivial on $\cJ$. The rest follows from the Schur lemma which tells us that any non-trivial invariant linear operator $\cF:\cI\to\cJ$ is an invariant isomorphism of $\cI$ onto $\cJ$ and the only invariant linear operators of $\cI$ into $\cI$ are just multiples of the identity.
 \end{proof}

\begin{rem}\label{r_duality} 
It is easy to see that
$(\Delta f,g)=(f,r^2g),\ f,g\in \cV.$
Let $\cI$ be an irreducible finite dimensional $SL(4,\bC)$-submodule of $\cV$. By the Schur lemma,  if $\Delta(r^2\cI)\subset\cI$ then 
$\Delta$ is an invariant isomorphism of $r^2\cI$ onto $\cI$ and $\Delta|_{r^2\cI}$  is a~non-zero multiple of 
the inverse operator to $r^2|_{\cI}$.
\end{rem}

\vskip 2mm
\subsection{The algebra $\cA^G$ of invariant operators}

Using Stein-Weiss gradients $D^{\la}_{\mu}$ and corresponding multiplication operators $z^{\la}_{\mu}$ given in Section \ref{ss_GCR}, we introduce  
basic invariant operators on $\cV$. 
For a~given field $\phi\in\cV$, its components $\phi_{\la}\in\Pol\otimes F_{\la}$ are labelled by the highest weights $\la$ in $V$,
i.e., $\phi=(\phi_{\la})_{\la\in V}$, see Definition \ref{d_values}.

\vskip 2mm \noindent  
  {\bf (1) The operators $D^+_-, z^+_-$ and $D^-_+, z^-_+$ } 
	
	First we introduce a~natural invariant differential operator on $\cV$ which adds one upper index and removes one lower index of fields.
	Thus we define the invariant differential operator $D^+_-:\cV\to\cV$, for $\phi\in\cV$, as $\psi=D^+_-(\phi)$ with
$\psi_{300}=0$,
$\psi_{311}=D^{300}_{311}\phi_{300}+D^{322}_{311}\phi_{322}+D^{210}_{311}\phi_{210}$,
$\psi_{221}=D^{322}_{221}\phi_{322}+D^{210}_{221}\phi_{210}$,
$\psi_{322}=D^{311}_{322}\phi_{311}+D^{221}_{322}\phi_{221}$,
$\psi_{210}=D^{311}_{210}\phi_{311}+D^{221}_{210}\phi_{221}$,
$\psi_{333}=D^{322}_{333}\phi_{322}$,
$\psi_{100}=0$,
$\psi_{111}=D^{100}_{111}\phi_{100}$.
For explicit expressions of $\psi_{\la}$ in tensorial language, see the previous section. 

Using these explicit expressions and interchanging lower and upper indices, we define the operator $D_+^-$ which adds one lower index and removes one upper index of fields.    

By interchanging $D^{\la}_{\mu}$ with $z^{\la}_{\mu}$ in the definition of $D^+_-$ and $D_+^-$, we get the multiplication operators $z^+_-$ and $z_+^-$ on $\cV$.

 \vskip 2mm \noindent  
{\bf (2) The operators $D^{--}, z^{--}$ and $D_{--}, z_{--}$}

Now we introduce a~differential operator on $\cV$ which removes two upper indices but keeps lower indices of fields.
	Thus we define the operator $D^{--}:\cV\to\cV$, for $\phi\in\cV$, as $\psi=D^{--}(\phi)$ with
$\psi_{100}=D^{210}_{100}\phi_{210}$ and $\psi_{\la}=0$ otherwise. 

By interchanging lower and upper indices, we define the operator $D_{--}$ which removes two lower indices but keeps upper indices of fields. By interchanging $D^{\la}_{\mu}$ with $z^{\la}_{\mu}$, we get the multiplication operators $z^{--}$ and $z_{--}$ on $\cV$.

\vskip 2mm \noindent  
{\bf (3) The operators $D^{++}, z^{++}$ and $D_{++}, z_{++}$}

Finally, we introduce a~differential operator on $\cV$ which adds two upper indices but keeps lower indices of fields.
	Thus we define the operator $D^{++}:\cV\to\cV$, for $\phi\in\cV$, as $\psi=D^{++}(\phi)$ with
$\psi_{210}=D^{100}_{210}\phi_{100}$ and $\psi_{\la}=0$ otherwise. 

By interchanging lower and upper indices, we define the operator $D_{++}$ which adds two lower indices but keeps upper indices of fields. By interchanging $D^{\la}_{\mu}$ with $z^{\la}_{\mu}$, we get the multiplication operators $z^{++}$ and $z_{++}$ on $\cV$.

\begin{dfn}
The algebra of invariant operators  generated by the following list of first order invariant differential operators
$$
D^+_-, D^-_+,
D^{--}, D_{--},
D^{++}, D_{++}
$$
and their algebraic  {versions} inside the algebra of $End(V)$-valued differential operators with polynomial coefficients will be denoted by $\cA^G$  where $G= SL(4,\bC)$. 
\end{dfn}


\section{Spectrum} 

As we know to study spin 3/2 case in dimension 6 we should consider the space $\cV=\Pol\otimes V$ of polynomials with values in the $SL(4,\bC)$-module
	$$
	V=(300)\oplus(311)\oplus (322)\oplus(333)\oplus
	(210)\oplus(221)\oplus(100)\oplus(111)
	$$
The spectrum of $\cV$ is the list of all types of irreducible $SL(4,\bC)$-modules contained in it, and so
the spectrum parametrizes the isotypic components of $\cV.$
Obviously, every isotypic component 
in $\cV$ is a module for the action of the algebra $\cA^G$ of invariant operators.

\begin{thm}[\bf The spectrum]\label{spectrum}	
	The spectrum $Spec(\cV)$ of $\cV$ is given by the list
 $$ (m+3,m,0),(m+3,m+1,1),(m+3,m+2,2),(m+3,m+3,3),(m+1,m,0),(m+1,m+1,1),\ m\geq 0.
 $$
\end{thm}
\begin{proof}
By Proposition \ref{l81}, the spectrum of the space $\Pol$ of 
complex valued polynomials
is given by the list $(m,m,0),\ m\geq 0,$ thus it is sufficient  to check tensor product decompositions given in Theorem \ref{6.1} below.
\end{proof}

\begin{thm}\label{6.1}{\ }\quad
	For $SL(4,\bC)$-modules, we have the following tensor product decompositions
	\begin{enumerate}
		\item

 {	$(300)\otimes(110)=(410)\oplus(311),
$
\\	
$	(300)\otimes(220)=(520)\oplus(421)\oplus(322),
$	
\\
$m\geq 3:	(300)\otimes (m,m,0)= (m+3,m,0)\oplus(m+2,m,1)\oplus(m+1,m,2)\oplus(m,m,3);
$}

\item
$(311)\otimes (110)=[(421)\oplus(322)]\oplus(210)\oplus(300),
$
\\
$(311)\otimes (220)=[(531)\oplus(432)\oplus(333)]\oplus[(320)\oplus(221)]\oplus[(410)\oplus(31            1)],
$
\\
$m\geq 3:
(311)\otimes(m,m,0)=[(m+3,m+1,1)\oplus(m+2,m+1,2)\oplus(m+1,m+1,3)]
$

\hskip 40mm
$\oplus[(m+1,m,0)\oplus(m,m,1)] $

\hskip 40mm
$\oplus
[(m+2,m-1,0)\oplus(m+1,m-1,1)\oplus(m,m-1,2)];
$

		\item
		$(322)\otimes (110)=[(432)\oplus(311)]\oplus(221)\oplus(333),
		$
		\\
		$(322)\otimes (220)=[(542)\oplus(421)\oplus(300)]\oplus[(331)\oplus(210)]\oplus[(443)\oplus(322)],
		$
		\\		
		$m\geq 3:
		(322)\otimes(m,m,0)=[(m+3,m+2,2)\oplus(m+2,m,1)\oplus(m+1,m-2,0)]\oplus
		$
		
\hskip 40mm
$
		 \oplus[(m+1,m+1,1)\oplus(m,m-1,0)] $
		
		\hskip 40mm
$ \oplus
		[(m+2,m+2,3)\oplus(m+1,m,2)\oplus(m,m-2,1)];
		$

 \item
$(333)\otimes(110)=(443)\oplus(322),
$
\\
$	(333)\otimes(220)=(553)\oplus(432)\oplus(311),
$	
\\
$m\geq 3:	(333)\otimes (m,m,0)= (m+3,m+3,3)\oplus(m+2,m+1,2)\oplus(m+1,m-1,1)\oplus(m,m-3,0);
$

\item
$(210)\otimes (110)=(320)\oplus(311)\oplus(221)\oplus(100),
$
\\
$m\geq 2: (210)\otimes (m,m,0)=(m+2,m+1,0)\oplus(m+2,m,1)\oplus(m+1,m+1,1)\oplus(m+1,m,2) $

\hskip 40mm
$\oplus(m,m-1,0)\oplus(m-1,m-1,1)$;

\item $(221)\otimes (110)=(331)\oplus(322)\oplus(210)\oplus(100),
$

$(221)\otimes (220)=(441)\oplus(432)\oplus(320)\oplus(311)\oplus(221)\oplus(100),
$
\\
$m\geq 2: (221)\otimes (m,m,0)=(m+2,m+2,1)\oplus(m+2,m+1,2)\oplus(m+1,m,0)\oplus(m+1,m-1,1) $

\hskip 40mm
$
 \oplus(m,m,1)\oplus(m-1,m-2,0)$;

\item
$m\geq 1: (100)\otimes (m,m,0)=(m+1,m,0)\oplus (m,m,1);
$

\item
$m\geq 1: (111)\otimes (m,m,0)=(m+1,m+1,1)\oplus (m,m-1,0).
$\end{enumerate}
\end{thm}
\begin{proof}
The first six statements are proved in Appendix \ref{app_LRR}.
The last two items follow also from the well known Pieri rules.
\end{proof}

\section
{The Fischer decomposition for isotypic components}

The Lie subalgebra in $\cA^G$ generated by the Laplacian 
$\Delta$ 
and mupltiplication by 
$r^2$
is isomorphic to $\mathfrak{sl}_2(\bC)$, see Proposition \ref{l81}, and each isotypic component of $\cV=\Pol\otimes V$ 
decomposes into a direct sum of irreducible pieces under its action. Such decompositions
are usually called {{\it  Fischer decompositions}}.

Now we are going to study structure of individual  isotypic components of $\cV$ one after another. 
We want to describe their decomposition into irreducible components under the action of $G=SL(4,\bC)$ and the action of invariant operators from $\cA^G.$
We are also going to show in the next sections  that each isotypic component has the unique irreducible $G$-module
of the lowest homogeneity formed just by homogeneous solutions of the GCR equation 
and that each isotypic component is an irreducible module for  the joint action of $G$ and $\cA^G.$

For each irreducible submodule $\la$ of $V$, let
$P(\la)$ be the space of $\la$-valued polynomials on $\bC^6$, and let $P_h(\la)$ be the space of such  polynomials which are homogeneous of degree $h$. The corresponding spaces of harmonic polynomials are
$$\cH(\la)=\{P\in P(\la)|\ \Delta P=0\}\text{\ \ and\ \ }\cH_h(\la)=\cH(\la)\cap P_h(\la).$$ 
Finally, for each $\la\in Spec(\cV)$, we denote by $\Is(\la)$ the isotypic component of $\cV$.

   \subsection{Isotypic component $\Is(m+3,m,0),\	m\geq 0$}

Now we describe Fischer decomposition for the isotypic component $\Is(m+3,m,0)$.
 
\begin{thm}[Fischer decomposition for $\Is(m+3,m,0)$]\label{fisher00}
	Let $m\geq 0.$ 
	Write shortly  $\Is=\Is(m+3,m,0)$. 
	Then the invariant subspaces
		$$
		\boxed{1}=\cH(300)\cap\Is, \;
			\boxed{2}=\cH(311)\cap\Is,\; 
				\boxed{3}=\cH(322)\cap\Is,\; 
					\boxed{4}=\cH(333)\cap\Is 
		$$
	are irreducible under the action of $G.$	
	The Fischer decomposition of the isotypic component $\Is$  is given by
	$$
	\left[\oplus_j\,r^{2j}\boxed{1}\right]\,\oplus\left[\oplus_j\,r^{2j}\boxed{2}\right]\,\oplus
	\left[\oplus_j\,r^{2j}\boxed{3}\right]\oplus 
	\left[\oplus_j\,r^{2j}\boxed{4}\right]. 
	$$
	where all the sums in  brackets  are taken over $j\in\bN_0$.
	All summands are irreducible, their structure is described in Diagram 1.

	Full arrows in the diagram    from line $\la$ to line $\mu$ (resp. their inverses) describe the    action of operators $z^\la_\mu$ (resp. $D^\mu_\la$) .  
	Dotted arrows describe the action of operators     given by multiplication by $r^2$,
	inverse of  dotted arrows denote  the action of $\Delta.$
	All arrows in the diagram and their inverses are isomorphisms. Triangles and squares in the diagram are commutative {(up to a nonvanishing factor).} 
	\end{thm}

	\setlength{\unitlength}{0.20mm}
\begin{picture}(200,340)(-220,-70)
		\put(-100,195){\tiny $m$}
		\put(-70,195){\tiny$m+1$}
		\put(-20,195){\tiny $m+2$}
		\put(30,195){\tiny $m+3$}
		\put(80,195){\tiny $m+4$}
		\put(130,195){\tiny $m+5$}
		\put(180,195){\tiny $m+6$}
		\put(230,195){\tiny $m+7$}
		\put(280,195){\tiny $m+8$}
		\put(-65,125){\small ${z^{300}_{311}}$}
		\put(-15,75){\small ${z^{311}_{322}}$}
		\put(35,25){\small ${z^{322}_{333}}$}
		\put(0,145){$\bullet$}
		\put(100,145){$\bullet$}
		\put(200,145){$\bullet$}
		\put(300,145){$\bullet$}
		\put(50,95){$\bullet$}
		\put(150,95){$\bullet$}
		\put(250,95){$\bullet$}
		\put(100,45){$\bullet$}
		\put(200,45){$\bullet$}
		\put(300,45){$\bullet$}
		\put(250,-5){$\bullet$}
		\put(250,-5){$\bullet$}
		\put(-60,95){$\boxed{2}$}
		\put(-10,45){$\boxed{3}$}
		\put(40,-5){$\boxed{4}$}
		\put(150,-5){$\bullet$} 	
		\put(-105,145){$\boxed{1}$}
		\put(-89,140){\vector(1,-1){32}}
		\put(7,146){\vector(1,-1){47}}
		\put(107,46){\vector(1,-1){47}}
		\put(207,46){\vector(1,-1){47}}
		\put(207,146){\vector(1,-1){47}}
		
		\put(107,146){\vector(1,-1){47}}	
		\put(57,96){\vector(1,-1){47}}	
		\put(157,96){\vector(1,-1){47}}		
		\put(257,96){\vector(1,-1){47}}			
		\put(57,3){\vector(1,1){45}} 	
		\put(7,53){\vector(1,1){45}} 	
		\put(107,53){\vector(1,1){45}} 	
		\put(207,53){\vector(1,1){45}} 	
		\put(157,3){\vector(1,1){45}}
		\put(257,3){\vector(1,1){45}} 		
		
		\put(-43,103){\vector(1,1){45}}
		\put(57,103){\vector(1,1){45}} 	 		
		\put(157,103){\vector(1,1){45}} 	 	
		\put(257,103){\vector(1,1){45}} 	
		\put(50,103){\vector(1,-1){51}} 		
		\put(-50,103){\vector(1,-1){45}} 	 	
		\put(9,42){\vector(1,-1){33}}
		\put(390,145){\small (300)$\;\;\;\;f_{(abc)}$}
		\put(390,95){\small (311)$\;\;\;\;f^{a}_{(bc)}$}
		\put(390,45){\small (322)$\;\;\;\;f^{(ab)}_{c}$}
		\put(390,-5){\small (333)$\;\;\;\;f^{(abc)}$}
		\put(350,145){\ldots}
		\put(350,95){\ldots}
		\put(350,45){\ldots}
		\put(350,-5){\ldots}
		
			\multiput(-85,149)(7,0){13}{$.$}
			\multiput(12,149)(7,0){13}{$.$}
			\multiput(112,149)(7,0){13}{$.$}
			\multiput(212,149)(7,0){13}{$.$}
		
			\multiput(-35,99)(7,0){13}{$.$}
		    \multiput(65,99)(7,0){13}{$.$}
		    \multiput(15,49)(7,0){13}{$.$}
		      \multiput(115,49)(7,0){13}{$.$}
		      	\multiput(212,49)(7,0){13}{$.$}
			  \multiput(165,99)(7,0){13}{$.$}	
				\multiput(165,-01)(7,0){13}{$.$}
					\multiput(65,-01)(7,0){13}{$.$}
					
		\put(43,100){\vector(1,0){7}}			
			\put(143,100){\vector(1,0){7}} 		
			      \put(243,100){\vector(1,0){7}} 	
		\put(93,50){\vector(1,0){7}} 		
			\put(193,50){\vector(1,0){7}}
				\put(293,50){\vector(1,0){7}} 	
		\put(143,0){\vector(1,0){7}} 		
			\put(243,0){\vector(1,0){7}}
		\put(43,100){\vector(1,0){7}}
			\put(143,100){\vector(1,0){7}} 		
				\put(243,100){\vector(1,0){7}} 	
		\put(-07,150){\vector(1,0){7}}
		     \put(93,150){\vector(1,0){7}} 		
		           \put(193,150){\vector(1,0){7}} 	 	
		           	    \put(293,150){\vector(1,0){7}}	
	\end{picture}
	
	\vskip-5mm
	\centerline{Diagram 1. Isotypic component $\Is(m+3,m,0).$}
	
 	 \begin{proof}
(i) First we show that $P_h(300)\cap\Is\not=0$ if and only if $h=m+2j$ for some $j\in\bN_0$. 
To do this we use the harmonic decomposition
$$\Pol_h=\bigoplus_{j=0}^{\lfloor h/2 \rfloor} r^{2j}\cH_{h-2j}$$ 
where 
$\cH_h=\{P\in\Pol_h|\ \Delta P=0\}$ is
the space of harmonic polynomials $P\in\Pol_h$ and $\cH_h\simeq (hh0)$, see Proposition \ref{l81}.
Hence, under the action of $G$, we have an irreducible decomposition 
$$\Pol_h\simeq\bigoplus_{j=0}^{\lfloor h/2 \rfloor} (h-2j,h-2j,0).$$ 
Since $P_h(300)=\Pol_h\otimes (300)$, by Theorem \ref{6.1}, it is easy to see that,
for $h=m+2j$ with $j\in\bN_0$, we have 
$$P_h(300)\cap\Is=r^{2j}\cH_m\boxtimes (300)=r^{2j}\boxed{1},$$	
and that $P_h(300)\cap\Is=0$ otherwise. This is depicted in the first row of Diagram 1. 
Of course, in Diagram~1  columns correspond to the degree $h$ of homogeneity. 	
	
The same proof works  for polynomials with values in $(311)$, $(322)$ and $(333)$, see row 2, 3 and 4 in Diagram 1, respectively.
 Theorem \ref{6.1} also shows that  the modules of type $(m+3,m,0)$ do not appear
 in the spaces of polynomials with values in $(210)$, $(221)$, $(111)$ or $(100).$

 \noindent\smallskip
 (ii) Now we show that the operators $\boxed{1}\rightarrow \boxed{2}$, $\boxed{2}\rightarrow r^2\boxed{1}$, $\boxed{2}\rightarrow \boxed{3}$,
$\boxed{3}\rightarrow r^2\boxed{2}$, $\boxed{3}\rightarrow \boxed{4}$ and $\boxed{4}\rightarrow r^2\boxed{3}$ are non-trivial.
Indeed, the highest weight vector  of  $\boxed{1}$  has the form   $(f_1) _{abc}=(z^{12})^m\delta^1_a\delta^1_b\delta^1_c $, whose images of $\boxed{1}\rightarrow \boxed{2}$, $\boxed{2}\rightarrow \boxed{3}$ and $\boxed{3}\rightarrow \boxed{4}$ are
 \begin{equation*}\begin{split}
 		(  z^{300}_{311}f_1)_{bc}^a &=  z^{d a}(f_1)_{dbc}=(z^{12})^mz^{1a}\delta^1_b\delta^1_c= :  ( f_2)_{bc}^a , \\  
		(z^{311}_{322 }f_2)^{a b}_c&   = z^{d (a}( f_2)_{cd}^{b)}
 		=(z^{12})^mz^{1a}z^{1b}\delta^1_c=:   ( f_3)^{ab}_c ,\\
 		(  z^{322}_{333}f_3)^{a b c} & = z^{d (a}( f_3)^{bc)}_d = (z^{12})^mz^{1a}z^{1b}z^{1c}=:  ( f_4)^{a b c},
 \end{split}\end{equation*}
 respectively. They are obviously nontrivial, $f_2$ and $f_3$ are traceless, and so these three arrows are isomorphisms. 
It is easy to see that $z_{300}^{311}f_2\doteq r^2f_1$, $z_{311}^{322 }f_3\doteq r^2f_2$ and $z_{322}^{333}f_4\doteq r^2f_3$ where $\doteq$ means the equality up to a~non-zero multiple. In particular, we have 
$r^2\doteq z^{311}_{300}z_{311}^{300}$ on $\boxed{1}$, 
$r^2\doteq z^{322}_{311}z^{311}_{322}$ on $\boxed{2}$, 
$r^2\doteq z^{333}_{322}z_{333}^{322}$ on $\boxed{3}$ and 
$r^2\doteq z^{322}_{333}z_{322}^{333}$ on $\boxed{4}$.

Since the operator $r^2$ is injective and commutes with the operators $z^{\la}_{\mu}$ all the arrows in Diagram 1 are non-trivial. Using Proposition \ref{p_duality}
and Remark \ref{r_duality}, we complete the proof. In particular, all the arrows are isomorphisms and inverses to full arrows are given as non-zero multiples of the operators $D^{\mu}_{\la}$.
  \end{proof}


  \subsection{Isotypic component $\Is(m+3,m+1,1),\ m\geq 0$}
This isotypic component has much more complicated structure due to the fact
that for certain values, there are higher multiplicities  in most homogeneities.

\begin{thm}[Fischer decomposition for $Is(m+3,m+1,1)$]\label{fischer11}
	Let $m\geq 0$ and $\Is=\Is(m+3,m+1,1)$. Then the invariant subspaces
	\begin{eqnarray*}
	&\boxed{1}=\cH_m(311)\cap \Is, \; \;
	\boxed{2}=\cH(300)\cap\Is,   \; \;
	\boxed{3}=\cH_{m+1}(322)\cap\Is,   \; \;
	\boxed{4}=\cH(210)\cap\Is,   \; \;\\
&  	\boxed{5}=\cH_{m+2}(311)\cap\Is ,  \; \;  \boxed{7}=\cH(333)\cap\Is, \; \; 
	\boxed{8}=\cH(221)\cap\Is ,  \; \;\\
&	\boxed{9}=\cH_{m+3}(322)\cap\Is , \; \;	
	\boxed{6}=r^2\boxed{1},  \; \;
    \boxed{10}=r^2\boxed{3}
	\end{eqnarray*}
	are irreducible under the action of $G.$	
	The Fischer decomposition of the isotypic component $\Is(m+3,m+1,1)$  is given by
\begin{eqnarray*}
&\left[\oplus_j\,r^{2j}\boxed{2}\right]\,\oplus
\left[\oplus_j\,r^{2j}\boxed{1}\right]\,\oplus
\left[\oplus_j\,r^{2j}\boxed{5}\right]\,\oplus
\left[\oplus_j\,r^{2j}\boxed{3}\right]\,\oplus
\left[\oplus_j\,r^{2j}\boxed{9}\right]\,\oplus
\left[\oplus_j\,r^{2j}\boxed{7}\right]\\
& {\oplus}\left[\oplus_j\,r^{2j}\boxed{4}\right]\,\oplus
\left[\oplus_j\,r^{2j}\boxed{8}\right].
\end{eqnarray*}

	All summands are irreducible, their structure is described in Diagram 2. 

	Full arrows in the diagram describe the action of corresponding operators $z^{\la}_{\mu}$ 
	restricted to the given  module and their inverses correspond to adjoint differential operators $D_{\la}^{\mu}$.
	Dotted arrows describe operators  from $\cA^G$     given by multiplication by $r^2,$
	inverses of  dotted arrows are given by the action of $\Delta.$

 All arrows and their inverses between two multiplicity one components are isomorphisms.
 Every sum of two arrows from two different multiplicity one components to a multiplicity two component is also an isomorphism. The same is true for a sum of their inverses.
 
 The operator $D^-_+$ maps $\boxed{5}$ to $\boxed{2}$ and $\boxed{9}$ to $\boxed{5}$
 and both these maps are nontrivial.

\end{thm}


   \vskip-10mm
\setlength{\unitlength}{0.20mm}
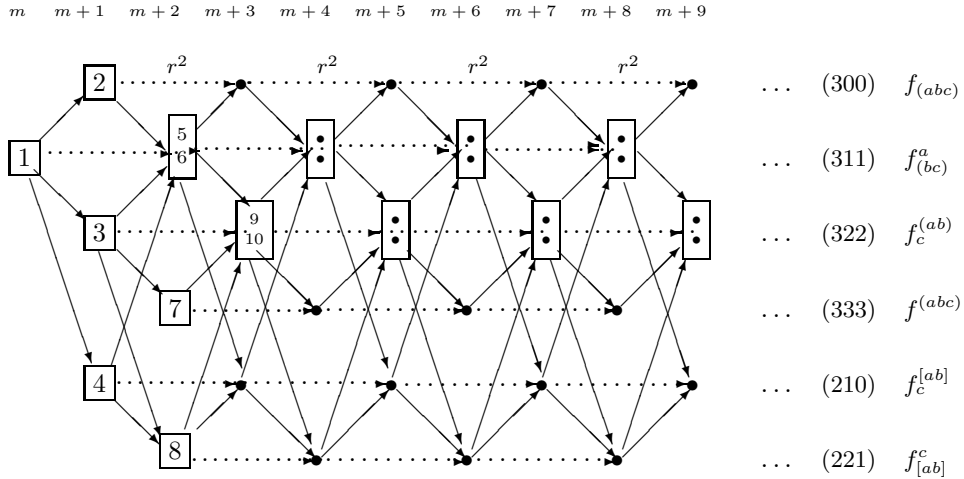
\begin{picture}(200,340)(-220,-170)
	\put(-150,95){\tiny $m$}
	\put(-120,95){\tiny$m+1$}
	\put(-70,95){\tiny $m+2$}
	\put(-20,95){\tiny $m+3$}
	\put(30,95){\tiny $m+4$}
	\put(80,95){\tiny $m+5$}
	\put(130,95){\tiny $m+6$}
	\put(180,95){\tiny $m+7$}
	\put(230,95){\tiny $m+8$}
	\put(280,95){\tiny $m+9$}

	\put(-50,-200){$\boxed{8}$}
	\put(50,-205){$\bullet$}
	\put(150,-205){$\bullet$}
	\put(250,-205){$\bullet$}    	
	\put(-25,-180){\vector(1,1){27}}
	\put(58,-198){\vector(1,1){45}}
	\put(158,-198){\vector(1,1){45}}
	\put(258,-198){\vector(1,1){45}}
	
	\put(-33,-180){\vector(1,3){37}}
	\put(58,-193){\vector(1,3){42}}
	\put(158,-193){\vector(1,3){42}}
	\put(258,-193){\vector(1,3){42}}

	\multiput(-29,-201)(7,0){13}{$.$}
	\put(48,-201){\vector(1,0){5}}
	\multiput(44,-201)(8,0){13}{$.$}
	\put(148,-202){\vector(1,0){5}}
	\multiput(144,-201)(8,0){13}{$.$}
	\put(248,-202){\vector(1,0){5}}

	\put(-100,-155){$\boxed{4}$}           	
	\put(0,-155){$\bullet$}
	\put(100,-155){$\bullet$}
	\put(200,-155){$\bullet$}
	\put(300,-155){$\bullet$}
	\put(8,-153){\vector(1,-1){46}}
	\put(108,-153){\vector(1,-1){46}}
	\put(208,-153){\vector(1,-1){46}}
	
	\put(-80,-160){\vector(1,-1){30}}
	
	\put(-82,-138){\vector(1,3){42}}
	\put(5,-145){\vector(1,3){44}}
	\put(105,-145){\vector(1,3){44}}
	\put(205,-145){\vector(1,3){44}}
	
	\multiput(-79,-150)(7,0){13}{$.$}
	\put(-4,-150){\vector(1,0){5}}
	\multiput(4,-150)(8,0){13}{$.$}
	\put(96,-150){\vector(1,0){5}}
	\multiput(104,-151)(8,0){13}{$.$}
	\put(196,-150){\vector(1,0){5}}
	\multiput(204,-151)(8,0){13}{$.$}
	\put(296,-150){\vector(1,0){5}}

	\put(-50,-105){$\boxed{7}$}
	\put(50,-105){$\bullet$}
	\put(150,-105){$\bullet$}
	\put(250,-105){$\bullet$}
	\put(-32,-88){\vector(1,1){32}}
	\put(57,- 98){\vector(1,1){39}}
	\put(157,- 98){\vector(1,1){39}}
	\put(257,- 98){\vector(1,1){39}}
	
	\multiput(-29,-102)(7,0){13}{$.$}
	\put(48,-101){\vector(1,0){5}}
	\multiput(44,-101)(8,0){13}{$.$}
	\put(148,-101){\vector(1,0){5}}
	\multiput(144,-101)(8,0){13}{$.$}
	\put(248,-101){\vector(1,0){5}}

	\put(-100,-55){$\boxed{3}$}
	\put(-5,-50){ {\tiny $\boxed{\begin{matrix}9\\10\end{matrix}}$}}
	\put(92,-50){ {\tiny $\boxed{\begin{matrix}\bullet\\\bullet\end{matrix}}$}}
	\put(192,-50){ {\tiny $\boxed{\begin{matrix}\bullet\\\bullet\end{matrix}}$}}
	\put(292,-50){ {\tiny $\boxed{\begin{matrix}\bullet\\\bullet\end{matrix}}$}}	
	
	
	\multiput(-79,-49)(7,0){13}{$.$}
    	\put(-5,-49){\vector(1,0){5}}
	\multiput(25,-49)(7,0){13}{$.$}
	      \put(91,-49){\vector(1,0){5}}
	\multiput(118,-49)(7,0){13}{$.$}
	    \put(192,-49){\vector(1,0){5}}
	\multiput(218,-49)(7,0){13}{$.$}
	     \put(292,-49){\vector(1,0){5}}

	\multiput(68,8)(7,0){13}{$.$}
	\put(137,6){\vector(1,0){5}}
	\put(237,6){\vector(1,0){5}}

	\put(15,-61){\vector(1,-1){38}}
	\put(115,-61){\vector(1,-1){38}}
	\put(215,-61){\vector(1,-1){38}}
	\put(-78,-60){\vector(1,-1){30}}
	
	\put(-78,-38){\vector(1,1){33}}      	
	\put(18,-27){\vector(1,1){33}}
	\put(118,-27){\vector(1,1){33}}
	\put(218,-27){\vector(1,1){33}}
	
	\put(-90,-62){\vector(1,-3){41}}
	\put(10,-67){\vector(1,-3){41}}
	\put(110,-67){\vector(1,-3){41}}
	\put(210,-67){\vector(1,-3){41}}

	\put(-150,-5){$\boxed{1}$}
	
	\put(-50,5){ {\scriptsize$\boxed{\begin{matrix}5\\6\end{matrix}}$}}
	\put(42,4){ {\tiny $\boxed{\begin{matrix}\bullet\\\bullet\end{matrix}}$}}
	\put(142,4){ {\tiny $\boxed{\begin{matrix}\bullet\\\bullet\end{matrix}}$}}
	\put(242,4){ {\tiny $\boxed{\begin{matrix}\bullet\\\bullet\end{matrix}}$}}

	\multiput(-125,3)(7,0){13}{$.$}
	\put(-47,3){\vector(1,0){5}}
	\multiput(-28,6)(7,0){13}{$.$}
		\put(-28,5){\vector(1,0){5}}
		\put(40,5){\vector(1,0){5}}
	\put(40,5){\vector(1,0){5}}
	\put(137,5){\vector(1,0){5}}
	\multiput(165,5)(7,0){13}{$.$}
	\put(237,9){\vector(1,0){5}}

	\put(-130,12){\vector(1,1){30}}
	\put(-25,20){\vector(1,1){27}}
	\put(67,13){\vector(1,1){35}}
	\put(167,13){\vector(1,1){35}}
	\put(267,13){\vector(1,1){35}}
	
	\put(-133,-7){\vector(1,-1){34}}
	\put(-26,5){\vector(1,-1){36}}
	\put(68,5){\vector(1,-1){33}}
	\put(168,5){\vector(1,-1){33}}
	\put(268,5){\vector(1,-1){33}}
	\put(-136,-13){\vector(1,-3){41}}
	\put(-36,-16){\vector(1,-3){41}}
	\put(61,-16){\vector(1,-3){41}}
	\put(161,-13){\vector(1,-3){43}}
	\put(259,-16){\vector(1,-3){43}}

	\put(-100,45){$\boxed{2}$}
	\put(-78,38){\vector(1,-1){32}}
	\put(7,49){\vector(1,-1){39}}
	\put(109,47){\vector(1,-1){37}}
	\put(209,47){\vector(1,-1){37}}
	
	\put(0,45){$\bullet$}
	\put(100,45){$\bullet$}
	\put(200,45){$\bullet$}
	\put(300,45){$\bullet$}
	
	\multiput(-77,50)(7,0){13}{$.$}
	\put(-5,50){\vector(1,0){5}}
	\multiput(0,50)(8,0){13}{$.$}
	\put(95,50){\vector(1,0){5}}
	\multiput(100,50)(8,0){13}{$.$}
	\put(195,50){\vector(1,0){5}}
	\multiput(200,50)(8,0){13}{$.$}
	\put(295,50){\vector(1,0){5}}
	
	\put(-45,57){\footnotesize $r^2$}
	\put(55,57){\footnotesize $r^2$}
	\put(155,57){\footnotesize $r^2$}
	\put(255,57){\footnotesize $r^2$}

	\put(390,-205){\small (221)$\;\;\;\;f_{[ab]}^{c}$}
	\put(390,-155){\small (210)$\;\;\;\;f^{[ab]}_{c}$}
	\put(390,45){\small (300)$\;\;\;\;f_{(abc)}$}
	\put(390,-5){\small (311)$\;\;\;\;f^a_{(bc)}$}
	\put(390,-55){\small (322)$\;\;\;\;f^{(ab)}_{c}$}
	\put(390,-105){\small (333)$\;\;\;\;f^{(abc)}$}
	\put(350,-155){\ldots}
	\put(350,-205){\ldots}
	\put(350,-55){\ldots}
	\put(350,-5){\ldots}
	\put(350,45){\ldots}
	\put(350,-105){\ldots}
	
\end{picture}
\vskip15mm

\nopagebreak
\centerline{Diagram 2. Isotypic component $\Is(m+3,m+1,1),\ m\geq 0$}

\begin{proof}
As above, it follows from Theorem  \ref{6.1} that in the given isotypic component and for
the chosen values of fields and fixed homogeneity, the corresponding $G$-module
has multiplicity at most two and the corresponding multiplicity is described by Diagram 2.
 
 The Laplace operator $\Delta$ preserves the isotypic component, hence the multiplicity one components {
 $$\boxed{1},\quad \boxed{2},\quad\boxed{3},\quad\boxed{4},\quad\boxed{7},\quad\boxed{8}$$ belong} to the space $\cH\otimes V$ of harmonic polynomials $P:\bC^6\to V$
 by homogeneity reasons. Irreducible modules $\boxed{6},$ resp $\boxed{10}$, are defined
 as images of $\boxed{1},$ resp. $ \boxed{3}$, under the multiplication by $r^2.$
 
 Since the Laplace operator is surjective the multiplicity two component $\boxed{\begin{matrix}5\\6\end{matrix}}$ is mapped by $\Delta$ onto $\boxed{1},$ hence the kernel of this map (denoted by $\boxed{5}$) is nontrivial by dimensional reasons and does not contain $\boxed{6}.$ So
 $\boxed{\begin{matrix}5\\6\end{matrix}}\simeq\boxed{5}\oplus\boxed{6}.$
 Similarly,  $\boxed{\begin{matrix}9\\10\end{matrix}}\simeq\boxed{9}\oplus\boxed{10}.$

 The highest weight vector  of  $\boxed{1}$  has the form   $(f_1)_{ab}^c=(z^{12})^{m }\delta^1_a\delta_b^1\delta^c_4 $ and its images  in $\boxed{2},\boxed{3}$ and $\boxed{4}$ are
 \begin{equation*}\begin{split}
 		(  z_{300}^{311}f_1)_{a bc}&=z_{d(a}\delta_b^1\delta^1_{c)}\delta^{d } _4(z^{12})^{m }= z_{4(a}\delta_b^1\delta^1_{c)} (z^{12})^{m }=:  ( f_2)_{a bc}, \\ (z_{322}^{311 }f_1)^{a b}_c&= z^{1(a} \delta^{b)} _4  \delta^1_{c }(z^{12})^{m }=:  ( f_3)^{a b}_c,\\
 		(  z_{210}^{311}f_1)^{a b}_c&= z^{1[a} \delta^{b]} _4  \delta^1_{c }(z^{12})^{m }=:  ( f_4)^{a b}_c,
 \end{split}\end{equation*} {respectively}.   They are obviously nontrivial and $f_3$ and $f_4$  are  traceless. Hence the corresponding  three operators are isomorphisms.  
 
 Images of $ f_3$ under operators $\boxed{3}\rightarrow \boxed{7}$, $\boxed{3}\rightarrow \boxed{8}$ are
 \begin{equation*}\begin{split}
 		(z_{333}^{322 }f_3)^{a b c} &=   z^{1(a}z^{|1| b}  \delta^{c)} _4(z^{12})^{m } =:  ( f_7 )^{a b c},\\
 		(  z_{221}^{322}f_3)_{a b}^c&= z_{d[a} \delta^1_{b] }  z^{1(c} \delta^{d)} _4(z^{12})^{m }=:  ( f_8)_{a b}^c,
 \end{split}\end{equation*} respectively, which are both nontrivial. The second one is traceless. So these two operators are isomorphisms.

Using commutativity with $r^2$ we prove that all operators between multiplicity one components
are isomorphisms. By Proposition \ref{p_duality}
and Remark \ref{r_duality}, this is true also for their inverses.
 
 The highest weight vector  of $\boxed{6}$ is $f_{6}:=r^2f_1$. 
 For operator $\boxed{2}\to\boxed{\begin{matrix}5\\6\end{matrix}},$ we  get
 \begin{equation}\label{eq:(311)}
 	3 (z_{311}^{300}f_2)^a_{bc}=3z^{d a}z_{4(b}\delta_c^1\delta_{d)}^1(z^{12})^{m }= -  (f_6)^a_{bc}+2z^{1a}z_{4(b}\delta_{c)}^1(z^{12})^{m } ,
 \end{equation}which is easily seen  traceless by Proposition \ref{l81} (i).
  For operator $\boxed{4}\to\boxed{\begin{matrix}5\\6\end{matrix}},$
  we have 
   \begin{equation}\label{eq:diagram3-f4}\begin{split}
  		2 (  z^{210}_{311}f_4)_{ b c}^a&=2z_{d(b}\delta^1_{c) }  z^{1[d} \delta^{a]} _4(z^{12})^{m }= - (f_6)_{b c}^a-z^{1a}z_{4(b}\delta_{c)}^1(z^{12})^{m } ,
  \end{split}\end{equation}
the result being traceless.
  Similarly, it is direct to check that for operator $\boxed{3}\to\boxed{\begin{matrix}5\\6\end{matrix}},$ it holds
  \begin{equation*}\begin{split}
  		2 (  z^{322}_{311}f_3)_{ b c}^a&=2z_{d(b}\delta^1_{c)}  z^{1(d} \delta^{a)} _{ 4 }(z^{12})^{m }=  -  (f_6)_{ b c}^a+ z^{1 a}z_{4(b}\delta_{c)}^1(z^{12})^{m } 
  \end{split}\end{equation*} with traceless image.
  So we showed that  images of   {highest weight vectors  of}  $\boxed{1}, \boxed{2},\boxed{3}$ and $ \boxed{4}$
  in $\boxed{\begin{matrix}5\\6\end{matrix}}$ are pairwise linearly independent.
  It follows that the sum of operators from any two different multiplicity one components to
  $\boxed{\begin{matrix}5\\6\end{matrix}}$ is isomorphism. 
  The result can be transferred to higher homogeneities by commutation with $r^2.$ By duality, this is true for the corresponding pairs  of
  inverse differential operators.  It is also visible that the map $D^-_+$ is injective on $\boxed{5}.$

  Using  Proposition \ref{l81}  {(i)}, we get
 $ 		\Delta   (f_6)^a_{bc} =(m+3) (f_1)^a_{bc},
		$
 	moreover
 	\begin{equation*}	\Delta \left((z^{12})^{m }z^{1a}z_{4(b}\delta_{c)}^1\right)=(z^{12})^{m }\Delta \left(z^{1a}z_{4(b}\delta_{c)}^1\right)=-(f_1)^a_{bc}.
 \end{equation*}
It implies that
 \begin{equation}\label{eq:diagram3-f5}
 	(f_5)^a_{bc}:= z^{1a}z_{4(b}\delta_{c)}^1 (z^{12})^{m } +\frac 1{m+3}   (f_6)^a_{bc}
 \end{equation}
 is harmonic and also traceless, and so it is the highest weight vector  of   $\boxed{5}$.

 Going to the next level, the highest weight vector  of $\boxed{10}$ is $f_{10}:=r^2f_3$, and
 \begin{equation}\label{eq:(322)}
 	3 (z_{322}^{333}f_7)^ { ab}_c=3z_{d c}z^{1(a}z^{|1| b}  \delta^{d)} _4(z^{12})^{m }= -2(f_{10})^{a b}_c -z^{1(a}z^{  b)1} z_{4c} (z^{12})^{m } ,
 \end{equation} which is obviously traceless.
It is direct to check that \begin{equation}\label{eq:(211-322)}
	6 (z_{322}^{221}f_8)^ { ab}_c=6z_{d[c} \delta^1_{e] }  z^{1(a} \delta^{d } _4 z^{b)e}(z^{12})^{m }= - {4}(f_{10})^{a b}_c+z^{1(a}z^{  b)1} z_{4c} (z^{12})^{m } .
\end{equation}
Explicitly, we have
\begin{equation*} \begin{split}
(6z^{221}_{322} f_8)^{ab}_c=&6(z^{12})^m z_{d[c}\delta_{e]}^1z^{1(a}\delta_4^dz^{b)e} \\
 {=}&\frac{1}{2}(z^{12})^m
	\left[z_{dc}\delta_{e}^1z^{1a}\delta_4^dz^{be}+
	z_{dc}\delta_{e}^1z^{1d}\delta_4^bz^{ae}+
	z_{dc}\delta_{e}^1z^{1b}\delta_4^az^{de}+
	z_{dc}\delta_{e}^1z^{1a}\delta_4^bz^{de}+\right.\\
	&z_{dc}\delta_{e}^1z^{1b}\delta_4^dz^{ae}+
	z_{dc}\delta_{e}^1z^{1d}\delta_4^az^{be}
	-z_{de}\delta_{c}^1z^{1a}\delta_4^dz^{be}
	-z_{de}\delta_{c}^1z^{1d}\delta_4^bz^{ae}\\
	&\left.-z_{de}\delta_{c}^1z^{1b}\delta_4^az^{de}
	-z_{de}\delta_{c}^1z^{1a}\delta_4^bz^{de}
	-z_{de}\delta_{c}^1z^{1b}\delta_4^dz^{ae}
	-z_{de}\delta_{c}^1z^{1d}\delta_4^az^{be}
	\right] \\
 {=}	&\frac{1}{2}(z^{12})^m
	\left[z_{4c}z^{1a}z^{b1}
	-r^2\delta^1_c\delta_4^bz^{a1}
	+r^2\delta_{c}^1z^{1b}\delta_4^a+
	r^2\delta_{c}^1z^{1a}\delta_4^b\right.\\
	&+z_{4c}z^{1b}z^{a1}
	-r^2\delta_{c}^1\delta_4^az^{b1}
	-r^2\delta_{c}^1z^{1a}\delta_4^b
	-r^2\delta_{c}^1z^{1a}\delta_4^b\\
	&\left.-4r^2\delta_{c}^1z^{1b}\delta_4^a
	-4r^2\delta_{c}^1z^{1a}\delta_4^b
	-r^2\delta_{c}^1z^{1b}\delta_4^a
	-r^2\delta_{c}^1z^{1b}\delta_4^a
	\right] \\
 {=}&(z^{12})^m\left[- {4}r^2z^{1(a}\delta^{b)}_4\delta^1_c +z^{1(a}z^{b)1}z_{4c}\right].
	  \end{split} \end{equation*}
Hence the images of $\boxed{3},\boxed{7},\boxed{8}$ in   $\boxed{\begin{matrix}9\\10\end{matrix}}$
 are pairwise linearly independent.
 We can get an explicit formula for the highest weight vector of $\boxed{9}.$ 
 We have by  Proposition \ref{l81}  {(i)},
 $$\Delta (f_{10})^{a b}_c=
 		\Delta\left ((z^{12})^{m }z^{1(a} \delta^{b)} _4  \delta^1_{c }r^2\right) 
 		=(m+4) (f_3)^{a b}_c ,$$
 		and by direct computation
 		$$
 		\Delta \left(z^{1(a}z^{  b)1} z_{4c} (z^{12})^{m }\right)=   (z^{12})^{m }\Delta \left(z^{1(a}z^{  b)1} z_{4c}\right)= (f_3)^{a b}_c.
 $$
So we get that
$$(f_9)^{a b}_c:= -\frac 1{m+4}  (f_{10})^{a b}_c+ z^{1(a}z^{  b)1} z_{4c} (z^{12})^{m }$$
 is harmonic and traceless, hence it is the highest weight vector  of   $\boxed{9}$.  
  Therefore,  $\boxed{\alpha}\oplus\boxed{8}\rightarrow\boxed{\begin{matrix}9\\10\end{matrix}}$ for  $\alpha=3,7$  are isomorphisms.

 The  operator $D^-_+$  commutes with the Laplace operator, hence it 
 maps $\boxed{5}$ to $\boxed{2}$ and $\boxed{9}$ to $\boxed{5}.$ Looking at the corresponding
 highest weight vectors, it is visible that the both maps are nontrivial.
\end{proof}

 \subsection{Isotypic component $\Is(m+3,m+2,2)$, $m\geq 0$.}
  The $G$-modules $(311)$ and $(322)$ are dual to each other and harmonic modules $(m,m,0)$ are self-dual. Hence the decomposition
  and properties
  of this isotypic component is given by  Diagram 3, where dots  (resp. boxes) are
  dual modules to those in Diagram 2 and arrows are
  operators dual to the corresponding operators in Diagram 2.
	The formulation of the theorem on Fischer decomposition of this isotypic component is 
 quite analogous to Theorem \ref{fischer11}. Its proof is again based on computation
 with the highest weight vectors of individual irreducible components and is left 
 as an exercise.

\vskip-13mm
\setlength{\unitlength}{0.20mm}
\begin{picture}(200,340)(-220,-150)
	\put(-150,95){\tiny $m$}
	\put(-120,95){\tiny$m+1$}
	\put(-70,95){\tiny $m+2$}
	\put(-20,95){\tiny $m+3$}
	\put(30,95){\tiny $m+4$}
	\put(80,95){\tiny $m+5$}
	\put(130,95){\tiny $m+6$}
	\put(180,95){\tiny $m+7$}
	\put(230,95){\tiny $m+8$}
	\put(280,95){\tiny $m+9$}

	\put(-50,-200){$\boxed{8}$}
	\put(50,-205){$\bullet$}
	\put(150,-205){$\bullet$}
	\put(250,-205){$\bullet$}    	
	\put(-25,-180){\vector(1,1){27}}
	\put(58,-198){\vector(1,1){45}}
	\put(158,-198){\vector(1,1){45}}
	\put(258,-198){\vector(1,1){45}}
	
	\put(-33,-180){\vector(1,3){37}}
	\put(58,-193){\vector(1,3){42}}
	\put(158,-193){\vector(1,3){42}}
	\put(258,-193){\vector(1,3){42}}

	\multiput(-29,-201)(7,0){13}{$.$}
	\put(48,-201){\vector(1,0){5}}
	\multiput(44,-201)(8,0){13}{$.$}
	\put(148,-202){\vector(1,0){5}}
	\multiput(144,-201)(8,0){13}{$.$}
	\put(248,-202){\vector(1,0){5}}

	\put(-100,-155){$\boxed{4}$}           	
	\put(0,-155){$\bullet$}
	\put(100,-155){$\bullet$}
	\put(200,-155){$\bullet$}
	\put(300,-155){$\bullet$}
	\put(8,-153){\vector(1,-1){46}}
	\put(108,-153){\vector(1,-1){46}}
	\put(208,-153){\vector(1,-1){46}}
	
	\put(-80,-160){\vector(1,-1){30}}
	
	\put(-82,-138){\vector(1,3){42}}
	\put(5,-145){\vector(1,3){44}}
	\put(105,-145){\vector(1,3){44}}
	\put(205,-145){\vector(1,3){44}}
	
	\multiput(-79,-150)(7,0){13}{$.$}
	\put(-4,-150){\vector(1,0){5}}
	\multiput(4,-150)(8,0){13}{$.$}
	\put(96,-150){\vector(1,0){5}}
	\multiput(104,-151)(8,0){13}{$.$}
	\put(196,-150){\vector(1,0){5}}
	\multiput(204,-151)(8,0){13}{$.$}
	\put(296,-150){\vector(1,0){5}}

	\put(-50,-105){$\boxed{7}$}
	\put(50,-105){$\bullet$}
	\put(150,-105){$\bullet$}
	\put(250,-105){$\bullet$}
	\put(-32,-88){\vector(1,1){32}}
	\put(57,- 98){\vector(1,1){39}}
	\put(157,- 98){\vector(1,1){39}}
	\put(257,- 98){\vector(1,1){39}}
	
	\multiput(-29,-102)(7,0){13}{$.$}
	\put(48,-101){\vector(1,0){5}}
	\multiput(44,-101)(8,0){13}{$.$}
	\put(148,-101){\vector(1,0){5}}
	\multiput(144,-101)(8,0){13}{$.$}
	\put(248,-101){\vector(1,0){5}}

	\put(-100,-55){$\boxed{3}$}
	\put(-5,-50){ {\tiny $\boxed{\begin{matrix}9\\10\end{matrix}}$}}
	\put(92,-50){ {\tiny $\boxed{\begin{matrix}\bullet\\\bullet\end{matrix}}$}}
	\put(192,-50){ {\tiny $\boxed{\begin{matrix}\bullet\\\bullet\end{matrix}}$}}
	\put(292,-50){ {\tiny $\boxed{\begin{matrix}\bullet\\\bullet\end{matrix}}$}}	
	
	
	\multiput(-79,-49)(7,0){13}{$.$}
	\put(-5,-49){\vector(1,0){5}}
	\multiput(25,-49)(7,0){13}{$.$}
	\put(91,-49){\vector(1,0){5}}
	\multiput(118,-49)(7,0){13}{$.$}
	\put(192,-49){\vector(1,0){5}}
	\multiput(218,-49)(7,0){13}{$.$}
	\put(292,-49){\vector(1,0){5}}

	\multiput(68,8)(7,0){13}{$.$}
	\put(137,6){\vector(1,0){5}}
	\put(237,6){\vector(1,0){5}}

	\put(15,-61){\vector(1,-1){38}}
	\put(115,-61){\vector(1,-1){38}}
	\put(215,-61){\vector(1,-1){38}}
	\put(-78,-60){\vector(1,-1){30}}
	
	\put(-78,-38){\vector(1,1){33}}      	
	\put(18,-27){\vector(1,1){33}}
	\put(118,-27){\vector(1,1){33}}
	\put(218,-27){\vector(1,1){33}}
	
	\put(-90,-62){\vector(1,-3){41}}
	\put(10,-67){\vector(1,-3){41}}
	\put(110,-67){\vector(1,-3){41}}
	\put(210,-67){\vector(1,-3){41}}

	\put(-150,-5){$\boxed{1}$}
	
	\put(-50,5){ {\scriptsize$\boxed{\begin{matrix}5\\6\end{matrix}}$}}
	\put(42,4){ {\tiny $\boxed{\begin{matrix}\bullet\\\bullet\end{matrix}}$}}
	\put(142,4){ {\tiny $\boxed{\begin{matrix}\bullet\\\bullet\end{matrix}}$}}
	\put(242,4){ {\tiny $\boxed{\begin{matrix}\bullet\\\bullet\end{matrix}}$}}


	\multiput(-125,3)(7,0){13}{$.$}
	\put(-47,3){\vector(1,0){5}}
	\multiput(-28,6)(7,0){13}{$.$}
	\put(-28,5){\vector(1,0){5}}
	\put(40,5){\vector(1,0){5}}
	\put(40,5){\vector(1,0){5}}
	\put(137,5){\vector(1,0){5}}
	\multiput(165,5)(7,0){13}{$.$}
	\put(237,9){\vector(1,0){5}}

	\put(-130,12){\vector(1,1){30}}
	\put(-25,20){\vector(1,1){27}}
	\put(67,13){\vector(1,1){35}}
	\put(167,13){\vector(1,1){35}}
	\put(267,13){\vector(1,1){35}}
	
	\put(-133,-7){\vector(1,-1){34}}
	\put(-26,5){\vector(1,-1){36}}
	\put(68,5){\vector(1,-1){33}}
	\put(168,5){\vector(1,-1){33}}
	\put(268,5){\vector(1,-1){33}}
	\put(-136,-13){\vector(1,-3){41}}
	\put(-36,-16){\vector(1,-3){41}}
	\put(61,-16){\vector(1,-3){41}}
	\put(161,-13){\vector(1,-3){43}}
	\put(259,-16){\vector(1,-3){43}}

	\put(-100,45){$\boxed{2}$}
	\put(-78,38){\vector(1,-1){32}}
	\put(7,49){\vector(1,-1){39}}
	\put(109,47){\vector(1,-1){37}}
	\put(209,47){\vector(1,-1){37}}
	
	\put(0,45){$\bullet$}
	\put(100,45){$\bullet$}
	\put(200,45){$\bullet$}
	\put(300,45){$\bullet$}
	
	\multiput(-77,50)(7,0){13}{$.$}
	\put(-5,50){\vector(1,0){5}}
	\multiput(0,50)(8,0){13}{$.$}
	\put(95,50){\vector(1,0){5}}
	\multiput(100,50)(8,0){13}{$.$}
	\put(195,50){\vector(1,0){5}}
	\multiput(200,50)(8,0){13}{$.$}
	\put(295,50){\vector(1,0){5}}
	
	\put(55,57){\footnotesize $r^2$}
	\put(155,57){\footnotesize $r^2$}
	\put(255,57){\footnotesize $r^2$}
		\put(-45,57){\footnotesize $r^2$}
	
	\put(390,-205){\small (210)$\;\;\;\;f^{[ab]}_c$}
	\put(390,-155){\small (221)$\;\;\;\;f_{[ab]}^c$}
	\put(390,45){\small (333)$\;\;\;\;f^{(abc)}$}
	\put(390,-5){\small (322)$\;\;\;\;f^{(ab)}_{c}$}
	\put(390,-55){\small (311)$\;\;\;\;f^a_{(bc)}$}
	\put(390,-105){\small (300)$\;\;\;\;f_{(abc)}$}
	\put(350,-155){\ldots}
	\put(350,-205){\ldots}
	\put(350,-55){\ldots}
	\put(350,-5){\ldots}
	\put(350,45){\ldots}
	\put(350,-105){\ldots}
	
\end{picture}
\vskip 15mm 

\nopagebreak
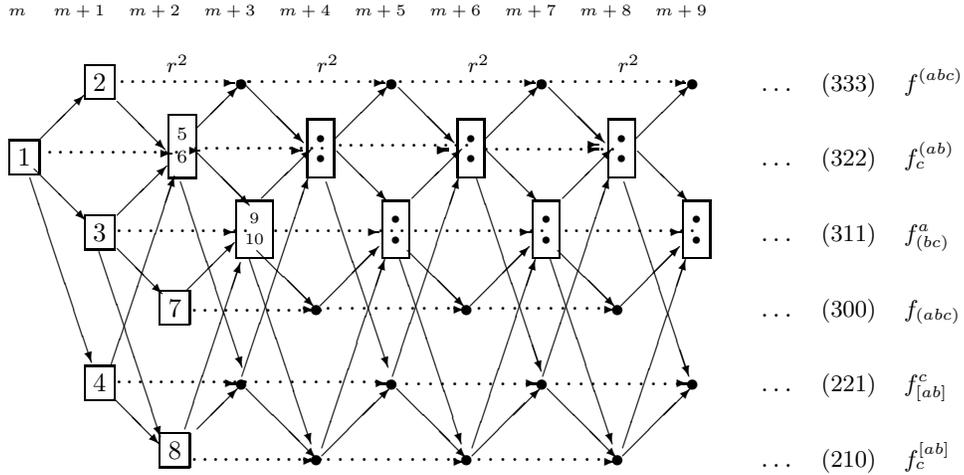
 \centerline{Diagram 3. Isotypic component $\Is(m+3,m+2,2)$, $m\geq 0$
}

\subsection{Isotypic component $\Is(m+3,m+3,3)$, $m\geq 0$.}
Again, modules and operators in Diagram 4 are dual to those in Diagram 1. The formulation of the theorem on Fischer decomposition of this isotypic component is 
 quite analogous to Theorem \ref{fisher00}.

	\	\setlength{\unitlength}{0.17mm}
	\begin{picture}(200,340)(-220,-70)
		\put(-100,195){\tiny $m$}
		\put(-70,195){\tiny$m+1$}
		\put(-20,195){\tiny $m+2$}
		\put(30,195){\tiny $m+3$}
		\put(80,195){\tiny $m+4$}
		\put(130,195){\tiny $m+5$}
		\put(180,195){\tiny $m+6$}
		\put(230,195){\tiny $m+7$}
		\put(280,195){\tiny $m+8$}
		\put(-65,125){\small ${z^{333}_{322}}$}
		\put(-15,75){\small ${z^{322}_{311}}$}
		\put(35,25){\small ${z^{311}_{300}}$}
		\put(0,145){$\bullet$}
		\put(100,145){$\bullet$}
		\put(200,145){$\bullet$}
		\put(300,145){$\bullet$}
		\put(50,95){$\bullet$}
		\put(150,95){$\bullet$}
		\put(250,95){$\bullet$}
		\put(100,45){$\bullet$}
		\put(200,45){$\bullet$}
		\put(300,45){$\bullet$}
		\put(250,-5){$\bullet$}
		\put(250,-5){$\bullet$}
		\put(-60,95){$\boxed{2}$}
		\put(-10,45){$\boxed{3}$}
		\put(40,-5){$\boxed{4}$}
		\put(150,-5){$\bullet$} 	
		\put(-105,145){$\boxed{1}$}
		\put(-89,140){\vector(1,-1){32}}
		\put(7,146){\vector(1,-1){47}}
		\put(107,46){\vector(1,-1){47}}
		\put(207,46){\vector(1,-1){47}}
		\put(207,146){\vector(1,-1){47}}
		
		\put(107,146){\vector(1,-1){47}}	
		\put(57,96){\vector(1,-1){47}}	
		\put(157,96){\vector(1,-1){47}}		
		\put(257,96){\vector(1,-1){47}}			
		\put(57,3){\vector(1,1){45}} 	
		\put(7,53){\vector(1,1){45}} 	
		\put(107,53){\vector(1,1){45}} 	
		\put(207,53){\vector(1,1){45}} 	
		\put(157,3){\vector(1,1){45}}
		\put(257,3){\vector(1,1){45}} 		
		
		\put(-43,103){\vector(1,1){45}}
		\put(57,103){\vector(1,1){45}} 	 		
		\put(157,103){\vector(1,1){45}} 	 	
		\put(257,103){\vector(1,1){45}} 	
		\put(50,103){\vector(1,-1){51}} 		
		\put(-50,103){\vector(1,-1){45}} 	 	
		\put(9,42){\vector(1,-1){33}}
		\put(390,145){\small (333)$\;\;\;\;f^{(abc)}$}
		\put(390,95){\small (322)$\;\;\;\;f^{(ab)}_c$}
		\put(390,45){\small (311)$\;\;\;\;f_{(bc)}^a$}
		\put(390,-5){\small (300)$\;\;\;\;f_{(abc)}$}
		\put(350,145){\ldots}
		\put(350,95){\ldots}
		\put(350,45){\ldots}
		\put(350,-5){\ldots}
		
		\multiput(-85,149)(7,0){13}{$.$}
		\multiput(12,149)(7,0){13}{$.$}
		\multiput(112,149)(7,0){13}{$.$}
		\multiput(212,149)(7,0){13}{$.$}
		
		\multiput(-35,99)(7,0){13}{$.$}
		\multiput(65,99)(7,0){13}{$.$}
		\multiput(15,49)(7,0){13}{$.$}
		\multiput(115,49)(7,0){13}{$.$}
		\multiput(212,49)(7,0){13}{$.$}
		\multiput(165,99)(7,0){13}{$.$}	
		\multiput(165,-01)(7,0){13}{$.$}
		\multiput(65,-01)(7,0){13}{$.$}
		
		\put(43,100){\vector(1,0){7}}			
		\put(143,100){\vector(1,0){7}} 		
		\put(243,100){\vector(1,0){7}} 	
		\put(93,50){\vector(1,0){7}} 		
		\put(193,50){\vector(1,0){7}}
		\put(293,50){\vector(1,0){7}} 	
		\put(143,0){\vector(1,0){7}} 		
		\put(243,0){\vector(1,0){7}}
		\put(43,100){\vector(1,0){7}}
		\put(143,100){\vector(1,0){7}} 		
		\put(243,100){\vector(1,0){7}} 	
		\put(-07,150){\vector(1,0){7}}
		\put(93,150){\vector(1,0){7}} 		
		\put(193,150){\vector(1,0){7}} 	 	
		\put(293,150){\vector(1,0){7}}	
	\end{picture}

\nopagebreak
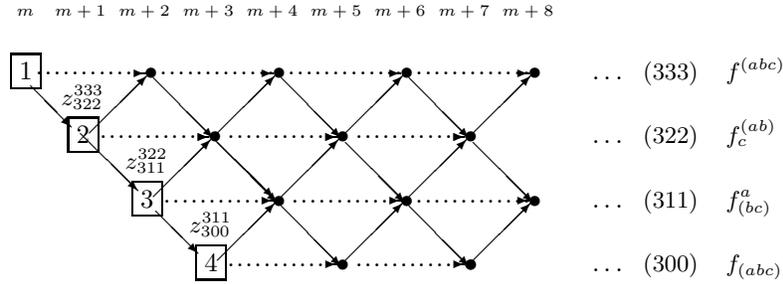
\centerline{Diagram 4. Isotypic component $\Is(m+3,m+3,3)$, $m\geq 0$}

\subsection{Isotypic component $\Is(m+1,m+1,1),\ m\geq 0$}

This case should be subdivided into two separate cases.


{\bf A. The case $m\geq 1, $ i.e., isotypic component $\Is(m+2,m+2,1),\ m\geq 0.$}

The structure of this isotypic component is described by Diagram 5a below
 and is identical with the structure of Diagram 2. 
 The formulation of the theorem on Fischer decomposition of this isotypic component is 
 quite analogous to Theorem \ref{fischer11}. 
\vskip 5mm
\setlength{\unitlength}{0.20mm}
\begin{picture}(200,340)(-220,-150)
	\put(-150,95){\tiny $m$}
	\put(-120,95){\tiny$m+1$}
	\put(-70,95){\tiny $m+2$}
	\put(-20,95){\tiny $m+3$}
	\put(30,95){\tiny $m+4$}
	\put(80,95){\tiny $m+5$}
	\put(130,95){\tiny $m+6$}
	\put(180,95){\tiny $m+7$}
	\put(230,95){\tiny $m+8$}
	\put(280,95){\tiny $m+9$}

	\put(-50,-200){$\boxed{8}$}
	\put(50,-205){$\bullet$}
	\put(150,-205){$\bullet$}
	\put(250,-205){$\bullet$}    	
	\put(-25,-180){\vector(1,1){27}}
	\put(58,-198){\vector(1,1){45}}
	\put(158,-198){\vector(1,1){45}}
	\put(258,-198){\vector(1,1){45}}
	
	\put(-33,-180){\vector(1,3){37}}
	\put(58,-193){\vector(1,3){42}}
	\put(158,-193){\vector(1,3){42}}
	\put(258,-193){\vector(1,3){42}}

	\multiput(-29,-201)(7,0){13}{$.$}
	\put(48,-201){\vector(1,0){5}}
	\multiput(44,-201)(8,0){13}{$.$}
	\put(148,-202){\vector(1,0){5}}
	\multiput(144,-201)(8,0){13}{$.$}
	\put(248,-202){\vector(1,0){5}}

	\put(-100,-155){$\boxed{4}$}           	
	\put(0,-155){$\bullet$}
	\put(100,-155){$\bullet$}
	\put(200,-155){$\bullet$}
	\put(300,-155){$\bullet$}
	\put(8,-153){\vector(1,-1){46}}
	\put(108,-153){\vector(1,-1){46}}
	\put(208,-153){\vector(1,-1){46}}
	
	\put(-80,-160){\vector(1,-1){30}}
	
	\put(-82,-138){\vector(1,3){42}}
	\put(5,-145){\vector(1,3){44}}
	\put(105,-145){\vector(1,3){44}}
	\put(205,-145){\vector(1,3){44}}
	
	\multiput(-79,-150)(7,0){13}{$.$}
	\put(-4,-150){\vector(1,0){5}}
	\multiput(4,-150)(8,0){13}{$.$}
	\put(96,-150){\vector(1,0){5}}
	\multiput(104,-151)(8,0){13}{$.$}
	\put(196,-150){\vector(1,0){5}}
	\multiput(204,-151)(8,0){13}{$.$}
	\put(296,-150){\vector(1,0){5}}

	\put(-50,-105){$\boxed{7}$}
	\put(50,-105){$\bullet$}
	\put(150,-105){$\bullet$}
	\put(250,-105){$\bullet$}
	\put(-32,-88){\vector(1,1){32}}
	\put(57,- 98){\vector(1,1){39}}
	\put(157,- 98){\vector(1,1){39}}
	\put(257,- 98){\vector(1,1){39}}
	
	\multiput(-29,-102)(7,0){13}{$.$}
	\put(48,-101){\vector(1,0){5}}
	\multiput(44,-101)(8,0){13}{$.$}
	\put(148,-101){\vector(1,0){5}}
	\multiput(144,-101)(8,0){13}{$.$}
	\put(248,-101){\vector(1,0){5}}

	\put(-100,-55){$\boxed{3}$}
	\put(-5,-50){ {\tiny $\boxed{\begin{matrix}9\\10\end{matrix}}$}}
	\put(92,-50){ {\tiny $\boxed{\begin{matrix}\bullet\\\bullet\end{matrix}}$}}
	\put(192,-50){ {\tiny $\boxed{\begin{matrix}\bullet\\\bullet\end{matrix}}$}}
	\put(292,-50){ {\tiny $\boxed{\begin{matrix}\bullet\\\bullet\end{matrix}}$}}	
	
	
	\multiput(-79,-49)(7,0){13}{$.$}
	\put(-5,-49){\vector(1,0){5}}
	\multiput(25,-49)(7,0){13}{$.$}
	\put(91,-49){\vector(1,0){5}}
	\multiput(118,-49)(7,0){13}{$.$}
	\put(192,-49){\vector(1,0){5}}
	\multiput(218,-49)(7,0){13}{$.$}
	\put(292,-49){\vector(1,0){5}}

	\multiput(68,8)(7,0){13}{$.$}
	\put(137,6){\vector(1,0){5}}
	\put(237,6){\vector(1,0){5}}

	\put(15,-61){\vector(1,-1){38}}
	\put(115,-61){\vector(1,-1){38}}
	\put(215,-61){\vector(1,-1){38}}
	\put(-78,-60){\vector(1,-1){30}}
	
	\put(-78,-38){\vector(1,1){33}}      	
	\put(18,-27){\vector(1,1){33}}
	\put(118,-27){\vector(1,1){33}}
	\put(218,-27){\vector(1,1){33}}
	
	\put(-90,-62){\vector(1,-3){41}}
	\put(10,-67){\vector(1,-3){41}}
	\put(110,-67){\vector(1,-3){41}}
	\put(210,-67){\vector(1,-3){41}}

	\put(-150,-5){$\boxed{1}$}
	
	\put(-50,5){ {\scriptsize$\boxed{\begin{matrix}5\\6\end{matrix}}$}}
	\put(42,4){ {\tiny $\boxed{\begin{matrix}\bullet\\\bullet\end{matrix}}$}}
	\put(142,4){ {\tiny $\boxed{\begin{matrix}\bullet\\\bullet\end{matrix}}$}}
	\put(242,4){ {\tiny $\boxed{\begin{matrix}\bullet\\\bullet\end{matrix}}$}}


	\multiput(-125,3)(7,0){13}{$.$}
	\put(-47,3){\vector(1,0){5}}
	\multiput(-28,6)(7,0){13}{$.$}
	\put(-28,5){\vector(1,0){5}}
	\put(40,5){\vector(1,0){5}}
	\put(40,5){\vector(1,0){5}}
	\put(137,5){\vector(1,0){5}}
	\multiput(165,5)(7,0){13}{$.$}
	\put(237,9){\vector(1,0){5}}

	\put(-130,12){\vector(1,1){30}}
	\put(-25,20){\vector(1,1){27}}
	\put(67,13){\vector(1,1){35}}
	\put(167,13){\vector(1,1){35}}
	\put(267,13){\vector(1,1){35}}
	
	\put(-133,-7){\vector(1,-1){34}}
	\put(-26,5){\vector(1,-1){36}}
	\put(68,5){\vector(1,-1){33}}
	\put(168,5){\vector(1,-1){33}}
	\put(268,5){\vector(1,-1){33}}
	\put(-136,-13){\vector(1,-3){41}}
	\put(-36,-16){\vector(1,-3){41}}
	\put(61,-16){\vector(1,-3){41}}
	\put(161,-13){\vector(1,-3){43}}
	\put(259,-16){\vector(1,-3){43}}

	\put(-100,45){$\boxed{2}$}
	\put(-78,38){\vector(1,-1){32}}
	\put(7,49){\vector(1,-1){39}}
	\put(109,47){\vector(1,-1){37}}
	\put(209,47){\vector(1,-1){37}}
	
	\put(0,45){$\bullet$}
	\put(100,45){$\bullet$}
	\put(200,45){$\bullet$}
	\put(300,45){$\bullet$}
	
	\multiput(-77,50)(7,0){13}{$.$}
	\put(-5,50){\vector(1,0){5}}
	\multiput(0,50)(8,0){13}{$.$}
	\put(95,50){\vector(1,0){5}}
	\multiput(100,50)(8,0){13}{$.$}
	\put(195,50){\vector(1,0){5}}
	\multiput(200,50)(8,0){13}{$.$}
	\put(295,50){\vector(1,0){5}}
	
	\put(-45,57){\footnotesize $r^2$}
	\put(55,57){\footnotesize $r^2$}
	\put(155,57){\footnotesize $r^2$}
	\put(255,57){\footnotesize $r^2$}

	\put(390,-205){\small (100)$\;\;\;\;f_a$}
	\put(390,-155){\small (111)$\;\;\;\;f^a$}
	\put(390,45){\small (322)$\;\;\;\;f^{(ab)}_c$}
	\put(390,-5){\small (221)$\;\;\;\;f_{[bc]}^a$}
	\put(390,-55){\small (210)$\;\;\;\;f^{[ab]}_c$}
	\put(390,-105){\small (311)$\;\;\;\;f_{(bc)}^a$}
	\put(350,-155){\ldots}
	\put(350,-205){\ldots}
	\put(350,-55){\ldots}
	\put(350,-5){\ldots}
	\put(350,45){\ldots}
	\put(350,-105){\ldots}
	
\end{picture}
 \vskip 15mm
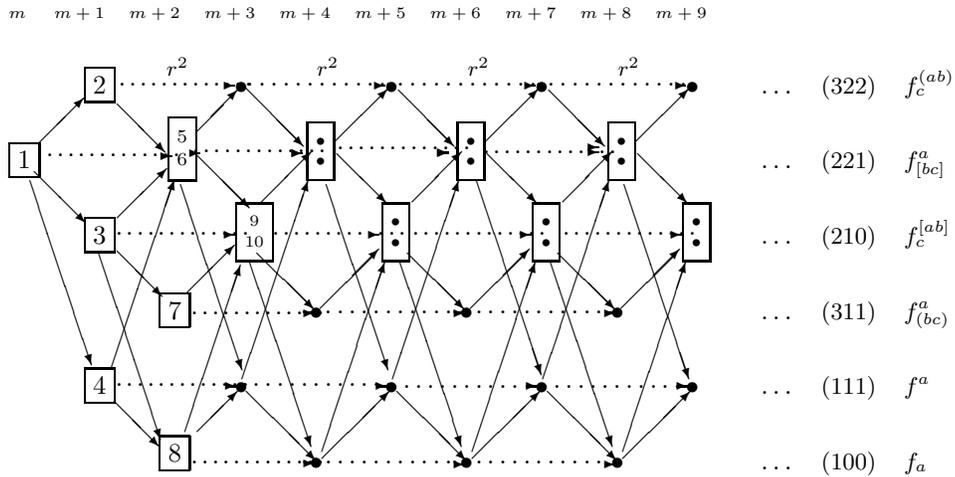
 \centerline{Diagram 5a. Isotypic component $\Is(m+2,m+2,1),\ m\geq 0.$}

 \vskip 5mm

{\bf B. The case $m= 0, $ isotypic component $\Is(111).$}

This isotypic component has a trivial structure.

\vskip 5mm
\setlength{\unitlength}{0.20mm}
\begin{picture}(200,340)(-220,-170)
	\put(-150,95){\tiny $0$}
	\put(-100,95){\tiny$1$}
	\put(-50,95){\tiny $2$}
	\put(-0,95){\tiny $3$}
	\put(50,95){\tiny $4$}
	\put(100,95){\tiny $5$}
	\put(150,95){\tiny $6$}
	\put(200,95){\tiny $7$}
	\put(250,95){\tiny $8$}
	\put(300,95){\tiny $9$}

\put(-80,40){\vector(1,-1){35}}
\put(9,45){\vector(1,-1){43}}
\put(109,45){\vector(1,-1){43}}
\put(209,45){\vector(1,-1){43}}

	\put(-100,45){$\boxed{2}$}
\put(0,45){$\bullet$}
\put(100,45){$\bullet$}
\put(200,45){$\bullet$}
\put(300,45){$\bullet$}

		\put(-150,-5){$\boxed{1}$}
	\put(-50,-5){$\bullet$}
	\put(50,-5){$\bullet$}
	\put(150,-5){$\bullet$}
		\put(250,-5){$\bullet$}
	\multiput(-125,-1)(7,0){10}{$.$}
		\multiput(-39,-1)(7,0){13}{$.$}
			\multiput(61,-1)(7,0){13}{$.$}
				\multiput(161,-1)(7,0){13}{$.$}
	\put(-52,-1){\vector(1,0){5}}
		\put(48,-1){\vector(1,0){5}}
			\put(148,-1){\vector(1,0){5}}
				\put(248,-1){\vector(1,0){5}}

\multiput(-75,49)(7,0){10}{$.$}
		\multiput(11,49)(7,0){13}{$.$}
			\multiput(111,49)(7,0){13}{$.$}
					\multiput(211,49)(7,0){13}{$.$}
				\put(-2,49){\vector(1,0){5}}
				\put(98,49){\vector(1,0){5}}
				\put(198,49){\vector(1,0){5}}
				\put(298,49){\vector(1,0){5}}

	\put(-130,12){\vector(1,1){30}}
\put(-43,2){\vector(1,1){44}}
\put(57,2){\vector(1,1){44}}
\put(157,2){\vector(1,1){44}}
\put(257,2){\vector(1,1){44}}

	\put(-45,57){\footnotesize $r^2$}
	\put(55,57){\footnotesize $r^2$}
	\put(155,57){\footnotesize $r^2$}
	\put(255,57){\footnotesize $r^2$}

	\put(390,45){\small (100)$\;\;\;\;f_a$}
	\put(390,-5){\small (111)$\;\;\;\;f^{a}$}

	\put(350,-5){\ldots}
	\put(350,45){\ldots}
	
\end{picture}

\vskip-25mm
 \centerline{Diagram 5b. Isotypic component $\Is(111).$}

\begin{thm}[Fischer decomposition for $\Is(111)$]\label{fischer111}
	Write shortly  $\Is=\Is(111)$. 
	Then the invariant subspaces
		$
		\boxed{1}=\cH(111)\cap\Is$ and
		$	\boxed{2}=\cH(100)\cap\Is
		$
	are irreducible under the action of $G.$	
	The Fischer decomposition of the isotypic component $\Is$  is given by
	$$
	\left[\oplus_j\,r^{2j}\boxed{1}\right]\,\oplus\left[\oplus_j\,r^{2j}\boxed{2}\right]
	$$
	where all the sums in  brackets  are taken over $j\in\bN_0$.
	All summands are irreducible, their structure is described in Diagram 5b.

	Full arrows in the diagram    from line $\la$ to line $\mu$ (resp. their inverses) describe the    action of operators $z^\la_\mu$ (resp. $D^\mu_\la$) .  
	Dotted arrows describe the action of operators     given by multiplication by $r^2$,
	inverse of  dotted arrows denote  the action of $\Delta.$
	All arrows in the diagram and their inverses are isomorphisms. Triangles and squares in the diagram are commutative {(up to a nonvanishing factor).} 
	\end{thm}

 \subsection{Isotypic component $\Is(m+1,m,0),\ m\geq 0$}
This case should be again subdivided into two separate cases, which are 
dual to the respective cases for isotypic component $\Is(m+1,m+1,1).$

{\bf A. The case $m\geq 1, $i.e., isotypic component $\Is(m+2,m+1,0),\ m\geq 0.$}

The structure of this isotypic component is described by Diagram 6a. The formulation of the theorem on Fischer decomposition of this isotypic component is
 quite analogous to Theorem \ref{fischer11}. 

\vskip 5mm

\setlength{\unitlength}{0.20mm}
\begin{picture}(200,340)(-220,-250)
	\put(-150,95){\tiny $m$}
	\put(-120,95){\tiny$m+1$}
	\put(-70,95){\tiny $m+2$}
	\put(-20,95){\tiny $m+3$}
	\put(30,95){\tiny $m+4$}
	\put(80,95){\tiny $m+5$}
	\put(130,95){\tiny $m+6$}
	\put(180,95){\tiny $m+7$}
	\put(230,95){\tiny $m+8$}
	\put(280,95){\tiny $m+9$}

	\put(-50,-200){$\boxed{8}$}
	\put(50,-205){$\bullet$}
	\put(150,-205){$\bullet$}
	\put(250,-205){$\bullet$}    	
	\put(-25,-180){\vector(1,1){27}}
	\put(58,-198){\vector(1,1){45}}
	\put(158,-198){\vector(1,1){45}}
	\put(258,-198){\vector(1,1){45}}
	
	\put(-33,-180){\vector(1,3){37}}
	\put(58,-193){\vector(1,3){42}}
	\put(158,-193){\vector(1,3){42}}
	\put(258,-193){\vector(1,3){42}}

	\multiput(-29,-201)(7,0){13}{$.$}
	\put(48,-201){\vector(1,0){5}}
	\multiput(44,-201)(8,0){13}{$.$}
	\put(148,-202){\vector(1,0){5}}
	\multiput(144,-201)(8,0){13}{$.$}
	\put(248,-202){\vector(1,0){5}}

	\put(-100,-155){$\boxed{4}$}           	
	\put(0,-155){$\bullet$}
	\put(100,-155){$\bullet$}
	\put(200,-155){$\bullet$}
	\put(300,-155){$\bullet$}
	\put(8,-153){\vector(1,-1){46}}
	\put(108,-153){\vector(1,-1){46}}
	\put(208,-153){\vector(1,-1){46}}
	
	\put(-80,-160){\vector(1,-1){30}}
	
	\put(-82,-138){\vector(1,3){42}}
	\put(5,-145){\vector(1,3){44}}
	\put(105,-145){\vector(1,3){44}}
	\put(205,-145){\vector(1,3){44}}
	
	\multiput(-79,-150)(7,0){13}{$.$}
	\put(-4,-150){\vector(1,0){5}}
	\multiput(4,-150)(8,0){13}{$.$}
	\put(96,-150){\vector(1,0){5}}
	\multiput(104,-151)(8,0){13}{$.$}
	\put(196,-150){\vector(1,0){5}}
	\multiput(204,-151)(8,0){13}{$.$}
	\put(296,-150){\vector(1,0){5}}

	\put(-50,-105){$\boxed{7}$}
	\put(50,-105){$\bullet$}
	\put(150,-105){$\bullet$}
	\put(250,-105){$\bullet$}
	\put(-32,-88){\vector(1,1){32}}
	\put(57,- 98){\vector(1,1){39}}
	\put(157,- 98){\vector(1,1){39}}
	\put(257,- 98){\vector(1,1){39}}
	
	\multiput(-29,-102)(7,0){13}{$.$}
	\put(48,-101){\vector(1,0){5}}
	\multiput(44,-101)(8,0){13}{$.$}
	\put(148,-101){\vector(1,0){5}}
	\multiput(144,-101)(8,0){13}{$.$}
	\put(248,-101){\vector(1,0){5}}

	\put(-100,-55){$\boxed{3}$}
	\put(-5,-50){ {\tiny $\boxed{\begin{matrix}9\\10\end{matrix}}$}}
	\put(92,-50){ {\tiny $\boxed{\begin{matrix}\bullet\\\bullet\end{matrix}}$}}
	\put(192,-50){ {\tiny $\boxed{\begin{matrix}\bullet\\\bullet\end{matrix}}$}}
	\put(292,-50){ {\tiny $\boxed{\begin{matrix}\bullet\\\bullet\end{matrix}}$}}	
	
	
	\multiput(-79,-49)(7,0){13}{$.$}
	\put(-5,-49){\vector(1,0){5}}
	\multiput(25,-49)(7,0){13}{$.$}
	\put(91,-49){\vector(1,0){5}}
	\multiput(118,-49)(7,0){13}{$.$}
	\put(192,-49){\vector(1,0){5}}
	\multiput(218,-49)(7,0){13}{$.$}
	\put(292,-49){\vector(1,0){5}}

	\multiput(68,8)(7,0){13}{$.$}
	\put(137,6){\vector(1,0){5}}
	\put(237,6){\vector(1,0){5}}

	\put(15,-61){\vector(1,-1){38}}
	\put(115,-61){\vector(1,-1){38}}
	\put(215,-61){\vector(1,-1){38}}
	\put(-78,-60){\vector(1,-1){30}}
	
	\put(-78,-38){\vector(1,1){33}}      	
	\put(18,-27){\vector(1,1){33}}
	\put(118,-27){\vector(1,1){33}}
	\put(218,-27){\vector(1,1){33}}
	
	\put(-90,-62){\vector(1,-3){41}}
	\put(10,-67){\vector(1,-3){41}}
	\put(110,-67){\vector(1,-3){41}}
	\put(210,-67){\vector(1,-3){41}}

	\put(-150,-5){$\boxed{1}$}
	
	\put(-50,5){ {\scriptsize$\boxed{\begin{matrix}5\\6\end{matrix}}$}}
	\put(42,4){ {\tiny $\boxed{\begin{matrix}\bullet\\\bullet\end{matrix}}$}}
	\put(142,4){ {\tiny $\boxed{\begin{matrix}\bullet\\\bullet\end{matrix}}$}}
	\put(242,4){ {\tiny $\boxed{\begin{matrix}\bullet\\\bullet\end{matrix}}$}}


	\multiput(-125,3)(7,0){13}{$.$}
	\put(-47,3){\vector(1,0){5}}
	\multiput(-28,6)(7,0){13}{$.$}
	\put(-28,5){\vector(1,0){5}}
	\put(40,5){\vector(1,0){5}}
	\put(40,5){\vector(1,0){5}}
	\put(137,5){\vector(1,0){5}}
	\multiput(165,5)(7,0){13}{$.$}
	\put(237,9){\vector(1,0){5}}

	\put(-130,12){\vector(1,1){30}}
	\put(-25,20){\vector(1,1){27}}
	\put(67,13){\vector(1,1){35}}
	\put(167,13){\vector(1,1){35}}
	\put(267,13){\vector(1,1){35}}
	
	\put(-133,-7){\vector(1,-1){34}}
	\put(-26,5){\vector(1,-1){36}}
	\put(68,5){\vector(1,-1){33}}
	\put(168,5){\vector(1,-1){33}}
	\put(268,5){\vector(1,-1){33}}
	\put(-136,-13){\vector(1,-3){41}}
	\put(-36,-16){\vector(1,-3){41}}
	\put(61,-16){\vector(1,-3){41}}
	\put(161,-13){\vector(1,-3){43}}
	\put(259,-16){\vector(1,-3){43}}

	\put(-100,45){$\boxed{2}$}
	\put(-78,38){\vector(1,-1){32}}
	\put(7,49){\vector(1,-1){39}}
	\put(109,47){\vector(1,-1){37}}
	\put(209,47){\vector(1,-1){37}}
	
	\put(0,45){$\bullet$}
	\put(100,45){$\bullet$}
	\put(200,45){$\bullet$}
	\put(300,45){$\bullet$}
	
	\multiput(-77,50)(7,0){13}{$.$}
	\put(-5,50){\vector(1,0){5}}
	\multiput(0,50)(8,0){13}{$.$}
	\put(95,50){\vector(1,0){5}}
	\multiput(100,50)(8,0){13}{$.$}
	\put(195,50){\vector(1,0){5}}
	\multiput(200,50)(8,0){13}{$.$}
	\put(295,50){\vector(1,0){5}}
	
	\put(-45,57){\footnotesize $r^2$}
	\put(55,57){\footnotesize $r^2$}
	\put(155,57){\footnotesize $r^2$}
	\put(255,57){\footnotesize $r^2$}

	\put(390,-205){\small (111)$\;\;\;\;f^a$}
	\put(390,-155){\small (100)$\;\;\;\;f_a$}
	\put(390,45){\small (311)$\;\;\;\;f_{(bc)}^a$}
	\put(390,-5){\small (210)$\;\;\;\;f^{[ab]}_c$}
	\put(390,-55){\small (221)$\;\;\;\;f_{[bc]}^a$}
	\put(390,-105){\small (322)$\;\;\;\;f^{(ab)}_c$}
	\put(350,-155){\ldots}
	\put(350,-205){\ldots}
	\put(350,-55){\ldots}
	\put(350,-5){\ldots}
	\put(350,45){\ldots}
	\put(350,-105){\ldots}
	
\end{picture}

\nopagebreak
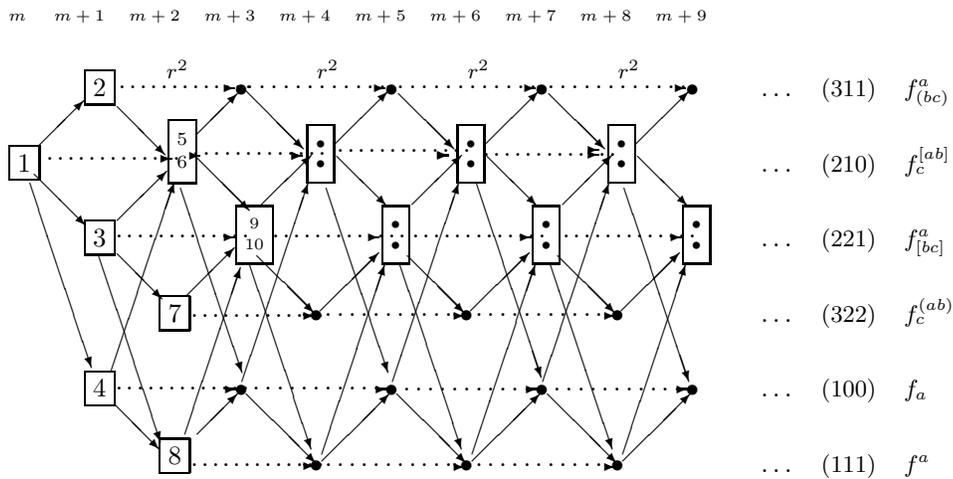
\centerline{Diagram 6a. Isotypic component $\Is(m+2,m+1,0),\ m\geq 0.$}

\vskip5mm

{\bf B. The case $m= 0,$ i.e., isotypic component $\Is(100).$}

The structure of this isotypic component is described by Diagram 6b. The formulation of the theorem on Fischer decomposition of this isotypic component is
 quite analogous to Theorem \ref{fischer111}. 

\vskip 5mm
\setlength{\unitlength}{0.20mm}
\begin{picture}(200,340)(-220,-170)
	\put(-150,95){\tiny $0$}
	\put(-100,95){\tiny$1$}
	\put(-50,95){\tiny $2$}
	\put(-0,95){\tiny $3$}
	\put(50,95){\tiny $4$}
	\put(100,95){\tiny $5$}
	\put(150,95){\tiny $6$}
	\put(200,95){\tiny $7$}
	\put(250,95){\tiny $8$}
	\put(300,95){\tiny $9$}

\put(-80,40){\vector(1,-1){35}}
\put(9,45){\vector(1,-1){43}}
\put(109,45){\vector(1,-1){43}}
\put(209,45){\vector(1,-1){43}}

	\put(-100,45){$\boxed{2}$}
\put(0,45){$\bullet$}
\put(100,45){$\bullet$}
\put(200,45){$\bullet$}
\put(300,45){$\bullet$}

		\put(-150,-5){$\boxed{1}$}
	\put(-50,-5){$\bullet$}
	\put(50,-5){$\bullet$}
	\put(150,-5){$\bullet$}
		\put(250,-5){$\bullet$}
	\multiput(-125,-1)(7,0){10}{$.$}
		\multiput(-39,-1)(7,0){13}{$.$}
			\multiput(61,-1)(7,0){13}{$.$}
				\multiput(161,-1)(7,0){13}{$.$}
	\put(-52,-1){\vector(1,0){5}}
		\put(48,-1){\vector(1,0){5}}
			\put(148,-1){\vector(1,0){5}}
				\put(248,-1){\vector(1,0){5}}

\multiput(-75,49)(7,0){10}{$.$}
		\multiput(11,49)(7,0){13}{$.$}
			\multiput(111,49)(7,0){13}{$.$}
					\multiput(211,49)(7,0){13}{$.$}
				\put(-2,49){\vector(1,0){5}}
				\put(98,49){\vector(1,0){5}}
				\put(198,49){\vector(1,0){5}}
				\put(298,49){\vector(1,0){5}}

	\put(-130,12){\vector(1,1){30}}
\put(-43,2){\vector(1,1){44}}
\put(57,2){\vector(1,1){44}}
\put(157,2){\vector(1,1){44}}
\put(257,2){\vector(1,1){44}}

	\put(-45,57){\footnotesize $r^2$}
	\put(55,57){\footnotesize $r^2$}
	\put(155,57){\footnotesize $r^2$}
	\put(255,57){\footnotesize $r^2$}

	\put(390,45){\small (111)$\;\;\;\;f^{a}$}
	\put(390,-5){\small (100)$\;\;\;\;f_a$}

	\put(350,-5){\ldots}
	\put(350,45){\ldots}
	
\end{picture}
\vskip-25mm

\centerline{Diagram 6b. Isotypic component $\Is(100).$}


\section{Homogeneous solutions of GCR equations}\label{s_homog_sols}

In this section, we prove that homogeneous solutions of the GCR equation which take values in a fixed irreducible spin 3/2 module form an irreducible representation of $SL(4,\bC).$ 
Recall that 
$P(\la)$ is the space of $\la$-valued polynomials on $\bC^6$, and $P_m(\la)$ is the subspace of $m$-homogeneous polynomials $f\in P(\la)$.

\begin{thm}\label{thm:irreducible} Let $m\geq 0.$ Suppose that $\la$ is one of the modules
		\begin{equation}\label{eq:V}
		(3 0 0),\qquad (3 1 1),\qquad  (3 2 2),\qquad (3 3 3),\qquad
		(2 1 0),\qquad (2 2 1).
	\end{equation}
	The GCR operator $\cD^{\la}$ is given in Section \ref{ss_GCR}. Denote by
 $$\cM_m(\la)=\{f\in P_m(\la)|\ \cD^{\la}f=0\}$$ 
the $G$-module of $m$-homogeneous solutions of the GCR equation for $\la$.
Then $\cM_m(\la)$ is irreducible and $\cM_m(\la)\simeq\cH_m\boxtimes\la$. In particular, the $G$-module $\cM_m(\la)$ has the highest weight $(m+\la_1,m+\la_2,\la_3)$
if $\la=(\la_1,\la_2,\la_3)$.
\end{thm}

\begin{proof}  We prove the theorem for $\la=(311)$. It is similar for other cases.
In this case, we have
	\begin{equation}\label{eq:D-311}
	 { \mathcal{D}^{311}}=   D^{311}_{322}\oplus   D^{311}_{210} \oplus    D^{311}_{300}.
	\end{equation}
We show that $\cM_m(\la)$	is isomorphic to the irreducible module $\boxed{1} $ of Diagram~2.

	Fields with values in the module $(311)$ appear in several isotypic components. We shall discuss
	them one after another. For $\mu\in Spec(\cV)$, we put $P_m(\la)_{\mu}=P_m(\la)\cap \Is(\mu)$.
	
	1.\ {Isotypic components $\Is( m+3,m ,0 )$}, see{ Diagram}  1. Note that $ P_h(311)_{ m+3,m ,0 }=0$ unless $h=m+1+2 a $ and $ P_h(300)_{ m+3,m ,0 }=0$  unless $h=m  +2a$, where $a\in \mathbb{N}_0 $. Moreover,
	\begin{equation*}\begin{split}
			P_{m+ 1+2a}(311)_{ m+3,m ,0 }&=r^{2 a  } \boxed{2}   ,\\
			P_{m  +2a}(300)_{ m+3,m ,0 }&=r^{2 a  } \boxed{1} .
	\end{split}\end{equation*}
Since the map
	$  \boxed{1}   \rightarrow \boxed{2}
	$ is an isomorphism by Theorem \ref{fischer11},
	we have the isomorphism
	\begin{equation*}
		z^{300}_{311} :P_h(300)_{ m+3,m ,0 } \rightarrow P_{h+1}(311)_{ m+3,m ,0 }
	\end{equation*}
	for any $h$. Therefore, the kernel of
	\begin{equation*}
		D^{311}_{300}  : P_{h+1}(311)_{ m+3,m ,0 } \rightarrow  P_h(300)_{ m+3,m ,0 }
	\end{equation*}
	is trivial for any $h$.

	2.\ {Isotypic components $\Is( m+3,m+1,1  )$, see{ Diagram}  2}. We know that $ P_h(311)_{ m+3,m+1,1 }=0$ unless $h=m +2 a $ and $ P_h(300)_{ m+3,m+1,1 }=0= P_h(322)_{ m+3,m+1,1 } $ unless $h=m+1 +2a$ for some $a\in \mathbb{N}_0 $. Moreover, we have
	\begin{equation*}\begin{split}
			P_{m+2a}(311)_{ m+3,m+1,1 }&=r^{2(a-1)} \boxed{5}\oplus  { r^{2 a  } \boxed{1} } ,\\ P_{m+1 +2a}(322)_{ m+3,m+1,1 }&=r^{2a} \boxed{3}\oplus r^{2(a-1)} \boxed{9},\\
			P_{m+1 +2a}(300)_{ m+3,m+1,1 }&=r^{2a} \boxed{2}.
	\end{split} \end{equation*}
	Since the map
	$ z^{300}_{311}\oplus
	z^{322}_{311}: \boxed{2}\oplus \boxed{3}  \rightarrow \boxed{5}\oplus  \boxed{6}
	$ is an isomorphism by Theorem \ref{fischer11},
	\begin{equation*}
		z^{300}_{311}\oplus
		z^{322}_{311}:P_h(300)_{ m+3,m+1,1 } \oplus P_h(322)_{ m+3,m+1,1 } \rightarrow P_{h+1}(311)_{ m+3,m+1,1 }
	\end{equation*}
	is surjective for any $h\geq m+1$. Consequently, by Proposition \ref{p_duality},  the kernel of
	\begin{equation*}
		D^{311}_{300} \oplus D^{311}_{322} : P_{h+1}(311)_{ m+3,m+1,1 } \rightarrow  P_h(322)_{ m+3,m+1,1 } \oplus P_h(300)_{ m+3,m+1,1 }
	\end{equation*}
	is $0$ for any $h\geq m+1$. Thus,  there is only one irreducible component $P_{m}(311)_{ m+3,m+1,1 }=\boxed{1} $  in the kernel of $ \mathcal{D}^{311}$.
	
	3.\
	{Isotypic components $\Is( m+3,m+2,2 )$, see{ Diagram}  3}. Note that  $ P_h(311)_{ m+3,m+2,2 }=0$ unless $h=m+1+2(a-1)$ and $ P_{h }(300)_{ m+3,m+2,2 }=0$, $ P_h(322)_{ m+3,m+2,2 }=0$  unless $h=m +2a$ for some $a\in \mathbb{N} $. Moreover, we have
	\begin{equation*}\begin{split}
			P_{m+1+2a}(311)_{m+3,m+2,2 }&=r^{2(a-1)} \boxed{9} \oplus r^{2 (a-1)  } \boxed{10} ,\\
			P_{m +2a}(300)_{ m+3,m+2,2 }&=r^{2(a-1)} \boxed{7}\\
			P_{m +2a}(322)_{ m+3,m+2,2 }&=r^{2(a-1)} \boxed{5}\oplus r^{2(a-1)} \boxed{6}.
	\end{split}\end{equation*}Since the map
	$  \boxed{3}\oplus \boxed{7}  \rightarrow \boxed{9}\oplus  \boxed{10}
	$ is an isomorphism,
	we have the isomorphism
	\begin{equation}\label{eq:300-311-322'} z^{300}_{311}\oplus
		z^{322}_{311}:  \boxed{6}\oplus \boxed{7}  \rightarrow \boxed{9}\oplus  \boxed{10}.
	\end{equation}This is because the map $ \boxed{6}  \rightarrow   \boxed{10}$ is just the map $ \boxed{1}  \rightarrow   \boxed{3}$ shifted by  $r^2$.
	Consequently,
	\begin{equation*}
		z^{300}_{311}\oplus
		z^{322}_{311}:P_h(300)_{ m+3,m+2,2 } \oplus P_h(322)_{ m+3,m+2,2 } \rightarrow P_{h+1}(311)_{m+3,m+2,2 }
	\end{equation*}
	is surjective for any $h$. Therefore, the kernel of
	\begin{equation*}
		D^{311}_{300} \oplus D^{311}_{322} : P_{h+1}(311)_{ m+3,m+2,2 } \rightarrow  P_h(322)_{ m+3,m+2,2 } \oplus P_h(300)_{ m+3,m+2,2}
	\end{equation*}
	is trivial for any $h$.
	
	4.\   {Isotypic components $\Is( m+3,m +3,3 )$, see { Diagram}  4}. We know that $ P_h(311)_{ m+3,m +3,3 }=0$ unless $h=m+2+2a$ and $ P_h(322)_{ m+3,m +3,3 }=0$  unless $h=m +1 +2a$ for some $a\in \mathbb{N}_0 $. Moreover, 
	\begin{equation*}\begin{split}
			P_{m+ 2+2a}(311)_{ m+3,m +3,3 }&=r^{2 a  } \boxed{3}   ,\\
			P_{m +1 +2a}(322)_{ m+3,m +3,3 }&=r^{2 a  } \boxed{2} .
	\end{split}\end{equation*}Since the map
	$  \boxed{2}   \rightarrow \boxed{3}
	$ is an isomorphism by Theorem \ref{fischer11},
	we have the isomorphism
	\begin{equation*}
		z^{322}_{311} :P_h(322)_{ m+3,m +3,3 } \rightarrow P_{h+1}(311)_{ m+3,m +3,3 }
	\end{equation*}
	for any $h$. Therefore, the kernel of
	\begin{equation*}
		D^{311}_{322}  : P_{h+1}(311)_{ m+3,m +3,3 } \rightarrow  P_h(322)_{ m+3,m +3,3 }
	\end{equation*}
	is trivial for any $h$.
	
	5.\   {Isotypic components $\Is(  m+2,m+2,1)$, see   { Diagram}    5}. Note that  $ P_h(311)_{ m+2,m+2,1 }=0$ unless $h=m+2+2a$ and $ P_h(210)_{  m+2,m+2,1 }=0$  unless $h=m  +1+2a$ for some $a\in \mathbb{N}_0 $. Moreover,
	\begin{equation*}\begin{split}
			P_{m+ 2+2a}(311)_{ m+2,m+2,1 }&=r^{2 a  } \boxed{7}   ,\\
			P_{m +1+2a}(210)_{ m+2,m+2,1 }&=r^{2 a  } \boxed{3} .
	\end{split}\end{equation*}Since the map
	$  \boxed{3}   \rightarrow \boxed{7}
	$ is an  isomorphism,  
	\begin{equation*}
		z^{210}_{311} :P_h(210)_{ m+2,m+2,1 }  \rightarrow P_{h+1}(311)_{ m+2,m+2,1 }
	\end{equation*}is surjective
	for any $h$. Therefore, the kernel of the map
	\begin{equation*}
		D^{311}_{210}  : P_{h+1}(311)_{ m+2,m+2,1 } \rightarrow  P_h(210)_{ m+2,m+2,1 }
	\end{equation*}
	is trivial for any $h$.
	
	6.\  {Isotypic components $\Is(  m+2,m+1 ,0 )$, see    { Diagram}   6}. Note that  $ P_h(311)_{ m+2,m+1 ,0 }=0$ unless $h=m+1+2(a-1)$,   while $ P_h(210)_{ m+2,m+1,0 }=0$  unless $h=m  +2a$ for some $a\in \mathbb{N}_0$. In the later case,
	\begin{equation*}\begin{split}
			P_{m+ 2+2a}(311)_{ m+2,m+1 ,0 }&=r^{2 a  } \boxed{2}   ,\\
			P_{m +1+2a}(210)_{ m+2,m+1 ,0 }&=r^{2 a  } \boxed{1} .
	\end{split}\end{equation*}Since the  map
	$  \boxed{1}   \rightarrow \boxed{2}
	$ is an isomorphism, the map  
	\begin{equation*}
		z^{210}_{311} :P_h(210)_{ m+2,m+1 ,0 } \rightarrow P_{h+1}(311)_{ m+2,m+1 ,0 }
	\end{equation*}  is surjective
	for any $h$. Therefore, the kernel of
	\begin{equation*}
		D^{311}_{210}  : P_{h+1}(311)_{ m+2,m+1 ,0 } \rightarrow  P_h(210)_{ m+2,m+1,0 }
	\end{equation*}
	is trivial for any $h$.
	
	To summarize we get that the kernel of the operator   $\mathcal{D}^{311}$  is just the irreducible module $\boxed{1} $ of Diagram~2.
\end{proof}

\subsection{Monomial bases of homogeneous solutions}

In this section, we explain how to construct bases of the spaces $\cM_m(\la)$ of homogeneous solutions.
In \cite{MoYa}, an explicit construction of the so-called monomial bases is given for all finite dimensional irreducible representations of a~reductive Lie algebra. 

We recall this construction for a~given irreducible $G$-module $F_{\la}$ with the highest weight $\la=(\la_1,\la_2,\la_3,0)$. Let $f_{\la}$ be a~highest weight vector in $F_{\la}$. As is well-known we can generate the whole module $F_{\la}$ by applying of lowering operators to $f_{\la}$. 
The Lie algebra of $G=SL(4,\bC)$ is $g=sl(4,\bC)$ and 
the lowering operators are given by the derived action of the matrices  $E^i_j$ on $F_{\la}$ for $1\leq i<j\leq 4$. Here  $E^i_j$ is a~square matrix of size 4 with 1 at place $ij$ and 0 otherwise. Actually, we can say more. Indeed, a~{\it Gelfand-Tsetlin pattern} $\Lambda$ associated with $\la$ is an array of integers 
$$
\begin{matrix}
\la_4^1&& \la_4^2 && \la_4^3 && \la_4^4 \\
 &\la_3^1&& \la_3^2 && \la_3^3 \\
&& \la_2^1 && \la_2^2 \\
&&& \la_1^1
\end{matrix}
$$   
where the first row equals to $\la$ and $\la_k^i\geq \la_{k-1}^i\geq \la_k^{i+1}$ for each $k=2,3,4$ and $i=1,\ldots,k-1$. Then the following result holds true, see \cite[Theorem A]{MoYa}.

 \begin{thm}\label{t_basis}
 Denote by $\{\psi_{\Lambda}\}$ the set of vectors
$$\psi_{\Lambda}=(E^1_2)^{\la^1_2-\la^1_1}(E^1_3)^{\la^1_3-\la^1_2}(E^2_3)^{\la^2_3-\la^2_2}(E^1_4)^{\la^1_4-\la^1_3}(E^2_4)^{\la^2_4-\la^2_3}(E^3_4)^{\la^3_4-\la^3_3} f_{\la}$$
where $\Lambda$ is taken over all Gelfand-Tsetlin patterns associated with $\la$. Then $\{\psi_{\Lambda}\}$ forms a~basis of $G$-module $F_{\la}$.

\end{thm}

It is clear that to construct a~basis for $\cM_m(\la)$ using Theorem \ref{t_basis} we need to describe explicitly the action of lowering operators $E^i_j$ and a~highest weight vector in $\cM_m(\la)$. Since $\cM_m(\la)\simeq\cH_m\boxtimes F_{\la}$ it is sufficient  to do it for the modules $\cH_m$ and $F_{\la}$.
We give such a~description in the following examples. Of course, it would be interesting and useful to express basis elements $\psi_{\Lambda}$ in terms of classical special polynomials, as was done in dimension 4 in \cite{LSW}.

\begin{ex}
On the space $\Pol$ of complex valued polynomials in the variables $z^{ab}$, we have $E^i_j=z^{ja}\nabla_{ia}$ and $(z^{12})^m$ is a~highest weight vector in $\cH_m$.
\end{ex}

\begin{ex}\label{hwv_values}
By \eqref{eq:action}, for $s\in S$, we have $(E^i_j\cdot s)_a=s_i\delta^j_a$ and, for $t\in S^*$, we have $(E^i_j\cdot t)^a=-t^j\delta^a_i$. Thus, highest weight vectors in 
$$ S^a,\ S_a,\ S^{(abc)},\ S_{(abc)},\ \left(S^{(ab)}_{c}\right)_0,\ \left(S^{[ab]}_{c}\right)_0,\ \left(S_{(ab)}^c\right)_0,\ \left (S_{[ab]}^c\right)_0$$ are, respectively,
$$\delta^a_4,\ \delta^1_a,\ \delta^a_4\delta^b_4\delta^c_4,\ \delta^1_a\delta^1_b\delta^1_c,\ \delta^a_4\delta^b_4\delta^1_c,\ 
\delta^{[a}_3\delta^{b]}_4\delta^1_c,\ \delta^1_a\delta^1_b\delta^c_4,\ \delta^1_{[a}\delta^2_{b]}\delta^c_4.
$$
\end{ex}

\begin{ex} For $p\otimes f\in\Pol\otimes F_{\la}$, we have $$E^i_j\cdot(p\otimes f)=(E^i_j\cdot p)\otimes f+p\otimes E^i_j\cdot f$$
Therefore, if $v_{\la}$ is a~highest weight vector in $F_{\la}$ as in Example \ref{hwv_values} above,
then $\cM_m(\la)\simeq\cH_m\boxtimes F_{\la}$ has a~highest weight vector $(z^{12})^m v_{\la}$.
In particular, for $\la=(300)$, we have $v_{\la}=\delta^1_a\delta^1_b\delta^1_c$ and $\cM_m(\la)$ has a~highest weight vector  $f_{\widetilde\la}=(z^{12})^m v_{\la}$ with $\widetilde\la=(m+3,m,0,0)$.
Then a~basis of $\cM_m(\la)$ is formed by the polynomials $\psi_{\Lambda}$ given in Theorem \ref{t_basis} above where $\Lambda$ is taken over all Gelfand-Tsetlin patterns associated with $\widetilde\la$.
\end{ex}

\section{ The Howe dual pair decomposition}
The general duality theorem for group representation is described
in \cite[Chapter 4, Theorem 4.2.1]{GW}.
In our situation, we  are going to prove the following version
of the duality theorem.

\begin{thm}[Howe duality theorem]\label{Howe}
Recall that $G=SL(4,\bC)$ and we consider the $G$-module
	$$V=(300)\oplus(311)\oplus(322)\oplus(333)\oplus(221)\oplus(210)\oplus(110)\oplus(100).$$
	Then the space $\mathcal{V}=\Pol(V)$  decomposes under the action of $G\times \mathcal{A}^G$
	as
	$$
	\cV\simeq \bigoplus_{\la\in Spec(\mathcal{V})} Is(\la)
	$$
	where $Spec(\mathcal{V})$ is given in Theorem 4.1 and $\mathcal{A}^G$ is introduced in Definition 3.2.
	All isotypic components $Is(\la)$ are irreducible modules under the action of $G\times\mathcal{A}^G.$
	In addition, $$Is(\la)\simeq F_\la\otimes L_\la$$ where $F_\la$ is an irreducible $G$-module of the highest weight $\la$ and $L_\la$ is an irreducible module for $\mathcal{A}^G.$

\end{thm}

\begin{rem} (1)	Main features of the theorem are similar as in the   { general duality theorem} treated 
	in \cite[Chapter 4]{GW}.  Our main goal is
	to understand the structure of individual isotypic components, the proof
	 of the Howe duality theorem is then easy to get.
	
\smallskip\noindent
(2) In classical examples of Howe duality, it is possible to show that the (infinite dimensional) algebra of invariant operators acting on isotypic components is the universal enveloping algebra of
a finite-dimensional Lie (super)algebra. In our case, it is not visible how to do that.
A~related duality theorem is presented in \cite[Appendix 1]{CKW} and  in \cite{CW},
where the algebra of invariant operators is the universal enveloping algebra of a (finite dimensional)
Lie superalgebra. In the particular case of spin $3/2$  fields, the vector space $V$ in
$\mathcal{V}=\Pol(V)$
used in the present paper is  replaced in \cite{CKW} by a bigger space $V',$ which is the sum
of the same spinor spaces but with higher multiplicities. It is not visible how to relate this to our results.
Similar effects can be expected for all spins higher or equal to $3/2,$ while for spins $0,$ $1/2$ and $1,$
the results of \cite{CKW} can be used because there are no multiplicities in $V'.$ It hence is also related
to the case of spin $1$ treated in \cite{L}.

\smallskip\noindent
(3)
The definition of the algebra $\mathcal{A}^G$ used in the paper is not the only possibility.
Our choice was motivated by an attempt to make the algebra possibly smaller.
	
\end{rem}	
\begin{proof}   { Since} the  proof is not uniform and it depends on individual isotypic component, we shall treat them one by one.

\vskip 2mm
\noindent 
 {\bf (a) Isotypic component $\Is(m+2,m,0)$} 
  
 The structure of the isotypic component is given in Diagram 1.
 Suppose that $W$ is a subspace in this isotypic component invariant  with respect to the action
 of $G\times\cA^G $ and $f\in W$ is a nontrivial element.  Consider a nonnegative integer $k$
 with the property that $g:=\Delta^k(f)\not=0$ and $\Delta^{k+1}(f)=0.$ Then $g$ is nontrivial and harmonic element in $W.$
 
 There is a nonnegative integer $j$ such that $h:=(D^-_+)^{ { j}}(g)\not=0$ and $(D^-_+)^{ {j+1}}(g)=0.$
 Note that the operator $D^-_+$ maps rows in Diagram 1 to the next higher row.
 Hence the nontrivial element $h\in W$ is harmonic and belongs to $\boxed{1}.$ 
 But $\boxed{1}$ generates the whole isotypic component, which is hence irreducible.  
 
\vskip 4mm
\noindent 
{\bf (b) Isotypic component  
	 $\Is(m+3,m+1,1)$} 
	
	See Diagram 2.
Consider again a  subspace $W$ in this isotypic component invariant under the action of
	$G\times\mathcal{A}^G $ and let $f$ be its nontrivial element.
	Applying a suitable power of the Laplace operator, we can suppose that $f$ is harmonic, i.e.,
	it belongs to the direct sum $\boxed{1}\oplus\boxed{2}\oplus\boxed{3}\oplus\boxed{4}\oplus\boxed{5}\oplus\boxed{7}\oplus\boxed{8}\oplus\boxed{9}.$

	If the component of $f$ in   $\boxed{9}$ is nontrivial, then we apply  
	the operator $D^+_-\circ D^-_+\circ D^-_+$ to get a nontrivial element in $\boxed{1}$
	which generates everything.  
	
	Consider next the case that the component of $f$ in  $\boxed{9}$
	is trivial.
If the component of $f$ in   $\boxed{5}$ is nontrivial, then we apply  
 the operator $D^+_-\circ D^-_+.$ The result is a nontrivial element in  $\boxed{1},$
 because all  components of $f$ different from $\boxed{5}$ are killed.
 	
 So we can suppose that the components of $f$ in  $\boxed{5}$ and   $\boxed{9}$
 are trivial.
 If the component of $f$ in   $\boxed{7}$ is nontrivial, then we apply  
 the operator $D^-_+\circ D^-_+$ and the result is a nontrivial element in  $\boxed{1}.$
 
 So we suppose next  that the components of $f$ in  $\boxed{5},$ $\boxed{7}$  and   $\boxed{9}$
 are trivial. If the component of $f$ in   $\boxed{8}$ is nontrivial, then we apply  
 the operator $D^-_+\circ D^+_-$ to get a nontrivial element in  $\boxed{1}.$ 
 Note that 
 $$
[ z^{322}_{221}z^{311}_{322}+z^{210}_{221}z^{311}_{210}](f_1)=
z_{d[a}\delta^1_{c]}(z^{12})^mz^{1d}\delta^b_ 4
$$
 is nontrivial,
 hence
 $D^{322}_{311}D^{221}_{322}+D^{210}_{311}D^{221}_{210}$ is nontrivial on $\boxed{8}$
 by duality.

 So finally we can suppose that $f$ belongs to $\boxed{1}\oplus \boxed{2}\oplus\boxed{3}\oplus \boxed{4}.$
The operator $D^-_+\circ z^+_-$ maps $f$ to a nontrivial element in  $\boxed{1}\oplus \boxed{2}\oplus\boxed{3}$
and further application of either $D^+_-$ or $D^-_+$ leads to a nontrivial element in $\boxed{1}.$

\vskip 2mm
\noindent 
{\bf (c) Isotypic component  
	$\Is(m+2,m+2,1)$} 
	
	See Diagram 5a.	 	
	 Suppose that $W$ is an invariant subspace in the isotypic component and $f$ is its nontrivial element.
	 As before, we can suppose that $f$ is harmonic, hence it belongs to the sum $$\bigoplus_{j=1,\ldots,9;j\not=6}^{10}\;\;\boxed{j}.$$
	 
	 	 If the component in $\boxed{9}$ is nontrivial,  it can be checked using highest weight vectors that
	the element $D^-_+\circ D^-_+\circ D^-_+(f),$  is nontrivial element in $\boxed{1},$ which generates the whole component.
	 
	 If the component in $\boxed{8}$ is nontrivial, then  $D^{++}\circ D^+_-(f)$ is a nontrivial element in $\boxed{1}.$
	 
	 	 If the component in $\boxed{8}$ is trivial and  the component in $\boxed{4}$ is nontrivial, then  $D^{++}(f)$ is a nontrivial element in $\boxed{1}.$

	If the component in $\boxed{9}$ is trivial  and  the component in $\boxed{5}$ is nontrivial, then  $D^{++}\circ D^{--}(f)$ is a nontrivial element in  $\boxed{1}.$
	 
	 So we are left with the case that $f$ belongs to 
	 $\boxed{1}\oplus\boxed{2}\oplus\boxed{3}\oplus\boxed{7}.$
	 
	 If the component in $\boxed{7}$ is nontrivial, we apply $D^-_+\circ D^+_-$ to get nontrivial element
	 in $\boxed{1}.$

	 If the component in $\boxed{3}$ is nontrivial, then $z^{--}(f)$ has a nontrivial component in $\boxed{8}.$
	 	
	 If the component in $\boxed{3}$ is trivial and the component in $\boxed{2}$ is nontrivial, then $D^-_+(f)$
	  is  nontrivial element in $\boxed{1}.$

\vskip 2mm
\noindent 
{\bf (d) }
The other cases are either similar to ones already treated, or are trivial.
\end{proof}


\appendix

\section{Real forms of $SL(4,\bC)$ and real structures on $\bC^6$}\label{app_Real}	

It is well-known that $Spin(6,\bC) \simeq SL(4,\bC)$ and the basic spinor representations of $Spin(6,\bC)$ correspond to the defining representation $S=\bC^4$ of $G:=SL(4,\bC)$ and its dual $S^*=(\bC^4)^*$. Then the defining representation $\bC^6$ of $Spin(6,\bC)$ is isomorphic to the $G$-module $\Lambda^2 S$. See \cite[23.1]{FH} for more details.
Now we describe real slices $M_{p,q}$ of $\bC^6$ on which the metric has signature $(p,q)$. We treat below Euclidean, Lorentzian and split signature cases, i.e., the cases when $(p,q)=(6,0),$ $(1,5)$, $(3,3)$ and the corresponding real forms of $SL(4,\bC)$ are $G_{p,q}:=SU(4)$, $SL(2,\bH)$, $SL(4,\bR)$, respectively, see \cite[26.1]{FH}.

\medskip\noindent
(i) Euclidean space $\bR^6$: 
For the compact form $SU(4)$ of $SL(4,\bC)$, the spinor space $\bS:=S\oplus S^*$ is of real type, i.e., it admits a real structure $\tau:\bS\to\bS$ defined by
$\tau(s_a,t^a):=(\overline{t^a},\overline{s_a})$
where $\overline{c}$ is the complex conjugate of a complex number $c$.
Thus the complex representation $\bC^6$ admits a real structure 
$$\omega(z^{ab}):=\overline{z_{ab}}=\frac 12 \varepsilon_{abcd}\overline{z^{cd}},\ z^{ab}\in\bC^6.$$  
Then $M_0:=\{z\in\bC^6|\ \omega z=z\}$ is a~real representation of $SU(4)$ which corresponds to the defining representation $\bR^6$ of $Spin(6)$. Indeed, on $M_0$ we take the real coordinates as 
$$(z^{ab})=\frac{1}{\sqrt{2}}
\begin{pmatrix}
0 & ix_1+x_6 & x_4+ix_5 & x_2+ix_3 \\
-ix_1-x_6 & 0 & x_2-ix_3 & -x_4+ix_5 \\
-x_4-ix_5 & -x_2+ix_3 & 0 & -ix_1+x_6 \\
-x_2-ix_3 & x_4-ix_5 & ix_1-x_6 & 0
\end{pmatrix}
$$
and the corresponding derivatives as
$$(\nabla^{ab})=\frac{1}{\sqrt{2}}
\begin{pmatrix}
0 & i\partial_{x_1}+\partial_{x_6} & \partial_{x_4}+i\partial_{x_5} & \partial_{x_2}+i\partial_{x_3} \\
-i\partial_{x_1}-\partial_{x_6} & 0 & \partial_{x_2}-i\partial_{x_3} & -\partial_{x_4}+i\partial_{x_5} \\
-\partial_{x_4}-i\partial_{x_5} & -\partial_{x_2}+i\partial_{x_3} & 0 & -i\partial_{x_1}+\partial_{x_6} \\
-\partial_{x_2}-i\partial_{x_3} & \partial_{x_4}-i\partial_{x_5} & i\partial_{x_1}-\partial_{x_6} & 0
\end{pmatrix}
$$
for $x=(x_1,x_2,x_3,x_4,x_5,x_6)\in\bR^6$. Then on $M_0$ we have
$\overline{z^{ab}}=z_{ab}$ and 
$$r^2=\sum_{j=1}^6 x_j^2,\ \ \Delta=\sum_{j=1}^6 \partial_{x_j}^2\text{\ \ and\ \ } E=\sum_{j=1}^6 x_j \partial_{x_j}.$$

\medskip\noindent
(ii) Lorentzian space $\bR^{1,5}$: The complex representation $S$ has a~quaternionic structure 
$$\tau(s_1,s_2,s_3,s_4)=(\overline{s_2},-\overline{s_1},\overline{s_4},-\overline{s_3}),\ s\in S$$
for the subgroup $SL(2,\mathbb{H})$ of $g\in SL(4,\mathbb{C})$ commuting with $\tau$ (i.e., $\tau(g\cdot s)=g\cdot\tau(s)$ for each $s\in S$).
For $$(C^a_b)=\begin{pmatrix}
0 & -1 & 0 & 0\\
1 & 0 & 0 & 0\\
0 & 0 & 0 & -1\\
0 & 0 & 1 & 0\\
\end{pmatrix}$$
we have $\tau(s_a)=\overline{s_a}C^a_b$ and $SL(2,\mathbb{H})=\{g^a_b\in SL(4,\mathbb{C})|\ \overline{g^c_a}C^a_b=C^c_ag^a_b\}$. The corresponding quaternionic structure on $S^*$ is 
$$\tau^*(t^1,t^2,t^3,t^4)=(-\overline{t^2},\overline{t^1},-\overline{t^4},\overline{t^3}),\ t\in S^*,$$
i.e., $\tau^*(t^b)=(C^{-1})^a_b\;\overline{t^b}$.
Thus the complex representation $\bC^6$ admits a real structure 
$$\omega(z^{cd}):=(C^{-1})^a_c(C^{-1})^b_d\;\overline{z^{cd}},\ z^{cd}\in\bC^6.$$  
Then $M_0:=\{z\in\bC^6|\ \omega z=z\}$ is a~real representation of $SL(2,\mathbb{H})$ which corresponds to the defining representation $\bR^{1,5}$ of $Spin(1,5)$. 
Indeed, we have that $z\in M_0$ iff $z^{12}=\overline{z^{12}}$, $z^{34}=\overline{z^{34}}$, $z^{13}=\overline{z^{24}}$, $z^{14}=-\overline{z^{23}}$. On $M_0$ we take the real coordinates 
$$z^{12}=\frac{1}{\sqrt{2}}(x_0+x_5),\ z^{34}=\frac{1}{\sqrt{2}}(x_0-x_5),\ z^{13}=\frac{1}{\sqrt{2}}(x_1+ix_2),$$ 
$$z^{24}=\frac{1}{\sqrt{2}}(x_1-ix_2),\ z^{14}=\frac{1}{\sqrt{2}}(x_3+ix_4),\ z^{23}=\frac{1}{\sqrt{2}}(-x_3+ix_4)$$ for $x=(x_0,x_1,x_2,x_3,x_4,x_5)\in\bR^6$. Then on $M_0$ we have
$$r^2=x_0^2-\sum_{j=1}^5x_j^2.$$

\medskip\noindent
(iii) The split signature case $\bR^{3,3}$: For the split form $SL(4,\bR)$ of $SL(4,\bC)$, the complex representations $S$, $S^*$ and $\bC^6$ admit real structures $\tau(s_a)=\overline{s_a}$, $\tau^*(t^b)=\overline{t^b}$ and $\omega(z^{cd}):=\overline{z^{cd}}$, respectively. Then $M_0:=\{z\in\bC^6|\ \omega z=z\}$ is a~real representation of $SL(4,\mathbb{R})$ which corresponds to the defining representation $\bR^{3,3}$ of $Spin(3,3)$. Indeed, on $M_0$ we take the real coordinates 
$$z^{12}=\frac{1}{\sqrt{2}}(x_1+x_6),\ z^{34}=\frac{1}{\sqrt{2}}(x_1-x_6),\ z^{13}=\frac{1}{\sqrt{2}}(x_2+x_5),$$ 
$$z^{24}=\frac{1}{\sqrt{2}}(-x_2+x_5),\ z^{14}=\frac{1}{\sqrt{2}}(x_3+x_4),\ z^{23}=\frac{1}{\sqrt{2}}(x_3-x_4)$$ for $x=(x_1,x_2,x_3,x_4,x_5,x_6)\in\bR^6$. Then on $M_0$ we have
$$r^2=x_1^2+x_2^2+x_3^2-x_4^2-x_5^2-x_6^2.$$


\section{ { Tensor product decompositions}} \label{app_LRR}

In this appendix, we first recall the standard Littlewood-Richardson (LR) rules for  $GL(n,\bC)$
 and then apply them for a description of tensor products of irreducible representations
 of $SL(4,\bC). $ 
 We shall consider only complex and finite dimensional representations. For more details, see e.g.\ \cite{HLee}.

All (complex and finite dimensional) irreducible representations of $GL(n,\bC)$
are parametrized by their highest weights which have a form  $$\la=(\la_1,\ldots,\la_n),\ \la_i\in\bZ,\ \la_1\geq\la_2\geq\cdots\geq\la_n.$$
We denote by $F_{\la}$ a~unique (up to an isomorphism) irreducible representation of $GL(n,\bC)$ of the highest weight $\la$.
For example, for $j\in\bZ$, $A\in GL(n,\bC)\rightarrow (\det A)^j\in \bC$ is one dimensional representation of the highest weight $j1_n=(j,\ldots,j)$ with $n$ components. 
If $\la$ is a partition (i.e., $\la_n\geq 0$) we call $F_{\la}$ a~{\it polynomial} representation of $GL(n,\bC)$.
It is easy to see that $F_{\la}\simeq F_{\underline{\la}}\otimes F_{j1_n}$ for some partition $\underline{\la}$ and $j\in\bZ$, e.g.,
$\underline{\la}=(\la_1-\la_n, \ldots, \la_{n-1}-\la_n, 0)$ and $j=\la_n$. Thus any representation $F_\la$ is the tensor product of a~polynomial representation and some integer power of $\det$.

To state LR rules for tensor products of irreducible polynomial representations of $GL(n,\bC)$ we need to introduce more notation.
Any partition $\la$ corresponds uniquely to the Young diagram with  at most $n$ rows and $\la_i$ boxes in the $i$-th row.
For a Young diagram $\la$, its weight $| \la|$ is the number of its boxes and $\ell(\la)$ is the number of nontrivial rows.

The inclusion
 $ \la\subset   {\nu}$ means that the Young diagram $ \la $ is contained in $  {\nu}.$
For $ \la\subset   {\nu}$, we denote by $  {\nu}/ \la$ the so-called { \it skew diagram} we get from the Young diagram $  {\nu}$ by removing $ \la.$
The shape of $ {\nu}/\la$  is $ {\nu}-\la.$

A  {\it skew tableaux} $T= {\nu}/\la$ is obtained by filling a skew  diagram $  {\nu}/ \la$  with positive integers.
 If $m_i$ is the number of times  {of $i$ appearing} in $T,$ we call $m=( {m_1},\ldots,m_r)$ the  {weight} of $T.$
 The word $w(T)$  of $T$  is the sequence obtained by reading the entries of $T$ from right to left in each row, from
 the first row to the  last row.

 A skew tableaux $T$ is {\it semi-standard}, if the entries of $T$ weakly increase across rows (from left to right)  and strongly increase down columns.
  We say that $T$ satisfies the {\it Yamanouchi word condition} ({\it Y-condition}), if the number of occurences of an integer
  $i$ never exceeds the number of occurences of $i-1$ for any initial segment of $w(T).$

  A  {\it LR tableau} is a semistandard skew tableau $T$ satisfying the Yamanouchi word condition.

\begin{thm}[Littlewood-Richardson rules for $GL(n,\bC)$] Let $\la$ and $\mu$ be partitions with $\ell(\la)\leq n$ and $\ell(\mu)\leq n$.
 Then we have
	$$F_\la\otimes F_\mu=\sum_\nu c^\nu_{ \la \mu}F_\nu,$$
	where $c^\nu_{ \la \mu}$ is the number of LR tableaux of shape $\nu/\la$  and weight $ \mu.$
	The sum is taken over all partitions $\nu$ with $\ell(\nu)\leq n$ and $|\nu|=|\la|+|\mu|.$	
\end{thm}

Representations of $GL(n,\bC)$ and $SL(n,\bC)$ are closely related. Indeed, every irreducible representation $F_\la$ of $GL(n,\bC)$ remains irreducible as an $SL(n,\bC)$-module. Since $\det $ is trivial as  an $SL(n,\bC)$-module $F_\la $ and $F_{\la'}$ are isomorphic as $SL(n,\bC)$-modules if and only if $\la-\la'$
 is an integer multiple of $1_n$. On the other hand,
every irreducible (complex and finite dimensional) representation $F$ of $SL(n,\bC)$ is isomorphic to $F_\la$ restricted to $SL(n,\bC)$ for some partition $\la$ with
 $\ell(\la)\leq n.$ It is clear that $SL(n,\bC)$-module $F$ is uniquely determined by the $(n-1)$-tuple $\underline{\la}=(\la_1-\la_n,\la_2-\la_n, \ldots, \la_{n-1}-\la_n)$.


Now we apply LR rules for a description of tensor products of irreducible representations  of $SL(4,\bC).$ 
For an irreducible module $F_\la$ we often write shortly just $\la$.

\begin{lem} For $SL(4,\bC)$-modules, we have
 \begin{equation*}\begin{split}
	(300)\otimes(110) { =}&(410)\oplus(311),\\
	 	(300)\otimes(220) { =}&(520)\oplus(421)\oplus(322),\\
	 (300)\otimes(mm0) { =}&(m+3,m,0)\oplus(m+2,m,1)\oplus(m+1,m,2)\oplus(m,m,3),\quad  {m\geq 3}.
	 \end{split} \end{equation*}
\end{lem}
\begin{proof} By LR rules, we have to discuss all possible skew diagrams $\nu/(3,0,0).$ For each choice, we have to find   all its  possible fillings leading to an LR tableaux (i.e. all fillings  satisfying
the corresponding conditions). The number of such possibilities is the multiplicity of $\nu$
in the decomposition.

The skew diagram $\nu/(3,0,0)$ should contain $2m$ boxes filled by $m$ copies of $1$ and $m$ copies of $2.$
Let us denote by $i_k$ the row $(i,\ldots,i)$ with $k$ copies of the number $i.$
The skew diagram has at most $4$ rows.  Its shape is given by  $n=(n_1,n_2,n_3,n_4),$
where $n_i$ denotes the number of boxes in the $i$-th row. Hence $\sum_in_i=2m.$

Let us now discuss in turn all possibilities of shapes $n,$ which have at least one possible
fillings satisfying all conditions for LR tableau.

\vskip2mm\noindent
Due to Y-condition, the first row contains only copies of $1.$

(a) Let $n_1= {m}.$ Then the simplest possibility is $n_2=m, n_3=n_4=0.$ It leads to  the  LR tableau
described by the matrix
$\begin{pmatrix}
	.\ .\ .	\ 1_m\\2_m\; \;\;\;\;\;\; \\
\end{pmatrix}.	
$
The corresponding word $w=1_m2_m$ satisfies the Y-condition and $\nu=(m+3,m,0).$
The result corresponds to the Cartan piece in the decomposition.

Suppose now that $n_2<m,$ hence $n_3>0.$ Then number $2$ appears both in the second and the third row, which violates the semi-standard condition.

\vskip2mm\noindent
(b) The next possibility is $n_1=m-1.$ The semistandard condition for fillings tells us that
the last copy of $1$ should appear at the first position of the second row.
Then we can consider $n=(m-1,m,1)$  and the filling given by
$\begin{pmatrix}
	.\ .\ .	\ 1_{m-1}\\1\,2_{m-1}\; \;\;\;\;\\ 2\; \;\;\;\;\;\;  \; \;\;\;\;\;\;
\end{pmatrix}.	
$
The corresponding word $w=1_{m-1}2_{m-1}12$ satisfies the   {\it Y-condition}  and  it is
clearly the only choice.
The result corresponds to $\nu=( {m+2},m,1)$ in the statement of Lemma and its multiplicity is one.

\vskip2mm\noindent
(c)  The next possibility is $n_1=m-2.$
The semistandard condition for fillings tells us that
the last two copies of $1$ should appear at the first two positions of the second row.
We have only two possible numbers available for fillings, hence the semi-standard condition  tells us that $n_4=0.$ So the only choice left now  is
  $n=(m-2,m,2)$  and the filling given by
$\begin{pmatrix}
	.\ .\ .	\ 1_{m-2}\\1\,1\,2_{m-2}\; \;\; \\2\, 2\; \;\;\;\;\;\;  \; \;\;\;\;
\end{pmatrix}.	
$
The corresponding word $w=1_{m-2}2_{m-2}1122$ satisfies the Y-condition.
The result corresponds to  $\nu=( {m+1},m,2)$  in the statement of Lemma and its multiplicity is one.

\vskip2mm\noindent
(d)
 The next possibility is $n_1=m-3.$
The semistandard condition for filllings tells  us again that
the left three copies of $1$ should appear at the first three  positions of the second row.
  As in c), the only choice left is
$n=(m-3,m,3)$  and the filling given by
$\begin{pmatrix}
	.\ .\ .	\ 1_{m-3}\\1\,1\,1\,2_{m-3}\; \\2\, 2\,2\; \;\;\;\;\;\;  \; \;\;
\end{pmatrix}.	
$
The corresponding word $w=1_{m-3}2_{m-3}111222$ satisfies the Y-condition.
The result corresponds to $\nu=( {m},m,3)$ in the statement of Lemma and its multiplicity is one.

\vskip2mm\noindent
(e)   The next possibility could be $n_1>m-3. $ But we cannot have more than three copies of $1$ in the
second row due to semi-standard condition, hence there is now possible fillings for such shapes and their multiplicity is zero.

\vskip2mm\noindent
(f) The case $m=1$ and $m=2$ is proved in the same way showing that some possibilities
of the general case are excluded.   \end{proof}

\begin{rem} {
It is possible e.g.\ to use a computer algebra system LiE (see \cite{LiE}) to prove this and the following  lemmas.}
\end{rem}

\begin{lem} For $SL(4,\bC)$-modules, we have
\begin{equation*}\begin{split}
	 (311)\otimes(110) { =}& (4210)\oplus(3220)\oplus(3211)\oplus(4111),\\
(311)\otimes(220) { =}& (5310)\oplus(4320)\oplus(3330)\oplus(4311)\oplus(3321)\oplus(5211),\\
(311)\otimes(m,m,0) { =}&(m+3,m+1,1,0)\oplus(m+2,m+1,2,0)\oplus(m+1,m+1,3,0)
		\oplus(m+2,m+1,1,1) \\
		&\oplus(m+1,m+1,2,1)
		\oplus(m+3,m,1,1)\oplus(m+2,m,2,1)\oplus(m+1,m,3,1),\quad m\geq 3.
	 \end{split} \end{equation*}
\end{lem}\begin{proof}We are going to discuss all possible skew diagrams $\nu/(3,1,1)$  and its LR fillings containing
$2m$ new boxes
  filled by $m$ copies of $1$ and $m$ copies of $2.$
Let us again denote by $i_k$ the row $(i,\ldots,i)$ with $k$ copies of the number $i.$
The skew diagram has at most $4$ rows.  Its shape is given by  $n=(n_1,n_2,n_3,n_4),$
where $n_i$ denotes the number of boxes in the $i$-th row and $\sum_in_i=2m.$

Let us now discuss in turn all possibilities of shapes $n,$ which have at least one possible
fillings satisfying all conditions for LR tableau.

\vskip2mm\noindent
(a) Let $n_1= { m}.$ Then the simplest possibility is $n_2=m, n_3=n_4=0.$ It leads to  the  LR tableau
described by the matrix

$\begin{pmatrix}
	.\ .\ .	\ 1_m\\ \,.\; 2_m\;\;\;\;\;\\ \;\; .\ \;\;\;\;\;\;\;\;\;\;
\end{pmatrix}.	
$
The corresponding word $w=1_m2_m$ satisfies the Y-condition and $\nu=(m+3,m+1,1,0).$
The result corresponds to the Cartan piece in the decomposition.

\vskip2mm\noindent
(b) The next possibility is $1_m$ in the first row and $2_{m-1}$ in the second row.
The last copy of $2$ cannot be in the third row  {by the semi-standard condition,} but it can be in the fourth row. The whole LR tableau is

$\begin{pmatrix}
	.\ .\ .	\ 1_m\\ \,.\; 2_{m-1}\;\;\\ \;\; .\ \;\;\;\;\;\;\;\;\;\;\,\,\\2\;\;\;\;\;\;\;\;\;\;
\end{pmatrix}.	
$ The Y-condition is satisfied and   $\nu=(m+3,m,1,1).$

\vskip2mm\noindent
(c) The next possibility with $n_1=m$ and $n_2=m-2$ does not admit any LR filling, because
 there is not enough place for two copies of $2$ in the third and the fourth row.
 So we start to consider possibilities with $1_{m-1}$ in the first row. Then due to Y-condition,
 the last copy of $1$ must appear in the first position in the second row or  at the fourth row.
 In the first case, the next possiblity is to have $m-1$ copies of $2$ in the second row, which
 leads to two new options corresponding to the LR tableaux

    $\begin{pmatrix}
    	.\ .\ .	\ 1_{m-1}\\ \,. \,1\; 2_{m-1}\;\;\\ \;\; .\ \;\;\;\;\;\;\;\;\;\;\;\;\;\;\\2\;\;\;\;\;\;\;\;\;\;\;\;
    \end{pmatrix},
 $
    $\begin{pmatrix}
 	.\ .\ .	\ 1_{m-1}\\ \,. \;1\; 2_{m-1}\;\;\\ \;.\; 2\ \;\;\;\;\;\;\;\;\;\;\,
 \end{pmatrix},
 $
 Y-conditions hold  { and  $\nu=(m+2,m+1,2)$, $(m+2,m+1,1,1),$}
 \\or to have $m-2$ copies of $2$ in the second row, which leads to a unique possibility

      $\begin{pmatrix}
  	.\ .\ .	\ 1_{m-1}\\ \,. \,1\; 2_{m-2}\;\;\\ \;\; .\,2\ \;\;\;\;\;\;\;\;\;\;\;\\2\;\;\;\;\;\;\;\;\;\;\;
  \end{pmatrix},
 $
 Y-condition holds and $\nu=(m+2,m,2,1).$

\vskip2mm\noindent
(d) The next possibility is to have $1_{m-2}$ in the first row. Then both missing copies of $1$
are necessarily  in the second row    { by the semi-standard condition}   as the first two items  (by Y-condition).
So the content of the second row must be $1\,1\,2_{m-2}$ or $1\,1\,2_{m-3}.$
The first case leads to two possibilities

 $
\begin{pmatrix}
	.\ .\ .	\ 1_{m-2}\\ \,. \;1\,1\; 2_{m-2}\;\;\\ \;.\; 2\,2\ \;\;\;\;\;\;\;\;\;\;\,
\end{pmatrix},
$
 $
\begin{pmatrix}
	.\ .\ .	\ 1_{m-2}\\ \,. \;1\,1\; 2_{m-2}\;\;\\ \;.\; \;\;\; \; \;\;\;\;\;\;\;\;\;\;\,\\ 2\ \;\;\;\;\;\;\;\;\;\;\;\;\;
\end{pmatrix}
$. The
Y-conditions hold, $\nu=(m+1,m+1,3,0)$ and $\nu=(m+1,m+1,2,1).$
 { The second case leads to
 $
\begin{pmatrix}
	.\ .\ .	\ 1_{m-2}\;\;\\ \,. \;1\,1\; 2_{m-3}\;\; \\ .\,2\,2\ \;\;\;\;\;\;\;\;\;
	\\2\;\; \; \;\;\;\;\;\;\;\;\;\;\;\,
\end{pmatrix}
$. The Y-condition  holds and $\nu=(m+1,m ,3,1) .$}
\vskip2mm\noindent
(e)
Finally let us try $1_{m-3}$ in the first row. Then two left copies of $1$ must be placed into
the second row and the last copy to the fourth row (by semi-standard condition).
The result is

 $
\begin{pmatrix}
	.\ .\ .	\ 1_{m-2}\;\;\\ \,. \;1\,1\; 2_{m-3}\;\; \\ .\,2\,2\ \;\;\;\;\;\;\;\;\;
	\\1\;\; \; \;\;\;\;\;\;\;\;\;\;\;\,
\end{pmatrix}
$

It is impossible to have only $m-4$ copies of $1$ in the first row, because then only two copies
of $1$ can be placed at the beginning of the first row and there is a~space only for $m-3$ copies of $2$ in the second row and we are left with two copies of $1$ and $3$ copies of $2$ for the
last two rows. And it is not possible to arrange it in such a way that all conditions hold.
\end{proof}

\begin{lem} For $SL(4,\bC)$-modules, we have
\begin{equation*}\begin{split}
	  (221)\otimes(m,m,0)=&(m+2,m+2,1)\oplus(m+2,m+1,2)\oplus(m+2,m+1, { 1},1)\oplus(m+2,m,2,1)
\\&	\oplus(m+1,m+1,2,1)\oplus(m+1,m,2,2).
	 \end{split} \end{equation*}
\end{lem}
\begin{proof}
 We have to discuss possible LR skew tableaux, i.e., all fillings of the skew diagram $\nu/(2,2,1)$ satisfying
all conditions.

Hence the skew diagram should contain $2m$ boxes filled by $m$ copies of $1$ and $m$ copies of $2.$
Let us denote by $i_k$ the row $(i,\ldots,i)$ with $k$ copies of the number $i.$
The skew diagram has at most $4$ rows.

 {The fillings in the first row should be weakly increasing  and by the Y-word condition, it should end
with $1.$ Hence the first row should contain only certain numbers of copies of $1.$}
 {Let us discuss first filling with $1_m$ in the first row.} The first possibility of its filling
is indicated by the matrix {
$\begin{pmatrix}
.\ .\	1_m\\.\ .\ 2_m\\ . \;\;\;\;\;\;\,
\end{pmatrix}.
$}	
It is easy to check that it is semi-standard and satisfies the  {Y-condition}.
The corresponding module in the decomposition is $(m+2,m+2,1).$

The next possibility is to have $2_{m-1}$
in the second row. Then we have one copy of $2$ left and we have only two possibilities
how to fill them:  {
$
\begin{pmatrix}.\ .\	1_m\;\;\; \\.\ .\ 2_{m-1} \\ .\; 2\;\;\;\;\,\;\; \;\;
	\end{pmatrix};
\begin{pmatrix}.\ .\	1_m\;\;\; \\.\ .\ 2_{m-1} \\ .\;  \;\;\;\;\,\;\; \;\;\;\; \\  2\;\;\;\;\,\;\; \;\; \;\;
	\end{pmatrix}.	
 $}
They correspond to the cases $(m+2,m+1,2)$ and $(m+2,m+1,  {1},1).$

The next possibility is to have $1_m$ in the first row and $2_{m-2}$ in the second row.
Then we have two copies of $2$ left and only possibility is   {$\begin{pmatrix}.\ .\	1_m\;\;\; \\.\ .\ 2_{m-2} \\ .\; \, 2\;  \;\;\;\,\;\; \;\;  \\  2\;\;\;\;\,\;\; \;\; \;\;
	\end{pmatrix}
$, }
which corresponds to $(m+2,m,2,1).$

It is clear that it is not possible to have $1_m$ in the first row and $2_{m-3}$  or less in the second row, because there is now place left for 3 copies of $2.$

The next possibility is to have $1_{m-1}$ in the first row. The first option then is to have also
$2_{m-1}$ in the second row. Then there are not more than two cases possible:   
{$\begin{pmatrix}.\ .\	1_{m-1}\;  \\.\ .\ 2_{m-1} \;\\ .\;\,1  \;\;\;\;\,\;\; \;\;  \\  2\;\;\;\;\,\;\; \;\; \;\;
	\end{pmatrix}
$, }  
{$\begin{pmatrix}.\ .\	1_{m-1}\; \\.\ .\ 2_{m-1} \;\\ .\;\,2  \;\;\;\;\,\;\; \;\; \\  1\;\;\;\;\,\;\; \;\; \;\;
	\end{pmatrix}
$, }
 but the second one does not  {satisfy the Y-condition}. The first one corresponds to
$(m+1,m+1,2,1).$

The next possibility is to have $1_{m-1}$ in the first row and $1_{m-2}$ in the second row.
We are left with one copy of $1$ and two copies of $2.$
Then the third row has to start with $1$, because otherwise there  {would} be no chance
to place $2$'s elsewhere. And then the only option is   
{$\begin{pmatrix}.\ .\	1_{m-1}  \\.\ .\ 2_{m-2}  \\ .\; 1\;  \;\;\;\;\,  \;\; \; \\  2\; 2\;\;\,\;\; \;\; \;\;
	\end{pmatrix}
$, }
 which corresponds to $(m+1,m,2,2).$
It is easy to see that  we cannot decrease anymore the number of $1$ in the first row  {and/or} the number
of $2$ in the second row and to keep all conditions satisfied. So we have got all possibilities
in the statement of the lemma and each one with the only possibility of the filling
of the corresponding skew Young diagram. This proves the lemma.
 \end{proof}

\begin{rem} 
The case of $(210)$ is similar to the previous one.
\end{rem}

\section{Conformal invariance} \label{app_ConfInv}
	
  It is well known that the classical massless field equations on Minkowski space are conformally invariant, see \cite{PR}. In this section, we are going to discuss the question, whether the GCR equations
  for spin $\frac{3}{2}$ fields in dimension $6$
  introduced in Section \ref{ss_GCR} are conformally invariant.
  
  Let $n=p+q>2.$ There is a  classification 
  of conformally invariant
  first order differential operators on $\bR^{p,q}$ for any signature.   It can be summarized as follows.
  
  \subsection{Conformal compactification $S^{p,q}$ of $\bR^{p,q}$}\label{s_conf_cpt} 
	
  There exists a
  conformal compactification $S^{p,q}$ of $\bR^{p,q}$ with the property that
  every conformal map of an connected open subset of $\bR^{p,q}$ to $S^{p,q}$ has a  unique conformal extension to  $S^{p,q}.$ All details of the construction can be found in
  the book \cite[Chapter 2]{Sch} or in \cite[Sect. 1.6.3]{CS}. We shall describe it shortly following the second reference.
  
  Let us start with the real vector space $V=\bR^n,$ $n=p+q$ together with the standard inner product with signature $(p,q)$ given in the canonical basis $\{e_1,\ldots,e_n\}$ by the matrix
  $$
  I_{p,q}=
  \begin{pmatrix}
  	I_p&0\\0&-I_q
  \end{pmatrix}	
  $$ 
  and 
  let us consider the extended space $W=\bR^{n+2}$  with the inner product of signature
  $(p+1,q+1)$ given in the canonical basis $\{e_0,e_1,\ldots,e_n,e_{n+1}\}$ by the matrix
  \begin{equation}
  J=	J_{p+1,q+1}=
  	\begin{pmatrix}
  		0&0&1\\0&I_{p,q}&0\\ 1&0&0
  	\end{pmatrix}.
  \end{equation}	 
  Hence the vectors $e_0$ and $e_{n+1}$ are isotropic.
  \begin{dfn}
  	The conformal compactification $S^{p,q}$ of the space $V=\bR^{p,q}$ is defined
  	as the projective quadric $Q_J$ of all isotropic lines in the (real) projective
  	space $\mathbb{P}(W).$
  	If $C$ is the null cone of isotropic vectors in $W,$ then $S^{p,q}=C/{\mathbb{R}^{\times}}$
  \end{dfn}
\begin{thm}
	There exists a two-fold covering $S^p\times S^q\to S^{p,q},$ hence
	$$
	S^{p,q}\simeq (S^p\times S^q)/\mathbb{Z}_2.
	$$	
	The action of the group $O(p+1,q+1)$ on $S^{p,q}$ is transitive  and
	$S^{p,q}\simeq \GG/\PP, $ where $\PP$ is a stabilizer of a null line in $W.$
	
\end{thm}
  
  The action of the group $O(p+1,q+1)$ on $S^{p,q}$  is not effective. 
  Let us consider the group 
  $$
  \GG=O(p + 1, q + 1)/\{\pm 1\}
  $$
  and its subgroups $\PP $ fixing the isotropic line generated by $e_{n+1}$ and
  $\GG_0$ fixing both lines generated by $e_0,$ resp.\ $e_{n+1}.$
  The group $\GG_0$ is isomorphic to $CO(p,q)$ and
  $\PP$ is isomorphic to the semidirect product $\GG_0\ltimes R^{p,q}.$
   A detailed description of these facts can be found in the second chapter of the book \cite[Sect.\ 1.6.3--1.6.4]{CS}. 
  
  To include also spin representations, it is necessary
  to consider $\tilde{\GG}=Spin(p+1,q+1)/\{\pm 1\}$, the double cover  of $\GG$,
  and its parabolic subgroup $\tilde{\PP}$ fixing the isotropic line generated
  by the vector $e_{n+1}$, and to treat the conformal compactification $S^{p,q}$
  as the homogeneous space $\tilde{\GG}/\tilde{\PP},$ where $\tilde{\PP}\simeq Spin(p,q)\ltimes \bR^{p,q}.$
  Many details on such double covers can be found in the book \cite[Chapter 2]{A}.
 The homogeneous spaces $S^{p,q}=\tilde{\GG}/\tilde{\PP}$ are basic examples of
 (flat) parabolic geometries with (pseudo-)conformal structure and it is possible
 to use the classification of first order conformally invariant operators
 given in \cite{SlS} to these cases.
 
 \subsection{First order differential operators invariant for (pseudo-)Euclidean motions}

For a~given (complex and finite dimensional) irreducible representation $F$ of the group $Spin(p,q)$
 consider first order differential operators acting on the space $\mathbb{F}$ of smooth
 maps from $\bR^{p,q}$ to $F$.
Then there is a natural action of the group 
 $\tilde{\PP}=Spin(p,q)\ltimes \bR^{p,q}$ of (pseudo-)Euclidean motions on $\mathbb{F}.$
 We can classify first order  differential operators on sections
 of the vector bundle $\mathbb{F}$ invariant under the action of the group $\tilde{\PP}$ as follows.

 \begin{thm}
 	 Let $\bC^n=\bR^{p,q}\otimes \bC$ be the complexification of the vector representation of $Spin(p,q)$.
 	The tensor product  $F\otimes\bC^n$ has a multiplicity free decomposition
 	\begin{equation}\label{eq_FxCn}
	F\otimes\bC^n =\bigoplus_{j=0}^k {F}_j
		\end{equation}
 	into irreducible $Spin(p,q)$-modules. Let $J\subset\{0,\ldots,k\}$ and let $\pi_J$ be the invariant projection
 	of $F\otimes\bC^n$ onto the submodule $\bigoplus_{j\in I} {F}_j.$ 	
 	Then the first order differential operator
 	$$
 	D_J f=\pi_J(\nabla f),\ f\in\mathbb{F}
 	$$
 is an invariant first order differential operator
invariant on  $ \mathbb{F}$
where $\nabla f$ denotes the gradient of $f\in \mathbb{F}.$ 

Moreover, any invariant first order differential operator acting on sections of $ \mathbb{F}$ is given
by such a~construction.
 \end{thm}
 
  \subsection{Conformally invariant first order differential operators on $S^{p,q}$}
  Let us consider the conformal compactification $S^{p,q}$ of $\bR^{p,q}$ as the homogeneous
  space $\tilde{\GG}/\tilde{\PP}$ described in Section \ref{s_conf_cpt}.
 Let $F$ be a (complex and finite dimensional) irreducible representation of the group $Spin(p,q)$
   and let $\mathbb{F}$ be the space of smooth sections of the homogeneous bundle $\tilde{\GG}\times_{\tilde{\PP}} F.$    
  Then we can classify conformally invariant operators on $\mathbb{F}$ as follows.
  
  \begin{thm}\label{t_confInv}
  	As in \eqref{eq_FxCn} above we have a multiplicity free decomposition 
		$$F\otimes\bC^n =\bigoplus_{j=0}^k {F}_j.$$
 	
	\noindent\smallskip
  	(i) Then, for $j=0,\ldots,k$, the differential operator 
  	$$D_j\, f=\pi_j(\nabla f),\ f\in\mathbb{F},$$ 
  	is conformally invariant for a unique conformal weight $w_j$
		where $\nabla$ denotes the Levi-Civita connection on $S^{p,q}$ and 
		$\pi_j$ is the invariant projection
 	of $F\otimes\bC^n$ onto the submodule ${F}_j.$		
		In addition, the conformal weight $w_j$ is independent
  	of the signature and, for $i\not=j$, we have $w_i\not=w_j$.
  	
		\noindent\smallskip
  	(ii) Moreover, any conformally invariant differential operator on $\mathbb{F}$ is given by such a~construction.
  	
 	\end{thm}

 There is an explicit formula for the conformal weight $w_j$, which depends on the highest weights
 of the modules $F,$ resp.\ $F_j.$ 
  
  The case of positive definite signature was proved by H. Fegan in \cite{F}. 
  The classification for any signature is a very special case
  of the  classification of first order invariant operators on general parabolic geometries
  described in
  \cite[Section 4.7, Chapter 5]{SlS}.

   \subsection{Conformal invariance for generalized Cauchy-Riemann operators}

	Using Theorem \ref{t_confInv}, it is easy to study conformal invariance of generalized Cauchy-Riemann operators introduced in Section \ref{ss_GCR}.
	
   \begin{thm}\label{t_confInv_GCR}
   	
   	(i) The operators $\mathcal{D}^{300},\mathcal{D}^{333},\mathcal{D}^{100},\mathcal{D}^{111}$
   	are conformally invariant.
   	
   	(ii) The operators $\mathcal{D}^{311},\mathcal{D}^{322},\mathcal{D}^{210},\mathcal{D}^{211}$
   	are not conformally invariant.
   	\end{thm}
		
		\begin{proof}
		Indeed, by the classification recalled above, $\mathcal{D}^{300}=D^{300}_{311}$ is a~conformally invariant operator for a~proper conformal weight. 
		But the operator $\mathcal{D}^{311}=D^{311}_{322}\oplus D^{311}_{210}\oplus D^{311}_{300}$ is not conformally invariant because the conformal weights for $D^{311}_{322}$, $D^{311}_{210}$ and $D^{311}_{300}$ are pairwise different. For other operators we argue similarly.
		\end{proof}


\section*{Acknowledgements}

We would like to express our sincere gratitude to the reviewer for his/her detailed and insightful comments, which greatly contributed to the improvement of our paper. In particular, we appreciate the specific suggestions regarding Introduction, notation and Examples 2.1 and 3.1.


 \end{document}